\documentclass[12pt,letterpaper]{article}
\usepackage[utf8]{inputenc}
\usepackage[english]{babel}
\usepackage{enumerate}
\usepackage{float}
\usepackage{subcaption}
\usepackage{graphicx,psfrag,epsf}
\graphicspath{{figures/}}
\usepackage{caption}
\usepackage{setspace}
\usepackage{amsfonts}
\usepackage{longtable}
\usepackage{amsthm}
\usepackage[export]{adjustbox}
\usepackage{amsmath} 
\usepackage{amssymb}
\usepackage{mathrsfs}
\usepackage{multirow}
\usepackage{mathtools}
\usepackage{xr-hyper}
\usepackage{hyperref}

\newtheorem{assumption}{Assumption}
\newtheorem{theorem}{Theorem}
\newtheorem{lemma}{Lemma}
\newtheorem{proposition}{Proposition}
\newtheorem{corollary}{Corollary}
\makeatletter
\newcommand*{\rom}[1]{\expandafter\@slowromancap\romannumeral #1@}
\makeatother

\newcommand{\blind}{0}

\usepackage{microtype}  % avoids citations that hang into the margin
\usepackage[round]{natbib}
\usepackage[nottoc,notlot,notlof]{tocbibind}

\newcommand\Tstrut{\rule{0pt}{2.6ex}}         % = `top' strut
\newcommand\Bstrut{\rule[-0.9ex]{0pt}{0pt}}   % = `bottom' strut

\newcommand{\norm}[1]{\left\lVert#1\right\rVert}

\usepackage{tabularx}
\newcolumntype{L}[1]{>{\raggedright\arraybackslash}p{#1}}
\newcolumntype{C}[1]{>{\centering\arraybackslash}p{#1}}
\newcolumntype{R}[1]{>{\raggedleft\arraybackslash}p{#1}}

\DeclareMathOperator*{\argmax}{arg\,max}
\DeclareMathOperator*{\argmin}{arg\,min}
\DeclareMathOperator*{\plim}{plim}

\newcommand{\bzero}{\textbf{0}}
\newcommand{\bbone}{\textbf{1}}

\newcommand{\bU}{\textbf{U}}
\newcommand{\bV}{\textbf{V}}

\newcommand{\by}{\textbf{y}}

\newcommand{\wtilde}{\widetilde}
\newcommand{\capeps}{\mathcal{E}}

\begin{document}

\def\spacingset#1{\renewcommand{\baselinestretch}%
{#1}\small\normalsize} \spacingset{1}

%%%%%%%%%%%%%%%%%%%%%%%%%%%%%%%%%%%%%%%%%%%%%%%%%%%%%%%%%%%%%%%%%%%%%%%%%%%%%%

\if0\blind
{
  \title{\bf Estimation of a Structural Break Point\\ in Linear Regression Models}
  \author{Yaein Baek\thanks{
    I am very grateful for helpful comments and suggestions from Graham Elliott, James Hamilton, Brendan Beare, Kaspar Wuthrich, Yixiao Sun, Juwon Seo, Jungmo Yoon and all seminar participants in UC San Diego Department of Economics. Website: \href{http://sites.google.com/view/yaeinbaek}{http://sites.google.com/view/yaeinbaek}. Email: yibaek@phbs.pku.edu.cn.}\hspace{.2cm}\\
    Peking University HSBC Business School}
  \maketitle
} \fi

\if1\blind
{
  \bigskip
  \bigskip
  \bigskip
  \begin{center}
    {\LARGE\bf Estimation of a Structural Break Point\vspace{0.2cm} in Linear Regression Models}
\end{center}
  \medskip
} \fi

\bigskip
\begin{abstract}
This study proposes a point estimator of the break location for a one-time structural break in linear regression models. If the break magnitude is small, the least-squares estimator of the break date has two modes at the ends of the finite sample period, regardless of the true break location. To solve this problem, I suggest an alternative estimator based on a modification of the least-squares objective function. The modified objective function incorporates estimation uncertainty that varies across potential break dates. The new break point estimator is consistent and has a unimodal finite sample distribution under small break magnitudes. A limit distribution is provided under an in-fill asymptotic framework. Monte Carlo simulation results suggest that the new estimator outperforms the least-squares estimator. I apply the method to estimate the break date in U.S. real GDP growth and U.S. and UK stock return prediction models. 
\end{abstract}

\vfill

\newpage
\spacingset{1.45} % DON'T change the spacing!

% \jelcodes{C13; C32}

%%%%%%%%%%%%%%%%%%%%%%%%%%%%%%%%%%%%%%%%%%%%%%%%%%%%%%%%%%%%%
\section{Introduction}

Researchers in many economic fields extensively address parameter instability in models, which is a common empirical problem in macroeconomics and finance, such as the decrease in output growth volatility in the 1980s, known as ``the Great Moderation," oil-price shocks, labor productivity changes, inflation uncertainty, and stock-return prediction models. It is often reasonable to assume that a change occurs over a long period of time or that some historical event affects the dynamics of a structural model. Hence, the interpretation of structural model dynamics or prediction models relies heavily on the estimation and testing of parameter instability. In econometrics literature, these changes in the underlying data generating process (DGP) of time-series are referenced as a structural break. The timing of the break, as a fraction of the sample size, is called the break point.

Researchers have used estimation methods in the structural break literature to analyze threshold effects and tipping points. Studies on policy change, income inequality dynamics, and social interaction models have used structural break estimation methods. \citet{Card2008} estimate a tipping point of segregation arising in neighborhoods with white preferences. \citet*{Gonzalez2012} explore the effect of child custody law reforms and Child Support Enforcement on U.S. divorce rates using the method developed by \citeauthor*{Bai1998a} (\citeyear{Bai1998a}, \citeyear{Bai2003}).

Extensive literature describes structural break estimation methods, starting with maximum likelihood estimators (MLE) on break points. \citet*{Hinkley1970}, \citet*{Bhattacharya1987} and \citet*{Yao1987} provide an asymptotic theory of the MLE of the break point in a sequence of independent and identically distributed random variables. The asymptotic theory of least-squares (LS) estimation of a one-time break in a linear regression model has been developed by \citeauthor*{Bai1994} (\citeyear{Bai1994}, \citeyear{Bai1997}), with extension to multiple breaks in \citet*{Bai1998a} and \citet{Bai1998b}. The main problem with the LS estimation of the break point is that its finite sample behavior depends on the size of the parameter shift. In many cases, break magnitudes that are empirically relevant are ``small" in a statistical sense. For instance, the quarterly U.S. real gross domestic product (GDP) growth rate from 1970Q1 to 2018Q2 has a mean of 0.68 and a standard deviation of 0.8 percent. A break that decreases the quarterly mean growth rate by 0.25 percentage points is less than half a standard deviation change but is equivalent to a 1 percentage-point decrease in annual growth, which is a significant event for the economy.

In asymptotic analysis, tests have local power against breaks with a magnitude of order $O(T^{-1/2})$\footnote{Asymptotic analysis under a DGP with a drifting sequence of parameters can be considered as a form of weak identification asymptotics. Literature on estimation and inference with a restricted parameter space under weak identification includes \citeauthor{Andrews2012} (\citeyear{Andrews2012}, \citeyear{Andrews2013}, \citeyear{Andrews2014}), and \citet*{Han2019}. Their results do not cover an abrupt structural change model, which is our model of interest.}. The magnitude represents a small break, shrinking with sample size $T$, so that structural breaks tests have asymptotic power strictly less than one (\citeauthor{Elliott2007} \citeyear{Elliott2007}). In the presence of small but detectable breaks, the LS estimator of the break point has a finite sample distribution that exhibits tri-modality with one mode at the true value and two modes at zero and one. Break points at zero or one do not provide any information about a structural break, nor are they likely to be true in practice. Therefore, inference in practical applications based on LS estimation of structural breaks would seem unreliable. Surprisingly, although the methodology is used widely, there are few alternatives for estimating the location of a structural break. Recent literature such as \citeauthor{Casini2019} (\citeyear{Casini2019}, \citeyear{Casini2020}) suggest a Laplace-based procedure to provide an estimator of the break point, which is defined by an integration, rather than an optimization-based method. 

This study provides an estimator of the structural break point, which is a generalization of LS estimation, and hence, easy to implement in practice. The new estimator resolves the finite sample issue of LS estimation; it has a finite sample distribution with a unique mode at the true break and flat tails. This is achieved by imposing weights on the LS objective function. Under small breaks, the LS estimator picks boundaries with high probability due to the functional form of the objective. I construct a weight function of the break point and impose it on the LS objective function to incorporate different estimation uncertainties across potential break points. I provide conditions on the weight function that ensure consistency of the break point estimator. I also suggest a representative weight function for empirical researchers to use.

The break point estimator is consistent with the same rate of convergence as the LS estimator (\citeauthor{Bai1997} \citeyear{Bai1997}) under regularity conditions on the weight functional form in a linear regression model with a structural break on a subset (or all) coefficients. The limit distribution of the break point estimator is derived when the break magnitude is small, under an in-fill asymptotic framework, following the approach of \citeauthor{Jiang2017} (\citeyear{Jiang2017}, \citeyear{Jiang2018}). Monte Carlo simulations show that the break point estimator has smaller root mean squared error (RMSE) than the LS estimator in a finite sample for all break point values considered.

This study provides two empirical applications: estimation of structural breaks on post-war U.S. real GDP growth rate and the U.S. and UK stock return prediction models. For the quarterly U.S. real GDP growth rate under different sample periods, the new method estimates a break in the early 1970s, whereas the LS estimates vary from the 1970s to 1952 or 2000, which are near boundaries of the sample. The break date estimate in the early 1970s is matched with the ``productivity growth slowdown" suggested in literature, such as \citet*{Perron1989} and \citet*{Hansen2001}. Thus, my estimation method yields reasonable break point estimates compared to LS estimates, which is sensitive to trimming the sample.

The remainder of this paper proceeds as follows: Section \ref{sec:motivation} constructs the break point estimator for a mean shift in a linear process. Section \ref{sec:multimodel} provides a generalized linear regression model with multiple regressors, and proves the consistency of the break point estimator. Section \ref{sec:limitdist} presents the in-fill asymptotic theory for stationary and local-to-unit root processes. Monte Carlo simulation results are in Section \ref{sec:montecarlo} and Section \ref{sec:empirical} provides three empirical applications of the new structural break estimation method. We provide concluding remarks in Section \ref{sec:conclusion}. Additional theoretical results and proofs are in the Appendix.

%%%%%%%%%%%%%%%%%%%%%%%%%%%%%%%%%%%%%%%%%%%%%%%%%%%%%%%%%%%%%%
\section{Structural Break Point Estimator}\label{sec:motivation}

In this section, I consider the simplest regression model with a constant term to provide an intuitive explanation of the construction of the break point estimator. I provide theoretical results in Section \ref{sec:multimodel} under a general linear regression model with multiple regressors. Suppose a single break occurs at time $k_0 = [\rho_0 T]$, where $\rho_0 \in (0,1)$, $[\cdot]$ is the greatest smaller integer function, and $\bbone\{t>k_0\}$ is an indicator function that equals one if $t > k_0$ and zero otherwise.
\begin{equation}\label{eq:DGP}
    y_{t} = \mu +  \delta\bbone\{t > k_0\} + \varepsilon_{t}, \;\;\; t=1,\ldots,T
\end{equation}
The disturbances $\{\varepsilon_{t}\}$ are independent and identically distributed (i.i.d.) with mean zero and $E\varepsilon_t^2 = \sigma^2$. The pre-break mean $y_t$ is $\mu$ and the post-break mean is $\mu + \delta$. Assume we know a one-time break occurs, but the break point $\rho_0$ and parameters $(\mu, \delta, \sigma^2)$ are unknown.

The conventional estimation method of the break location in the literature is least-squares. One obtains the LS estimator by finding a value $k$ that minimizes the objective function $S_{T}(k)^2$, which is the sum of squared residuals (SSR) under the assumption that $k$ is the break date, $S_{T}(k)^2 =\sum_{t=1}^{k}(y_t - \bar{y}_{k})^2 + \sum_{t=k+1}^{T}(y_t - \bar{y}_{k}^{*})^2$, where $\bar{y}_{k} = k^{-1}\sum_{j=1}^{k}y_j$ and $\bar{y}_{k}^{*} = (T-k)^{-1}\sum_{j=k+1}^{T}y_j$ are pre- and post-break LS estimates under break date $k$, respectively. Following \citeauthor{Bai1994}'s \citeyearpar{Bai1994} expression, I use the identity $\sum_{t=1}^{T}(y_t - \bar{y})^2 = S_{T}(k)^2 + TV_{T}(k)^{2}$ (\citeauthor{Amemiya1985} \citeyear{Amemiya1985}), where $V_{T}(k)^2  = k/T(1-k/T)\left(\bar{y}_{k}^{*} - \bar{y}_{k}\right)^2$, to substitute for the SSR. Then the LS estimator of the break date is equivalent to
\begin{equation}\label{Def:hatrhols_breakmean}
    \hat{k}_{LS} = \argmax_{k=1,\ldots,T-1} \left\vert V_{T}(k) \right\vert, \;\;\;\; \hat{\rho}_{LS} = \hat{k}_{LS}/T.
\end{equation}
Denote $\rho = k/T$ and $\rho_0 = k_0/T$. An issue with the LS estimator $\hat{\rho}_{LS}$ is that under a small magnitude $\vert \delta \vert$, $\hat{\rho}_{LS}$ has a finite distribution that is tri-modal with two modes at the ends of the unit interval and one mode at the true break point $\rho_0$.

A break magnitude that is statistically small is not necessarily small in an economic sense. For example, quarterly U.S. real GDP growth rate from 1970Q1 to 2018Q2 has a mean of 0.68 and a standard deviation of around 0.8 percent. A break that decreases the mean quarterly growth rate by 0.3 percentage points (a 1.2 percentage point decrease in yearly growth) is a significant event for the economy. Suppose model (\ref{eq:DGP}) has parameter values similar to the U.S. real GDP growth rate; assume $\rho_0 = 0.3$, the pre-break mean is $\mu = 0.88$ percent and the shift in the mean of growth rate is $\delta = -0.29$. The expectation of $y_t$ is $\mu+(1-\rho_0)\delta = 0.68$, which matches the quarterly U.S. real GDP growth rate. Suppose we have $T=100$ observations and Gaussian disturbances $\varepsilon_t \overset{i.i.d.}{\sim} N(0,0.8^2)$. The left plot of Figure \ref{fig:example_gdp_rho} shows the finite sample distribution of the LS estimator of $\rho$ from a Monte Carlo simulation with 2,000 replications. The LS estimator fails to accurately detect the break that occurs in the constant term of a univariate linear regression model. Thus, we expect that in practice, structural breaks that are economically important are not large enough for the LS estimator to detect in many cases.

This study focuses on such empirically relevant breaks that are not ``large" enough. I follow the approach of \citet*{Elliott2007} to provide an asymptotic approximation to finite sample properties under this small break magnitude. The break magnitude has the same order as sampling uncertainty, $\delta = T^{-1/2}d$, where $d$ is fixed. These asymptotics reflect an important feature of finite sample properties under moderate breaks, because the $p$ values of tests for breaks are typically significant, but not zero. See \citeauthor{Elliott2007} (\citeyear{Elliott2007}, \citeyear{Elliott2014}) for details on the justification of this break magnitude.

\begin{figure}[H]
    \centering
    \begin{subfigure}{.48\textwidth}
  \centering
    \includegraphics[width=\linewidth]{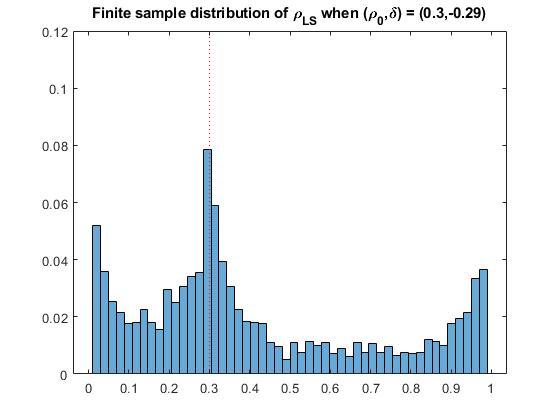}
\end{subfigure}
\begin{subfigure}{.48\textwidth}
  \centering
    \includegraphics[width=\linewidth]{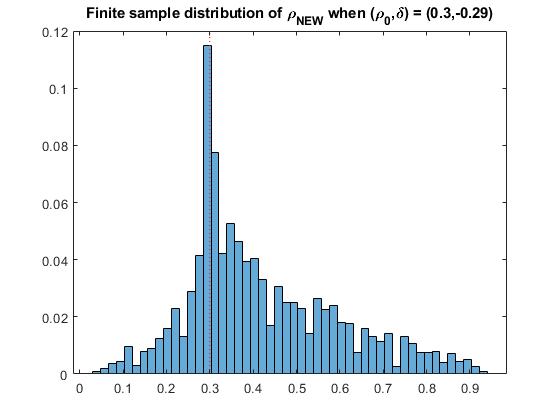}
\end{subfigure}
    \caption{Finite sample distribution of the LS estimator (left) and new estimator (right) of break point $\rho$ with weight function $w(\rho) = (\rho(1-\rho))^{1/2}$ when $(\rho_0,\delta) = (0.3, -0.29)$, $T = 100$, and $\varepsilon_t \overset{i.i.d.}{\sim} N(0,0.8^2)$ with 2,000 replications.}
\label{fig:example_gdp_rho}
\end{figure}

Importantly, in literature, it is standard to trim the boundaries of the optimization space so that $\hat{k}_{LS}$ in (\ref{Def:hatrhols_breakmean}) is the argmax function across $k = [\alpha T], \ldots, [(1-\alpha) T]$ for some $0 < \alpha < 1/2$. Trimming the optimization space may help reduce the build-up mass at the boundaries of the finite sample distribution; however, this has its own drawbacks. Figure \ref{fig:optspace_rho50} shows the finite sample distribution of the LS estimator of $\rho$ with various fractions of trimming, $\alpha \in \{0, 0.1, 0.15, 0.2\}$, from a Monte Carlo simulation with 2,000 replications. Under small break magnitudes $\delta = T^{-1/2}d$, $d \in \{2, 4\}$, the modes at the boundaries remain even after trimming. It is unclear if there is a trade-off between the break size and how large a trimming is needed. With larger trimming ($\alpha = 0.2$), the mass at the boundaries accumulate even more. In addition, there is no reason to believe a break occurs in a restricted period. Thus, we need an alternative method to resolve this issue.

\begin{figure}[H]
    \centering
    \begin{subfigure}{.48\textwidth}
  \centering
    \includegraphics[width=\linewidth]{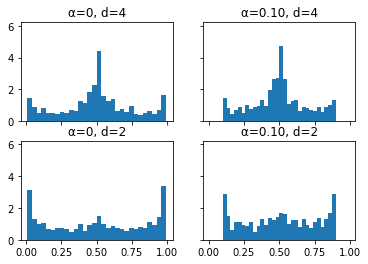}
\end{subfigure}
\begin{subfigure}{.48\textwidth}
  \centering
    \includegraphics[width=\linewidth]{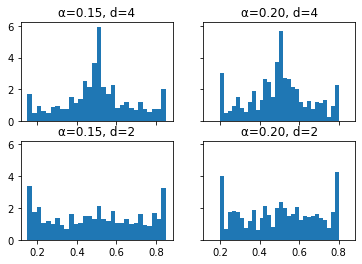}
\end{subfigure}
    \caption{Finite sample distribution of $\hat{\rho}_{LS}$ when $(\rho_0,\delta) = (0.5, 4T^{-1/2})$ (upper plots), $(\rho_0,\delta) = (0.5, 2T^{-1/2})$ (lower plots) $T = 100$, and $\varepsilon_t \overset{i.i.d}{\sim} N(0, 1)$ with 2,000 replications. Each distribution has different trimmed optimization space $[\alpha, 1-\alpha]$, with $\alpha \in \{0, 0.1, 0.15, 0.2\}$.}
    \label{fig:optspace_rho50}
\end{figure}

The finite sample distribution of the LS estimator has build-up mass at boundaries under small break magnitudes because of how the objective function is constructed. For each potential break date $k$, the objective function is constructed by partitioning the sample into two sub-samples, before and after $k$. Each sub-sample is used to estimate two different means, $\bar{y}_k$ and $\bar{y}_k^*$. If $k$ is near 1, the pre-break sub-sample size $k$ is small; similarly, if $k$ is near $T-1$, the post-break sub-sample size $T-k$ is small\footnote{Estimation theory does not require $\rho$ to be bounded away from zero and one, provided that a change point is assumed to exist. However, identification in a finite sample typically needs more than one observation pre- and post-break. In Section \ref{sec:multimodel}, I assume the break point exists; see Assumption \ref{A:multimodel}$(i)$.}. Hence, when the potential break date $k$ of $\left\vert V_{T}(k) \right\vert$ is near the boundaries, the estimates of pre- or post-break mean are imprecise because of the small sub-sample size. Estimation uncertainty at boundaries distorts picking up the true break location if the break magnitude is small relative to sampling variability. 

Because the issue arises from the large variance of the objective function at the boundaries, one can think of shrinking the variance accordingly. Suppose there are non-negative ``weights" $\omega_k$ imposed on the LS objective function so that $k$ with a large estimation error has smaller weights than $k$ with a small estimation error. When $k=1$ and $T-1$, weights near zero are imposed, which implies the variance of the weighted objective function $\omega_k \vert V_T(k)\vert$ would shrink toward zero. If the sample period is normalized into a unit interval, the weights are represented by a continuous function $\omega(\rho)$ on $\rho \in [0, 1]$, which is zero at $\rho \in \{0,1\}$, and has positive values otherwise. A continuous function with such properties would look like an inverse U-shaped (or concave downward) function on the unit interval.

I propose a new break point estimator to maximize the value of the objective function $\left\vert Q_{T}(k) \right\vert$, equal to weights $\omega_k$ multiplied by the LS objective $\left\vert V_{T}(k) \right\vert$.
\begin{gather}\label{Def:newestimator_breakmean}
    \hat{k} = \argmax_{k =1,\ldots,T-1} \left\vert Q_{T}(k) \right\vert, \;\;\;\; \hat{\rho} = \hat{k}/T \\
\nonumber
    \left\vert Q_{T}(k)\right\vert := \omega_k \left\vert V_{T}(k) \right\vert= \omega_k \left(\frac{k(T-k)}{T^2}\right)^{1/2}\left\vert\bar{y}_{k}^{*} - \bar{y}_{k}\right\vert. 
\end{gather}
The weight function shrinks the variance of the LS objective function when $k$ is near the boundaries, and thus, the maximizing value $\hat{k}$ is less likely to pick either end. The right plot of Figure \ref{fig:example_gdp_rho} shows the finite sample distribution of the break point estimator (\ref{Def:newestimator_breakmean}) under the DGP calibrated from the U.S. real GDP growth rate. As expected, the break point estimator has flat tails at boundaries with a mode at the true break point $\rho_0 = 0.3$, whereas the LS estimator has modes at $0.01$ and $0.99$. Section \ref{sec:montecarlo} provides additional Monte Carlo simulations comparing the two estimators.

The break point estimator (\ref{Def:newestimator_breakmean}) is easy to implement, as we simply modify the objective function by multiplying the weight function. It is a generalization of LS estimation because the LS estimator is a special case, when $\omega(\rho) = 1$. Moreover, by employing the weight function, we no longer need to trim the search grid, because the boundaries have zero weight. Section \ref{sec:multimodel} provides a set of conditions on the weight function that ensures consistency of the break point estimator under a general linear regression model. 

I also suggest a representative weight function $\omega(\rho) = (\rho(1-\rho))^{1/2}$ under model (\ref{eq:DGP}) (see Section \ref{sec:multimodel} for its analogue under a model with multiple regressors). The weight function has two interpretations. First, it is related to the weighting function from \citet*{Anderson1954}, which tests whether the sample is drawn from a particular distribution. A non-negative weight function is chosen to accentuate the boundaries of the sample space, where the test is desired to have sensitivity. If the cumulative distribution function (cdf) under the null hypothesis is $F(\cdot)$, the weight function is $\left[F(x)(1-F(x))\right]^{-1}$, which increases as $x$ approaches the boundaries of the sample space. In contrast, I want to down-weight the boundaries of the parameter space $\rho \in [0, 1]$. If $\rho$ is a random variable with cdf $F(\rho)$, the weight would be the reciprocal of the weight function in \citet*{Anderson1954}, $F(\rho)(1-F(\rho))$. We further assume that $\rho$ is uniformly distributed on the unit interval, $F(\rho) = \rho$. We obtain the weight function $\rho(1-\rho)$, which down-weights the variance of the LS objective function $V_T(k)^2$ near the boundaries.

Second, from a Bayesian perspective, the weight function $\omega(\rho)$ can be interpreted as a prior on parameters $\delta$ and $\rho$. We assume the Gaussian disturbances in (\ref{eq:DGP}), $\omega(\rho) = (\rho(1-\rho))^{1/2}$ are equivalent to the square root of the Fisher information up to a constant. The Fisher information is interpreted as a way to measure the amount of information the data gives us about the unknown parameter $\delta$, given $\rho$. A prior distribution based on the Fisher information reflects our belief that a structural break is less likely to occur near the boundaries. See Appendix A for details.

%%%%%%%%%%%%%%%%%%%%%%%%%%%%%%%%%%%%%%
\section{Partial Break with Multiple Regressors}\label{sec:multimodel}

This section provides consistency of the break point estimator under a general linear regression model with multiple regressors. The model incorporates a partial break in coefficients and assumes that a one-time break occurs at an unknown date $k_0 = [\rho_0 T]$ with $\rho_0 \in (0,1)$. I follow the notations of \citet*{Bai1997} by denoting the vector of variables associated with a stable coefficient as $w_t$ and the variables associated with coefficients under a break as $z_t$. Let $x_{t} = (w_t', z_t')'$ be a $(p \times 1)$ vector and $z_{t}$ is a $(q\times 1)$ vector with $q \leq p$,
\begin{equation}\label{eq:DGPmulti}
    y_{t} = \begin{cases}
        x_{t}'\beta + \varepsilon_{t} & \text{if } t=1,\ldots,k_0 \\
        x_{t}'\beta + z_{t}'\delta_T + \varepsilon_{t} & \text{if } t=k_0+1,\ldots, T, 
    \end{cases}
\end{equation}
where $\varepsilon_t$ is a mean zero error term. In general, $z_{t}$ can be expressed as a linear function of $x_{t}$ so that $z_{t} = R'x_{t}$, where $R$ is a $(p\times q)$ matrix with full column rank. Let $Y = (y_1,\ldots,y_T)'$ and define $X_{k} := (0,\ldots,0,x_{k+1},\ldots,x_{T})'$ and $X_{0} :=(0,\ldots,0,x_{k_0 +1},$ $\ldots,x_{T})'$. Define $Z_k$ and $Z_0$ analogously so that $Z_k = X_k R$ and $Z_0 = X_0 R$. Let $M := I-X(X'X)^{-1}X'$ and use the maximal invariant to eliminate the nuisance parameter $\beta$. 

The subscript on $\delta_T$ shows that the break magnitude may depend on the sample size. We assume the break magnitude is outside the local $T^{-1/2}$ neighborhood of zero. This is because the break point is not consistently estimable if the break magnitude is in the local $T^{-1/2}$ neighborhood of zero. This corresponds to the case of small break magnitudes discussed previously, $\delta_T = O(T^{-1/2})$, in which structural break tests have asymptotic power strictly less than one. In this section, I proceed by assuming that $\delta_T$ is fixed or it converges to zero at a rate \emph{slower} than $T^{-1/2}$ so that the power of the structural break tests converge to one (Assumption \ref{A:consistency_delta}).

Let $\bar{S} =Y'MY$, and denote $S_{T}(k)^2$ as the SSR regressing $MY$ on $MZ_k$. The LS estimator of break date $\hat{k}_{LS}$ is the value that minimizes $S_{T}(k)^2$, and thus, maximizes $V_{T}(k)^2$ from the identity $\bar{S} = S_{T}(k)^2 + V_{T}(k)^2$ (\citeauthor{Amemiya1985} \citeyear{Amemiya1985}),
\begin{gather*}
    \hat{k}_{LS} = \argmax_{k} V_{T}(k)^2, \;\;\;\; \hat{\rho}_{LS} = \hat{k}_{LS}/T \\
    V_{T}(k)^2 := \hat{\delta}_{k}'(Z_{k}'MZ_{k})\hat{\delta}_{k},
\end{gather*}
where $\hat{\delta}_{k}$ is the LS estimate of $\delta_T$ by regressing $MY$ on $MZ_k$. Note that $V_{T}(k)^2$ is non-negative from the inner product of the vector $(Z_{k}'MZ_{k})^{1/2}\hat{\delta}_{k}$. The LS objective function is modified by multiplying a $(q \times q)$ positive definite weight matrix $\Omega_k$, which is a generalization of $\omega_k$ in Section \ref{sec:motivation} for linear regression models with multiple regressors. Decompose the weight matrix so that $\Omega_k = \Omega_k^{1/2\prime}\Omega_k^{1/2}$ and multiply $\Omega_k^{1/2}$ to the vector $(Z_{k}'MZ_{k})^{1/2}\hat{\delta}_{k}$. Take the inner product and obtain the objective function $Q_T(k)^2 :=  \hat{\delta}_{k}'(Z_{k}'MZ_{k})^{1/2}\Omega_k (Z_{k}'MZ_{k})^{1/2}\hat{\delta}_{k}$. Then the estimator of the break point is 
\begin{equation}\label{eq:newhatk_multi}
    \hat{k} = \argmax_{k} Q_T(k)^2, \;\;\;\; \hat{\rho} = \hat{k}/T.
\end{equation}
An example of the weight matrix is $\Omega_k = T^{-1}Z_k'MZ_k$, which is equal to the square of the representative weight function $\omega_k^2 = k/T(1-k/T)$ in model (\ref{eq:DGP}) if $R=I$ and $X$ is a $(T\times 1)$ vector of ones. Similarly, the matrix $T^{-1}Z_k'MZ_k$ ``decreases" as $k$ approaches either end of the sample from the following rearrangement of terms:
\begin{align}
\nonumber
    T^{-1}Z_k'MZ_k & = T^{-1}[Z_{k}'Z_{k} - Z_{k}'X(X'X)^{-1}X'Z_{k}] \\
\label{eq:zkmzk}
    & = T^{-1}R'(X_{k}'X_{k})(X'X)^{-1}(X'X-X_{k}'X_{k})R.
\end{align}

I prove the consistency of the break point estimator $\hat{\rho}$ in (\ref{eq:newhatk_multi}) under regularity conditions on model (\ref{eq:DGPmulti}) and weight matrix $\Omega_k$. The notation $\norm{\cdot}$ denotes the Euclidean norm $\norm{x} = \left(\sum_{i=1}^{p}x_i^2\right)^{1/2}$ for $x \in \mathbb{R}^p$. For a matrix $A$, $\norm{A}$ represents the vector induced norm $\norm{A} = \sup_x \norm{Ax}/\norm{x}$ for $x \in \mathbb{R}^p$ and $A \in \mathbb{R}^{p \times p}$.

\begin{assumption}\label{A:multimodel}
\begin{enumerate}[(i)]
    \item $k_0 = [\rho_0 T]$, where $\rho_0 \in [\alpha, 1-\alpha]$, $0 < \alpha < \frac{1}{2}$;
    \item The data $\{y_{tT}, x_{tT}, z_{tT}: 1\leq t \leq T, T \geq 1\}$ form a triangular array. The subscript $T$ is omitted for simplicity. In addition, $z_t = R'x_t$, where $R$ is $p\times q$, rank($R$) = $q$, $z_t \in \mathbb{R}^{q}$, $x_t \in \mathbb{R}^{p}$, and $q \leq p$;
    \item The matrices $\left(j^{-1}\sum_{t=1}^{j}x_t x_t'\right)$, $\left(j^{-1}\sum_{t=T-j+1}^{T}x_t x_t'\right)$, $\left(j^{-1}\sum_{t=k_0-j+1}^{k_0}x_t x_t'\right)$ and\\ $\left(j^{-1}\sum_{t=k_0+1}^{k_0+j}x_t x_t'\right)$ have minimum eigenvalues bounded away from zero in probability for all large $j$. For simplicity, assume these matrices are invertible when $j \geq p$. In addition, these four matrices have stochastically bounded norms uniformly in $j$. That is, for example, $\sup_{j \geq 1} \norm{ j^{-1} \sum_{t=1}^j x_t x_t'}$ is stochastically bounded;
    \item $T^{-1}\sum_{t=1}^{[sT]}x_t x_t' \overset{p}{\rightarrow} s\Sigma_x$ uniformly in $s \in [0,1]$, where $\Sigma_x$ is a nonrandom positive definite matrix;
    \item For random regressors, $\sup_t E \norm{x_t}^{4+\gamma} \leq K$ for some $\gamma >0 $ and $K < \infty$;
    \item The disturbance $\varepsilon_t$ is independent of the regressor $x_s$ for all $t$ and $s$. For an increasing sequence of $\sigma$-fields $\mathcal{F}_t$, $\{\varepsilon_t, \mathcal{F}_t\}$ form a sequence of $L^r$-mixingale sequence with $r=4+\gamma$ for some $\gamma > 0$ (\citet*{McLeish1975} and \citet*{Andrews1988}). That is, there exists nonnegative constants $\{c_t: t\geq 1\}$ and $\{\psi_j: j\geq 0\}$, such that $\psi_j \downarrow 0$ as $j \rightarrow \infty$ and for all $t \geq 1$ and $j \geq 0$, we have (a) $E\left\vert E(\varepsilon_t \vert \mathcal{F}_{t-j}) \right\vert^{r} \leq c_t^r \psi_{j}^r$, (b) $E\left\vert \varepsilon_t - E(\varepsilon_t\vert\mathcal{F}_{t+j})\right\vert^{r} \leq c_t^r \psi_{j+1}^r$, (c) $\max_j \vert c_j \vert < K < \infty$, (d) $\sum_j j^{1+\kappa}\psi_j < \infty$ for some $\kappa >0$.
\end{enumerate}
\end{assumption}

\begin{assumption}\label{A:Omegak} 
$\Omega_k$ is a positive definite $(q\times q)$ matrix ($q =$dim($z_t$)) that is a continuous function of data $\{y_t, x_t, z_t; 1\leq t \leq T\}$ and have stochastically bounded norms uniformly in $k=1,\ldots,T-1$. In addition, for any nonzero vector $c \in \mathbb{R}^q$,
\begin{equation*}
    \norm{\Omega_{k_0}^{1/2}(Z_0'MZ_0)^{1/2}c} > \norm{\Omega_k^{1/2}(Z_k'MZ_k)^{-1/2}(Z_k'MZ_0)c}
\end{equation*}
holds for all $k$ and $k_0$, where $M = I-X(X'X)^{-1}X'$. When $k/T \rightarrow \rho$ as $T \rightarrow \infty$, then $\Omega_k \overset{p}{\rightarrow} \bar{\Omega}(\rho)$, where $\bar{\Omega}(\rho)$ is a differentiable function of $\rho$, element-wise.
\end{assumption}

The conditions of assumption \ref{A:multimodel} are similar to assumptions A1 to A6 in \citet*{Bai1997}, with additional restrictions $(iv)$ and $(vi)$. Assumption \ref{A:multimodel}$(vi)$ allows for general serial correlation in disturbances and requires $x_t$ to be strictly exogeneous. This is because $\Omega_k$ depends on the moments of regressors and we want to impose zero weights on the boundaries of the unit interval. For instance, if the second moments of $z_t$ changes at $\rho_0$, the boundaries of the unit interval may have positive weights that depend on the distribution of $z_t$. These cases are avoided under strict exogeneity because $\Omega_k$ converges in probability to a nonrandom matrix that varies across $\rho$ only. Note that if $\Omega_k$ is a non-stochastic matrix that satisfies the norm inequality in Assumption \ref{A:Omegak}, consistency holds under weakly exogeneous regressors (see Assumption \ref{A:weakexog}).

Assumption \ref{A:Omegak} guarantees that the matrix
\begin{align}
\nonumber
  A_{T}(k) & :=\frac{1}{\vert k_0 - k\vert}\left[  (Z_0'MZ_0)^{1/2}\Omega_{k_0}(Z_0'MZ_0)^{1/2} \right. \\
\label{eq:Ak}
  & \hspace{0.7cm} \left. -(Z_0'MZ_k)(Z_k'MZ_k)^{-1/2}\Omega_k(Z_k'MZ_k)^{-1/2}(Z_k'MZ_0)\right] 
\end{align}
is positive definite, and hence, $\norm{A_{T}(k)} \geq \lambda_{\min}(A_{T}(k)) > 0$, where $\lambda_{\min}$ denotes the minimum eigenvalue of $A_{T}(k)$. The condition can be interpreted as follows: for simplicity, consider the univariate model (\ref{eq:DGP}). Assumption \ref{A:Omegak} is equivalent to $\vert \omega'(\rho)/ \omega(\rho) \vert < (2\rho(1-\rho))^{-1}$ for all $\rho$, where $\omega'(\rho) = \partial \omega(x)/\partial x \vert_{x=\rho}$. The slope magnitude of the logarithm of $\omega(\rho)$ has an upper bound that increases as $\rho$ approaches zero or one. A sufficient condition is the function $\omega(\rho) = (\rho(1-\rho))^{\gamma}$, with $-1/2 \leq \gamma \leq 1/2$.\footnote{The function $\omega(\rho) = (\rho(1-\rho))^{\gamma}$, with $-1/2 \leq \gamma \leq 1/2$, allows a convex function that has a large weight on the boundaries. This is because I assume ``large" break magnitudes (Assumption \ref{A:consistency_delta}) for consistency of the estimator. That is, if the break magnitude is large enough, we no longer have build-up mass at the boundaries and imposing large weights does not matter for consistency.} Note that this is sufficient under Assumption \ref{A:multimodel}$(i)$. If $\alpha$ is arbitrarily close to zero, it may restrict the functional of $\omega(\cdot)$. The weight matrix $\Omega_k$ may be close to a singular matrix in a finite sample if $\alpha$ is extremely close to zero under model (\ref{eq:DGPmulti}). Under Assumption \ref{A:multimodel}$(iv)$, the weight matrix converges in probability to a function of $\rho$ and $\Sigma_x$ as $T$ increases. Because $\bar{\Omega}(\rho)$ is a differentiable function of $\rho$ element-wise, $\norm{\Omega_k - \Omega_{k_0}} \leq b\vert k-k_0 \vert/T$ for some finite $b > 0$ and all $k$.

\begin{assumption}\label{A:consistency_delta}
$\delta_T \rightarrow 0$ and $T^{1/2-\gamma}\delta_T \rightarrow \infty$ for some $\gamma \in \left(0,\frac{1}{2}\right)$.
\end{assumption}

The consistency of the break point estimator is proved by showing that if $\delta_T \neq 0$, then with high probability, $Q_T(k)^2$ can only be maximized near the true break $k_0$. The objective function $Q_T(k)^2$ is defined in (\ref{eq:newhatk_multi}) and $\hat{\delta}_k$ is the LS estimator of the break magnitude, assuming that $k$ is the break date: $\hat{\delta}_k = (Z_k'MZ_k)^{-1}(Z_k'MZ_0)\delta_T + (Z_k'MZ_k)^{-1}Z_k'M\varepsilon$. If $k = k_0$, then $\hat{\delta}_{k_0} = \delta_T + (Z_0'MZ_0)^{-1}Z_0'M\varepsilon$.

\begin{theorem}\label{Thm:rateofconvergence}
Under Assumptions \ref{A:multimodel} and \ref{A:Omegak}, suppose $\delta_T$ is fixed or shrinking toward zero such that Assumption \ref{A:consistency_delta} is satisfied. Then, $\hat{k} = k_0 + O_{p}(\norm{\delta_T}^{-2})$ and the break point estimator $\hat{\rho}$ in (\ref{eq:newhatk_multi}) is consistent.
\begin{equation*}
    \vert\hat{\rho}-\rho_0 \vert = O_{p}(T^{-1}\norm{\delta_T}^{-2})  = o_p(1).
\end{equation*}
\end{theorem}

See Appendix B for proof of Theorem \ref{Thm:rateofconvergence}. For weakly exogenous regressors, the break point estimator is consistent with the same rate of convergence in Theorem \ref{Thm:rateofconvergence}, under the following conditions that substitute Assumptions \ref{A:multimodel} and \ref{A:Omegak}.

\begin{assumption}\label{A:weakexog}
Assume the following conditions in model (\ref{eq:DGPmulti}) with Assumption \ref{A:multimodel}$(i)$-$(iii)$ and $(v)$.  
\begin{enumerate}[(i)]
    \item $(X'X)/T$ converges in probability to a nonrandom positive definite matrix, as $T \rightarrow \infty$;
    \item $\{\varepsilon_t, \mathcal{F}_t\}$ form a sequence of martingale differences for $\mathcal{F}_t = \sigma$-field $\{\varepsilon_s,x_{s+1}: s\leq t\}$. Moreover, for all $t$, $E\vert\varepsilon_t\vert^{4+\gamma} < K$ for some $K < \infty$ and $\gamma > 0$;
    \item The weight matrix $\Omega_k$ is a nonrandom $(q\times q)$ positive definite matrix, and for any nonzero vector $c \in \mathbb{R}^q$,
    \begin{equation*}
        \norm{\Omega_{k_0}^{1/2}(Z_0'MZ_0)^{1/2}c} > \norm{\Omega_k^{1/2}(Z_k'MZ_k)^{-1/2}(Z_k'MZ_0)c}
    \end{equation*}
    holds for all $k$ and $k_0$, where $M = I-X(X'X)^{-1}X'$. $\Omega_k$ converges to $\bar{\Omega}(\rho)$ as $k/T \rightarrow \infty$, which is a differentiable function of $\rho$ on the unit interval.
    \end{enumerate}
\end{assumption}

\begin{theorem}\label{Thm:weakexog_rateofconv}
Under Assumption \ref{A:weakexog}, suppose $\delta_T$ is fixed or shrinking toward zero that satisfies $\delta_T \rightarrow 0$ and $T^{1/2-\gamma}\delta_T \rightarrow \infty$ for some $\gamma \in (0,\frac{1}{2})$. Then, $\hat{k} = k_0 + O_{p}(\norm{\delta_T}^{-2})$ and the break point estimator $\hat{\rho}$ in (\ref{eq:newhatk_multi}) is consistent.
\begin{equation*}
    \vert\hat{\rho}-\rho_0 \vert = O_{p}(T^{-1}\norm{\delta_T}^{-2})  = o_p(1).
\end{equation*}
\end{theorem}

The proof of Theorem \ref{Thm:weakexog_rateofconv} is similar to the proof of Theorem \ref{Thm:rateofconvergence}; hence we have omitted it. Under Assumptions \ref{A:multimodel}$(v)$, \ref{Thm:weakexog_rateofconv}$(i)$, and \ref{Thm:weakexog_rateofconv}$(ii)$, the strong law of large numbers holds for $x_t \varepsilon_t$, because the conditions in \citet*{Hansen1991} are satisfied. The weight matrix $\Omega_k$ in Assumption \ref{A:weakexog}$(iii)$ depends on $k/T$ but not on the data $\{x_t,\varepsilon_t\}$. Thus, by setting $\rho = k/T$, $\Omega_k$ is a function of $\rho$, which is assumed to be differentiable with respect to $\rho$. Then, for some finite $c>0$, the bound $\norm{\Omega_{k_1} - \Omega_{k_2}} \leq c \vert k_1-k_2\vert/T$ holds for any $k_1$ and $k_2$. Using these properties, proving the consistency of the estimator under Assumption \ref{A:weakexog} follows the same process as in the proof under Assumptions \ref{A:multimodel} and \ref{A:Omegak}.

Given the consistency of the break point estimator from Theorem \ref{Thm:rateofconvergence} or \ref{Thm:weakexog_rateofconv}, the estimator of the break magnitude corresponding to $\hat{k}$ is consistent and asymptotically normally distributed. Let $\hat{\delta}(\hat{\rho}) = \hat{\delta}_{\hat{k}}$, then the following results hold. The proof is provided in the Appendix.

\begin{corollary}\label{Cor:delta_asym}
Under Assumptions \ref{A:multimodel} and \ref{A:Omegak}, suppose $\delta_T$ is fixed or shrinking toward zero such that Assumption \ref{A:consistency_delta} is satisfied. Let $\hat{\delta}(\hat{\rho})$ be a consistent estimator of $\delta_T$ corresponding to $\hat{k}$, which is defined in (\ref{eq:newhatk_multi}). Then,
\begin{equation*}
    \sqrt{T}\left(\hat{\delta}(\hat{\rho})-\delta_T\right) \overset{d}{\longrightarrow} N\left(0, \bV^{-1}\bU\bV^{-1}\right),
\end{equation*}
where
\begin{equation*}
    \bV := \plim_{T\rightarrow \infty} T^{-1}Z_0'MZ_0, \;\;\; \bU:= \lim_{T\rightarrow \infty} E\left[\left(T^{-1/2}Z_0'M\varepsilon\right)^2\right]. 
\end{equation*}
\end{corollary}

%%%%%%%%%%%%%%%%%%%%%%%%%%%%%%%%%%%%%%%%%%%%%%%%%%%%%% Limit Distribution 
%%%%%%%%%%%%%%%%%%%%%%%%%%%%%%%%%%%%%%%%%%%%%%
\section{In-fill Asymptotic Distribution}\label{sec:limitdist}

\citet*{Bai1997} provides the limit distribution of the LS estimator assuming large breaks ($\delta = O(T^{-1/2+\epsilon})$ with $0<\epsilon< 1/2$). The asymptotic distribution is symmetric at the true break point, if the second moment of variables associated with coefficients under break ($z_t$ in Section \ref{sec:multimodel}) do not change before and after break. I am interested in small breaks in which the asymptotic distribution depends on the parameters in a complicated manner (\citeauthor{Elliott2007} \citeyear{Elliott2007}).

\citeauthor{Jiang2017} (\citeyear{Jiang2017}, \citeyear{Jiang2018}) and \citet*{Casini2019} employed a continuous record asymptotic framework to derive the limit distribution of the break point estimator. By assuming that a continuous record is available, a continuous time approximation to the discrete time model is constructed and an in-fill asymptotic distribution is developed. In contrast to the long-span asymptotic, where the time span of the data increases, the in-fill asymptotic assumes a fixed time span with shrinking sampling intervals. The in-fill asymptotic distribution is asymmetric, tri-modal, and dependent on the initial condition. However, the long-span asymptotic distribution of the LS estimator under local-to-unity processes do not depend on the initial condition. See \citet{Chong2001}, \citet{Pang2014} and \citet{Pang2018} on the long-span asymptotic distribution of the LS estimator under different settings of the AR root before and after the break. I follow the approach of \citeauthor{Jiang2017} (\citeyear{Jiang2017}, \citeyear{Jiang2018}) to derive the limit distribution of the break point estimator under a stationary and local-to-unity autoregressive process.

%%%%%%%%%%%%%%%%%%%%%%%%%%%%%%%%%%%%%%%%%%%%%%%%%%%%%%%%%%%%%%%%%%%
\subsection{Partial break in a stationary process}\label{subsec:infill_stationary}

Consider the linear regression model (\ref{eq:DGPmulti}) with continuous time process $\{W_s,Z_s,\capeps_s\}_{s\geq 0}$ defined on a filtered probability space $(\Omega,\mathcal{F},(\mathcal{F}_s)_{s\geq 0},P)$, where $s$ can be interpreted as a continuous time index. Assume we observe at discrete points of time so that $\{Y_{th}, W_{th}, Z_{th}: t=0,1,\ldots,T = N/h\}$, where $N$ is the time span. We normalize the time span $N=1$ for simplicity. We denote the increment of processes as $\Delta_h Y_t := Y_{th}-Y_{(t-1)h}$. Let $X_{th} = (W_{th}', Z_{th}')'$ so that $Z_{th} = R'X_{th}$. The model (\ref{eq:DGPmulti}) can be expressed as
\begin{equation*}
    \Delta_h Y_{t} = \begin{cases}
    (\Delta_h X_t)'\beta_h +  \Delta_h \capeps_t & \text{ if } t = 1,\ldots,[\rho_0 T] \\
    (\Delta_h X_t)'\beta_h + (\Delta_h Z_t)'\delta_h + \Delta_h \capeps_t & \text{ if } t = [\rho_0 T]+1,\ldots,T, 
    \end{cases}
\end{equation*}
Divide both sides by $\sqrt{h}$ so that the error term variance is $O(1)$. The parameters $\beta_h$ and $\delta_h$ may depend on the sampling interval, denoted by subscript $h$. Let $\varepsilon_{t} := \Delta_h \capeps_t/\sqrt{h}$, $y_t := \Delta_h Y_t/\sqrt{h}$, $x_t := \Delta_h X_t/\sqrt{h}$, $z_t := \Delta_h Z_t/\sqrt{h} = R'x_t$, 
\begin{equation}\label{eq:DGPinfill_multi}
    y_t = \begin{cases}
    x_t'\beta_h + \varepsilon_t & \text{ if } t = 1,\ldots, [\rho_0 T] \\
    x_t'\beta_h + z_t'\delta_h + \varepsilon_t & \text{ if } t = [\rho_0 T]+1,\ldots,T.
    \end{cases}
\end{equation}

\smallskip

\begin{assumption}\label{A:infillmult_covstat}
$\{z_t, \varepsilon_t\}$ is a covariance stationary process that satisfies the functional central limit theorem as $T = 1/h \rightarrow \infty$,
    \begin{equation*}
        T^{-1/2}\sum_{t=1}^{[sT]} z_t \varepsilon_t \Rightarrow B_1(s),
    \end{equation*}
    where $B_1(s)$ is a multivariate Gaussian process on $[0, 1]$ with mean zero and covariance \\
    $E[B_1(u)B_1(v)']= \min\{u,v\}\Xi$, and $\Xi := \lim_{T\rightarrow \infty} E\left[\left(T^{-1/2}\sum_{t=1}^{T} z_t \varepsilon_t \right)^2\right]$.
\end{assumption}

\smallskip

\begin{assumption}\label{A:infillmult_delta}
The break magnitude is $\delta_h = d_0 \lambda_h$, where $d_0 \in \mathbb{R}^q$ is a fixed vector and $\lambda_h$ is a scalar that depends on the sampling interval $h$. Assume one of the following cases on $\lambda_h$ as $h \rightarrow 0$,
\begin{enumerate}[(i)]
    \item $\lambda_h = O(h^{1/2})$ so that $\delta_h = d_0 \sqrt{h}$;
    \item $\lambda_h = O(h^{1/2-\gamma})$, where $0<\gamma < 1/2$ so that $\delta_h/ \sqrt{h} \rightarrow \infty$ simultaneously with $\delta_h \rightarrow 0$.
\end{enumerate}
\end{assumption}

\smallskip

\noindent Notations from Section \ref{sec:multimodel} are used for model (\ref{eq:DGPinfill_multi}): $MY = MZ_0 \delta_h + M\varepsilon$, where $\varepsilon = (\varepsilon_1,\ldots,\varepsilon_T)'$ and $M = I-X(X'X)^{-1}X'$. The objective function of the estimator $\hat{k}$ in (\ref{eq:newhatk_multi}) is restated as follows:
\begin{equation}\label{eq:objfn_Tscale}
    Q_T(k)^2 = \sqrt{T}\hat{\delta}_k'(T^{-1}Z_k'MZ_k)^{1/2}\Omega_k (T^{-1}Z_k'MZ_k)^{1/2}\sqrt{T}\hat{\delta}_k
\end{equation}
The in-fill asymptotic distribution is derived for the two different magnitudes of $\delta_h$ in Assumption \ref{A:infillmult_delta}. Theorem \ref{Thm:infillmult_dfixed} provides the limit distribution under \ref{A:infillmult_delta}$(i)$, which represents small breaks. For proof, see Appendix B.

%%%%%%%%% Theorem: In fill asymptotic multvariate dist, d fixed
\begin{theorem}\label{Thm:infillmult_dfixed}
Consider the model (\ref{eq:DGPinfill_multi}) with unknown parameters $(\beta_h,\delta_h)$. Assumption \ref{A:multimodel}, \ref{A:Omegak}, \ref{A:infillmult_covstat}, and \ref{A:infillmult_delta}$(i)$ holds. Then the break point estimator $\hat{\rho} = \hat{k}/T$ defined in (\ref{eq:newhatk_multi}) has the following in-fill asymptotic distribution as $h \rightarrow 0$,
\begin{equation*}
    T\norm{\delta_h}^2 \hat{\rho}  \overset{d}{\longrightarrow} \norm{d_0}^2\argmax_{\rho \in (0, 1)} \wtilde{W}(\rho)'\bar{\Omega}(\rho) \wtilde{W}(\rho),
\end{equation*}
with 
\begin{align*}
    \wtilde{W}(\rho) & := 
    \Sigma_z^{-1/2}\frac{B_1(\rho)-\rho B_1(1)}{\sqrt{\rho(1-\rho)}}- (1-\rho_0)\left(\frac{\rho}{1-\rho}\right)^{1/2}\Sigma_z^{1/2} d_0 \hspace{0.7cm}\text{ if } \rho \leq \rho_0 \\
    & := \Sigma_z^{-1/2}\frac{B_1(\rho)- \rho B_1(1)}{\sqrt{\rho(1-\rho)}}-\rho_0\left(\frac{1-\rho}{\rho}\right)^{1/2}\Sigma_z^{1/2}d_0 \hspace{1.7cm} \text{ if } \rho > \rho_0, 
\end{align*}
where $B_1(\cdot)$ is a Brownian motion defined in Assumption \ref{A:infillmult_covstat}.
\end{theorem}

\noindent An equivalent representation of the in-fill asymptotic distribution is (let $\rho = \rho_0+u$)
\begin{equation*}
    T\norm{\delta_h}^2 (\hat{\rho}-\rho_0) \overset{d}{\longrightarrow} \norm{d_0}^2\argmax_{u \in (\rho_0,1-\rho_0)} \wtilde{W}(\rho_0+u)'\bar{\Omega}(\rho_0+u)\wtilde{W}(\rho_0+u),
\end{equation*}
where $\wtilde{W}(\cdot)$ is defined in Theorem \ref{Thm:infillmult_dfixed}.

Next, consider the case of Assumption \ref{A:infillmult_delta}$(ii)$. The proof of Theorem \ref{Thm:infillmult_dlarge} is in Appendix B.

\begin{theorem}\label{Thm:infillmult_dlarge}
Consider the model (\ref{eq:DGPinfill_multi}) with unknown parameters $(\beta_h,\delta_h)$. Assumptions \ref{A:multimodel}, \ref{A:Omegak}, \ref{A:infillmult_covstat}, and \ref{A:infillmult_delta}$(ii)$ hold. For simplicity, we denote $\bar{\Omega}_0$ for $\bar{\Omega}(\rho_0)$. Then the break point estimator $\hat{\rho} = \hat{k}/T$ defined in (\ref{eq:newhatk_multi}) has the following in-fill asymptotic distribution as $h \rightarrow 0$,
\begin{equation*}
    \lambda_h^2 T(\hat{\rho}-\rho_0) \overset{d}{\longrightarrow} \frac{(d_0'\bar{\Omega}_0\Xi\bar{\Omega}_0 d_0)}{(d_0'\Sigma_z A_u d_0)^2} \argmax_{u \in (-\infty,\infty)}(d_0'\Sigma_z A_u d_0)^{-1} \left\{W(u)-\frac{\vert u \vert}{2}\right\},
\end{equation*}
where $A_u = \bar{\Omega}_0 -\text{sgn}(u)\rho_0(1-\rho_0)\nabla \bar{\Omega}_0$, $\nabla\bar{\Omega}_0 \equiv \left. \frac{\partial \bar{\Omega}(\rho)}{\partial \rho}\right\vert_{\rho = \rho_0}$, $W(u) = W_1(-u)$ for $u \leq 0$ and $W(u) = W_2(u)$ for $u>0$. $W_1(\cdot)$ and $W_2(\cdot)$ are two independent Wiener processes on $[0,\infty)$.
\end{theorem}

If the weight matrix is $\Omega_k = I_q$, the estimator is equivalent to the LS estimator and the limiting distribution reduces to the distribution in Proposition 3 of \citet*{Bai1997}. The term $A_u$ shows how the weight matrix down-weights break points near the boundaries. Suppose $\Omega_k = T^{-1}Z_k'MZ_k$ and $\rho_0 > 0.5$ so that $\bar{\Omega}(\rho) = \rho(1-\rho)\Sigma_z$ and  $\nabla \bar{\Omega}_0 < 0$. If $u > 0$ increases in a positive direction toward the boundary (i.e., $\rho > \rho_0 > 0.5$), then $A_u$ increases and the term multiplied to the Wiener process with drift, $(d_0'\Sigma_z A_u d_0)^{-1}$ decreases. In contrast, if $u<0$ decreases such that $\rho$ shifts toward the median, then $A_u$ decreases and $(d_0'\Sigma_z A_u d_0)^{-1}$ increases. The result is opposite if $\rho_0 < 0.5$. That is, there is larger weight on $\rho$ near the median $0.5$ and less weight near the boundaries.

%%%%%%%%%%%%%%%%%%%%%%%%%%%%%%%%%%%%%%%%%%%
%%%%%%%%% Autoregressive model in-fill asymptotics
%%%%%%%%%%%%%%%%%%%%%%%%%%%%%%%%%%%%%%%%%%%%
\subsection{Break in an autoregressive model}\label{subsec:infill_ar}

In this section, I derive the in-fill asymptotic distribution of an autoregressive (AR) model with a structural break in its lag coefficient, using a deterministic weight function $\omega(\cdot)$. As mentioned in Section \ref{sec:multimodel}, Assumption \ref{A:multimodel} excludes lagged dependent variables, due to the dependence of the weight function on regressors. This condition is relaxed to allow weakly exogeneous regressors by assuming non-stochastic weights. Consider a discrete model closely related to the Ornstein-Uhlenbeck process with a break in the drift function:
\begin{equation*}
    dx(t) = -(\mu + \delta\bbone\{t>\rho_0\}) x(t)dt + \sigma dB(t),
\end{equation*}
where $t \in [0,1]$ and $B(\cdot)$ denote a standard Brownian motion. The discrete time model has the form
\begin{equation*}
    x_t = \left(\beta_1 \bbone\{t\leq k_0\} + \beta_2 \bbone\{t > k_0\}\right)x_{t-1} + \sqrt{h}\varepsilon_t, \;\; \varepsilon_t \overset{i.i.d.}{\sim} (0,\sigma^2), \;\; x_0 = O_p(1), 
\end{equation*}
where $\beta_1 = \exp\{-\mu/T\}$ and $\beta_2 = \exp\{-(\mu+\delta)/T\}$ are the AR roots before and after the break. We denote $y_t = x_t/\sqrt{h}$ so that the order of errors is $O_p(1)$ as in model (\ref{eq:DGPmulti}). Then, I have for $t = 1,\ldots,T$,
\begin{equation}\label{eq:DGPinfill_ar}
    y_t = (\beta_1 \bbone\{t\leq k_0\} + \beta_2 \bbone\{t > k_0\})y_{t-1} + \varepsilon_t, \;\; \varepsilon_t \overset{i.i.d.}{\sim} (0,\sigma^2), \;\; y_0 = x_o/\sqrt{h} = O_p\left(T^{1/2}\right).
\end{equation}
The initial condition of $y_t$ in (\ref{eq:DGPinfill_ar}) diverges at rate $T^{1/2}$; thus, the in-fill asymptotic distribution will depend explicitly on the initial value $x_0$. The break size is $\beta_2-\beta_1 = O(T^{-1})$, whereas the literature on long-span asymptotics assumes $O(T^{-\gamma})$ with $0<\gamma <1$. The model (\ref{eq:DGPinfill_ar}) is a local-to-unit root process: $\beta_1 = \exp\{-\mu /T\} \rightarrow 1$ and $\beta_2 = \exp \{-(\mu+\delta)/T\} \rightarrow 1$, as $T \rightarrow \infty$ for any finite $(\mu,\delta)$. In contrast, the long-span asymptotic theory incorporates stationary AR(1) processes, where $\vert \beta_1 \vert < 1$ and $\vert \beta_2 \vert < 1$. \citet*{Chong2001} derives the long-span distribution under $\vert\beta_2-\beta_1\vert = O(T^{-1/2+\gamma})$ with $0 < \gamma < 1/2$. \citet{Jiang2017} provides simulation results that the in-fill asymptotic theory works well even when $\beta_1$ and/or $\beta_2$ are distant from unity in the finite sample.

The break point estimator and the LS estimator in model (\ref{eq:DGPinfill_ar}) takes the form 
\begin{gather}
\nonumber
    S(k)^2 = \sum_{t=1}^k \left(y_t - \hat{\beta}_1(k)y_{t-1}\right)^2 +\sum_{t=k+1}^T \left(y_t - \hat{\beta}_2(k)y_{t-1}\right)^2\\
\label{eq:hatk_infillar}    
    \hat{k} = \argmin_{k=1,\ldots,T-1} \omega_k^2 \, S(k)^2, \;\;\; \hat{\rho} = \hat{k}/T\\
\nonumber    
    \hat{k}_{LS} = \argmin_{k=1,\ldots,T-1} S(k)^2, \;\;\; \hat{\rho}_{LS} = \hat{k}_{LS}/T,
\end{gather}
where $\hat{\beta}_1(k) = \sum_{t=1}^k y_t y_{t-1}/ \sum_{t=1}^k y_{t-1}^2$ and $\hat{\beta}_2(k) = \sum_{t=k+1}^T y_t y_{t-1}/ \sum_{t=k+1}^T y_{t-1}^2$
are LS estimates of $\beta_1$ and $\beta_2$ under break at $k$, respectively.

\begin{theorem}\label{Thm:infill_ar}
Consider the model (\ref{eq:DGPinfill_ar}) with fixed parameters $(\mu,\delta)$ so that $\ln \beta_1 = O(T^{-1})$ and $\, \ln\beta_2 = O(T^{-1})$. Assume the weight function $\omega_k$ is nonrandom and bounded on the unit interval with $\omega_k \rightarrow \omega(\rho)$ as $k/T\rightarrow \rho$. Then, the break point estimator $\hat{\rho}= \hat{k}/T$ in (\ref{eq:hatk_infillar}) has the in-fill asymptotic distribution as
\begin{equation*}
    \hat{\rho} \Longrightarrow \argmax_{\rho \in (0, 1)}\; \omega(\rho)^2 \left[\frac{\left(\wtilde{J}_{0}(\rho)^2 - \wtilde{J}_{0}(0)^2-\rho\right)^2}{\int_{0}^{\rho}\wtilde{J}_{0}(r)^2 dr} + \frac{\left(\wtilde{J}_{0}(1)^2 - \wtilde{J}_{0}(\rho)^2-(1-\rho)\right)^2}{\int_{\rho}^{1}\wtilde{J}_{0}(r)^2 dr}\right],
\end{equation*}
where $\wtilde{J}_0(r)$, for $r \in [0, 1]$ is a Gaussian process defined by
\begin{equation}\label{eq:Jtilde0}
    d\wtilde{J}_0(r) = -(\mu + \delta \bbone\{r>\rho_0\})\wtilde{J}_0(r)dr + dB(r),
\end{equation}
with the initial condition $\wtilde{J}_0(0) = y_0/\sigma = x_0/(\sigma\sqrt{h})$, and $B(\cdot)$ is a standard Brownian motion.
\end{theorem}

\smallskip

The results of Theorem \ref{Thm:infill_ar} derive from applying the continuous mapping theorem to the limit distribution $S(k)^2$ in Theorem 4.1 from \citet{Jiang2017}. See Appendix B for the proof. The difference between the asymptotic distributions of the two estimators is the weight function multiplied by the stochastic process in the argmax function. Both estimators are asymmetrically distributed around the true point and biased when $\rho_0 \neq 1/2$.

%%%%%%%%%%%%%%%%%%%%%%%%%%%%%%%%%%%%%%%%%%%%%%%%%%%%% Monte Carlo Simulation
%%%%%%%%%%%%%%%%%%%%%%%%%%%%%%%%%%%%%%%%%%%%%%
\section{Monte Carlo Simulation}\label{sec:montecarlo}

This section compares finite sample distributions of the new estimator and the LS estimator using Monte Carlo simulation. It considers two different models; a break in the mean of a univariate regression model and a break in the lag coefficient of the AR(1) process. I compare the root mean squared error (RMSE), bias, and standard errors of the two estimators in the finite sample and in-fill asymptotics.

%%%%%%%%%%%%%%%%%%%%%%%%%%%%%%%%%%%%%%%%%%%%%%%%%%%%%%%%
\subsection{Univariate stationary process}\label{subsec:mc_stationary}

The first model is when a structural break occurs in model (\ref{eq:DGPmulti}), where $x_t = z_t = 1$ for all $t$. The break magnitude $\delta_T = T^{-1/2}d_0$ is in the local $T^{-1/2}$ neighborhood of zero to represent small break magnitudes.
\begin{equation}\label{eq:mc_stationarydgp}
    y_t = \mu + \delta_T\bbone\{t > [\rho_0 T]\}  + \varepsilon_t,
\end{equation}
where $\varepsilon_t \overset{i.i.d.}{\sim} N(0,\sigma^2)$ and $\sigma = 1$. Parameter values are $\rho_0 \in \{0.15, 0.3, 0.5, 0.7, 0.85\}$, $\mu = 4$, $d_0 \in \{1, 2, 4 \}$, and $T=100$ with 5,000 replications. The weight function is $\omega_k = (k/T(1-k/T))^{1/2}$, which is the representative weight function motivated in Section \ref{sec:motivation}\footnote{If the weight function is $\omega_k = (k/T(1-k/T))^{\gamma}$, Assumption \ref{A:Omegak} is satisfied if $-1/2\leq \gamma \leq 1/2$ for an arbitrary small $\alpha$ in Assumption \ref{A:multimodel}$(i)$. For $\gamma \in \{1/8, 1/4, 3/8\}$, the results (omitted due to space constraints) do not change qualitatively; the probability at the boundaries decrease compared to the finite sample distribution of LS. Because $\omega(\rho) = (\rho(1-\rho))^\gamma \rightarrow 1$ as $\gamma \rightarrow 0$, the difference between the two estimators finite sample behavior shrinks when $\gamma$ is close to zero.}. The break point estimator $\hat{\rho}_{NEW}$ is defined in (\ref{Def:newestimator_breakmean}) and the LS estimator $\hat{\rho}_{LS}$ in (\ref{Def:hatrhols_breakmean}). Although it is unnecessary to trim under the simple model (\ref{eq:mc_stationarydgp}), I trim the optimization space by fraction $\alpha=0.1$ on both ends, following the common practice in the literature.

Table \ref{table:rmse_fini_stationary} provides the RMSE, the bias, and the standard error for the finite sample distribution. For all $\rho_0$ and $d_0$ values considered, the RMSE of the estimator $\hat{\rho}_{NEW}$ is smaller than that of $\hat{\rho}_{LS}$ in the finite sample. A comparison of asymptotic RMSE shows the same results qualitatively (see Appendix C). A trade-off emerges of slightly larger bias but a large decrease in standard error for $\hat{\rho}_{NEW}$ compared to $\hat{\rho}_{LS}$, which leads to a decrease in RMSE. When $\rho_0 = 0.5$, both bias and standard error of the new estimator is smaller than the LS estimator. 

Figures \ref{fig:finitedist_r30} and \ref{fig:finitedist_r85} show the finite sample distribution of the two estimators under $\rho_0 = 0.30$, and 0.85. Under $\rho_0 = 0.30$, the finite sample distribution of $\hat{\rho}_{LS}$ is tri-modal, whereas the $\hat{\rho}_{NEW}$ has an unique mode at $\rho_0$ for all $d_0$ values considered. When the true break point is near the boundaries of the optimization space and the break magnitude is small ($\rho_0 = 0.85$ and $d_0 = 1$), the LS estimator performs particularly worse. The tri-modal LS estimator distribution becomes bi-modal with modes at $\alpha$ and $1-\alpha$. Under these parameter values, the new estimator has its drawbacks; the finite sample distribution is relatively flat. This is partly due to trimming the optimization space, which is not necessary for our estimation method. Without trimming ($\alpha = 0$), the finite sample distribution of our break point estimator has a unique mode at $\rho_0$ for all parameter values considered.

In short, the break point estimator is preferable than the LS estimator in terms of RMSE, under small break magnitudes. The weight function biases the estimator toward the median in trade-off to a significant decrease in standard error. When a break occurs near the boundaries, we need to be careful, because the new estimator can also be problematic. I suggest minimizing the trimming fraction $\alpha$ and using the new break point estimator.

\begin{table}[H]
\caption{Finite sample RMSE, bias, and the standard error of the new estimator and the LS estimator of the break point under model (\ref{eq:mc_stationarydgp}) with parameter values $(\rho_0,d_0)$ and $T = 100$. The number of replications is 5,000.}
\begin{center}  
\begin{tabular}{lL{0.8cm}C{1.5cm}C{1.5cm}C{1.5cm}C{1.5cm}C{1.5cm}C{1.5cm}} 
\hline
\hline
& & \multicolumn{2}{c}{RMSE} & \multicolumn{2}{c}{Bias} & \multicolumn{2}{c}{Standard error}\Tstrut\\
\cline{3-8}
 $\rho_0$ & $d_0$ & NEW & LS & NEW & LS & NEW & LS \Tstrut\\
\hline
 \multirow{3}{*}{0.15} & 1 & 0.4034 & 0.4381 & 0.3428 & 0.3401 & 0.2127 & 0.2762\Tstrut\\
  & 2 & 0.3897 & 0.4211 & 0.3243 & 0.3152 & 0.2161 & 0.2792\\
  &4 & 0.3455 & 0.3556 & 0.2703 & 0.2323 & 0.2151 & 0.2693\\
  \hline
 \multirow{3}{*}{0.30} & 1 & 0.2853 & 0.3303 & 0.1933 & 0.1898 & 0.2098 & 0.2703\Tstrut\\
  & 2 & 0.2669 & 0.3150 & 0.1741 & 0.1703 & 0.2023 & 0.2649\\
 & 4 & 0.2018 & 0.2435 & 0.1139 & 0.0999 & 0.1666 & 0.2221\\ 
 \hline
 \multirow{3}{*}{0.50} & 1 & 0.2051 & 0.2681 & -0.0029 & -0.0041 & 0.2051 & 0.2680\Tstrut\\
  & 2 & 0.1876 & 0.2511 & -0.0017 & -0.0029 & 0.1876 & 0.2511\\
 & 4 & 0.1359 & 0.1985 & 0.0002 & -0.0020 & 0.1359 & 0.1985\\
 \hline
 \multirow{3}{*}{0.70} & 1 & 0.2866 & 0.3334 & -0.1940 & -0.1915 & 0.2109 & 0.2729\Tstrut\\
  & 2 & 0.2640 & 0.3104 & -0.1693 & -0.1641 & 0.2025 & 0.2635\\
  & 4 & 0.2043 & 0.2448 & -0.1151 & -0.1029 & 0.1689 & 0.2222\\
  \hline
 \multirow{3}{*}{0.85} & 1 & 0.4018 & 0.4394 & -0.3405 & -0.3386 & 0.2134 & 0.2800\Tstrut\\
  & 2 & 0.3913 & 0.4224 & -0.3265 & -0.3176 & 0.2157 & 0.2785\\
  & 4 & 0.3438 & 0.3524 & -0.2678 & -0.2275 & 0.2155 & 0.2692\\ 
 \hline
\hline
\end{tabular}
\end{center}
\label{table:rmse_fini_stationary}
\end{table}

\begin{figure}[H]
\centering
\begin{subfigure}{.48\textwidth}
  \centering
    \includegraphics[width=\linewidth]{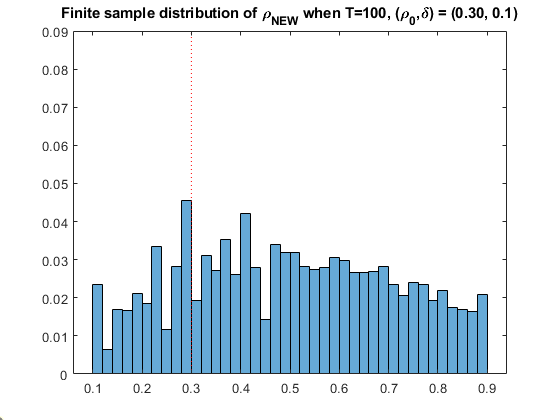}
\end{subfigure}
\begin{subfigure}{.48\textwidth}
  \centering
    \includegraphics[width=\linewidth]{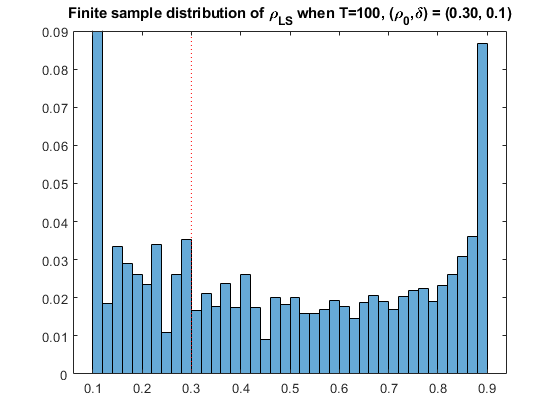}
\end{subfigure}    
\begin{subfigure}{.48\textwidth}
  \centering
    \includegraphics[width=\linewidth]{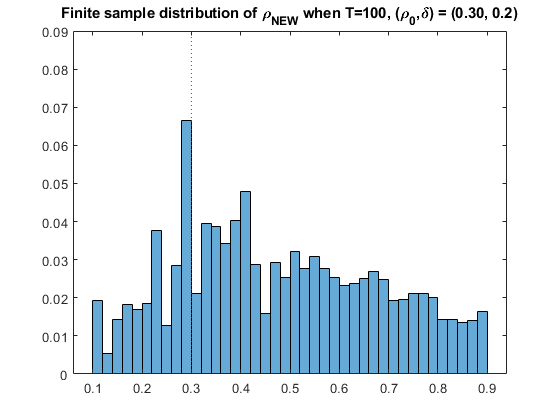}
\end{subfigure}
\begin{subfigure}{.48\textwidth}
  \centering
    \includegraphics[width=\linewidth]{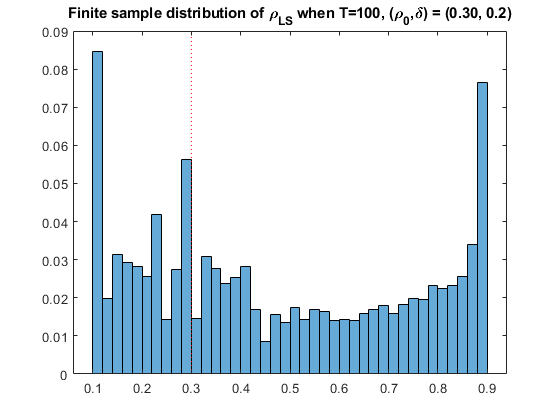}
\end{subfigure}
\begin{subfigure}{.48\textwidth}
  \centering
    \includegraphics[width=\linewidth]{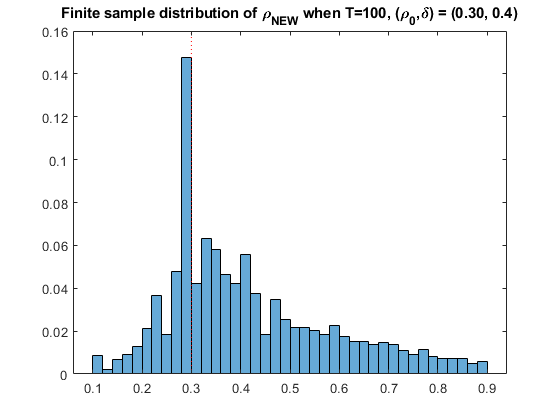}
\end{subfigure}
\begin{subfigure}{.48\textwidth}
  \centering
    \includegraphics[width=\linewidth]{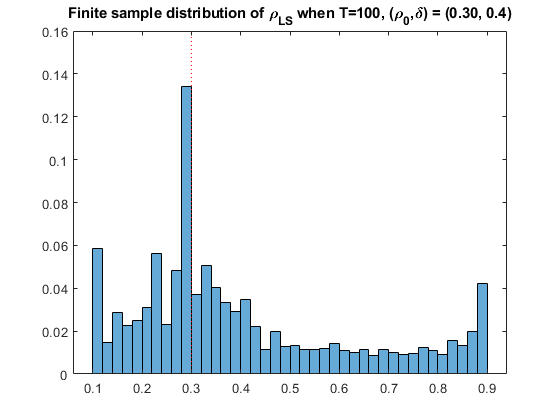}
\end{subfigure}
\caption{($\rho_0 = 0.30$) Finite sample distribution of the new estimator $\hat{\rho}_{NEW}$ (left) and the LS estimator $\hat{\rho}_{LS}$ (right) under model (\ref{eq:mc_stationarydgp}), with parameter values $(\rho_0,\delta_T) = (0.3, T^{-1/2})$, $(0.3, 2T^{-1/2})$, and $(0.3, 4T^{-1/2})$ and $T=100$, respectively. The optimization space is trimmed by fraction $\alpha$ on both ends.}
\label{fig:finitedist_r30}
\end{figure}

\begin{figure}[H]
\centering
\begin{subfigure}{.48\textwidth}
  \centering
    \includegraphics[width=\linewidth]{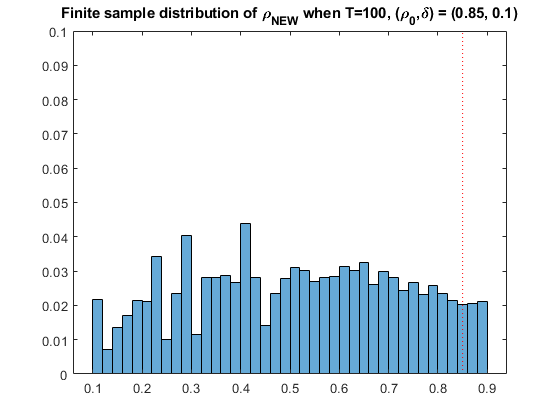}
\end{subfigure}
\begin{subfigure}{.48\textwidth}
  \centering
    \includegraphics[width=\linewidth]{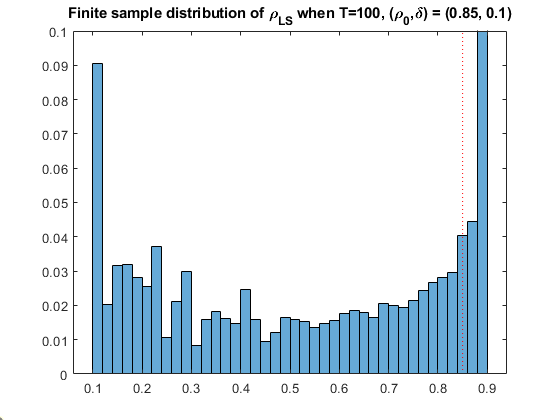}
\end{subfigure}
\begin{subfigure}{.48\textwidth}
  \centering
    \includegraphics[width=\linewidth]{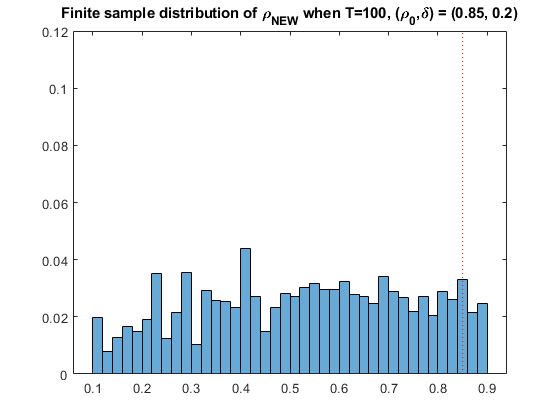}
\end{subfigure}
\begin{subfigure}{.48\textwidth}
  \centering
    \includegraphics[width=\linewidth]{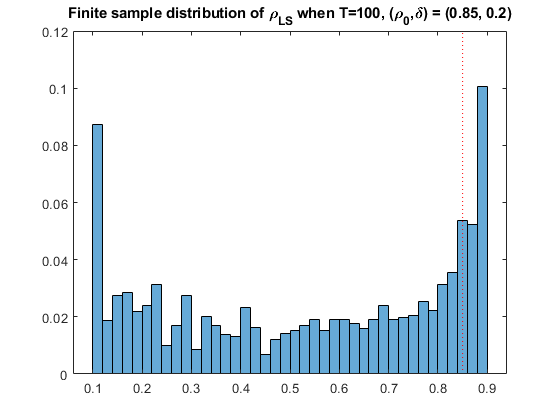}
\end{subfigure}
\begin{subfigure}{.48\textwidth}
  \centering
    \includegraphics[width=\linewidth]{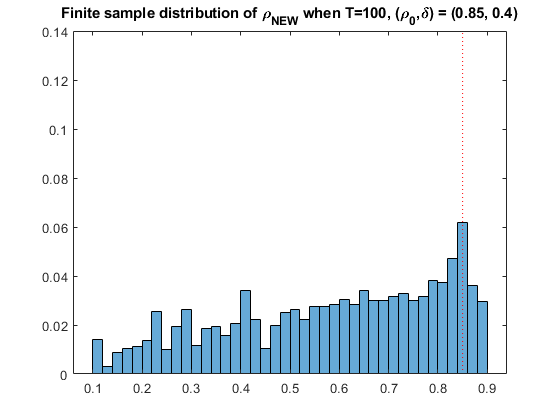}
\end{subfigure}
\begin{subfigure}{.48\textwidth}
  \centering
    \includegraphics[width=\linewidth]{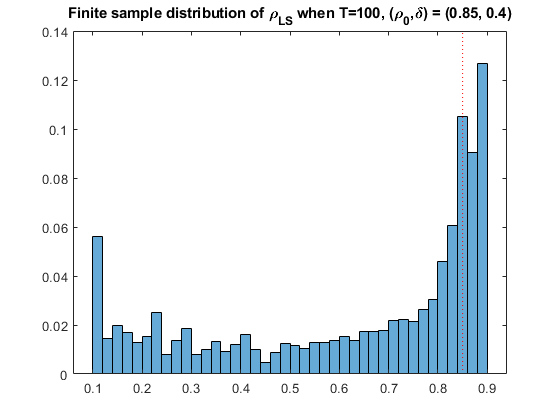}
\end{subfigure}
\caption{($\rho_0 = 0.85$) Finite sample distribution of the new estimator $\hat{\rho}_{NEW}$ (left) and the LS estimator $\hat{\rho}_{LS}$ (right) under model (\ref{eq:mc_stationarydgp}), with parameter values $(\rho_0,\delta_T) = (0.85, T^{-1/2})$, $(0.85, 2T^{-1/2})$, and $(0.85, 4T^{-1/2})$ and $T=100$, respectively.  The optimization space is trimmed by fraction $\alpha$ on both ends.}
\label{fig:finitedist_r85}
\end{figure}

%%%%%%%%%%%%%%%%%%%%%%%%%%%%%%%%%%%%%%%%%%%%%%%%%%%%%%%%%
\subsection{Autoregressive process}\label{subsec:mc_ar1}

For the AR(1) process, I replicate two experiments from \citet{Jiang2017}. The first experiment is a break in the lag coefficient so that the stationary process changes to another stationary AR(1) process. The second case is a change from a local-to-unit root to a stationary AR(1) process. Each experiment is generated from model (\ref{eq:DGPinfill_ar}) with $h=1/200$ ($T=200$), $\sigma = 1$, $\varepsilon_t \overset{i.i.d.}{\sim} N(0,1)$, $\rho_0 \in \{0.3, 0.5, 0.7\}$, and different combinations of $\mu$ and $\delta$ with $\beta_1 =\exp(-\mu/T)$ and $\beta_2 = \exp(-(\mu+\delta)/T)$. 
\begin{enumerate}
    \item Stationary to stationary: $(\mu,\delta) = (138,55)$, which implies $(\beta_1,\beta_2) = (0.5,0.38)$;
    \item Local-to-unity to stationary: $(\mu,\delta) = (1,5)$, which implies $(\beta_1,\beta_2) = (0.995,0.97)$.
\end{enumerate}
The stochastic integrals of in-fill asymptotic distributions are approximated over a grid size $h = 0.005$. The optimization space is trimmed by fraction $\alpha=0.1$. The break point estimator $\hat{\rho}_{NEW}$ of the AR(1) model is defined in (\ref{eq:hatk_infillar}) and its asymptotic distribution is stated in Theorem \ref{Thm:infillmult_dfixed}. The in-fill asymptotic distribution of the LS estimator $\hat{\rho}_{LS}$ is stated in Theorem 4.1 from \citet{Jiang2017}. 

Table \ref{table:rmse_ar1_fini} provides the RMSE, bias, and the standard error of $\hat{\rho}_{NEW}$ and $\hat{\rho}_{LS}$ for the finite sample, respectively (see Table \ref{table:rmse_ar1_infill} in Appendix C for the asymptotic distribution). Similar to results in section \ref{subsec:mc_stationary}, the RMSE of $\hat{\rho}_{NEW}$ is smaller than that of $\hat{\rho}_{LS}$ for all parameter values $(\beta_1,\beta_2,\rho_0)$ considered. This also holds in the limit. A decrease in RMSE of $\hat{\rho}_{NEW}$ emerges from the trade-off of a relatively large decrease in variance compared to the increase in the squared bias.

Figures \ref{fig:ar_fini_stationary} and \ref{fig:ar_fini_localtounit} are finite sample distributions of the break point in the two experiments. For the stationary to another stationary process change, the LS estimator $\hat{\rho}_{LS}$ mode at the true break point is almost negligible, unless it is the median $\rho_0 = 0.5$. In contrast, the estimator $\hat{\rho}_{NEW}$ has a unique mode at the true break point for all $\rho_0 \in \{0.3, 0.5, 0.7\}$. For the local-to-unit root to a stationary AR(1) change, both estimators have a higher probability at the true break point. However, the LS estimator continues to exhibit tri-modality with modes at the ends, whereas the new estimator has a unique mode at $\rho_0$.

\begin{figure}[H]
\centering
\begin{subfigure}{.48\textwidth}
  \centering
    \includegraphics[width=\linewidth]{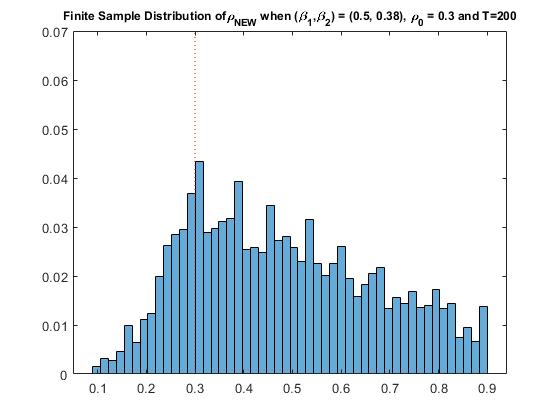}
\end{subfigure}
\begin{subfigure}{.48\textwidth}
  \centering
    \includegraphics[width=\linewidth]{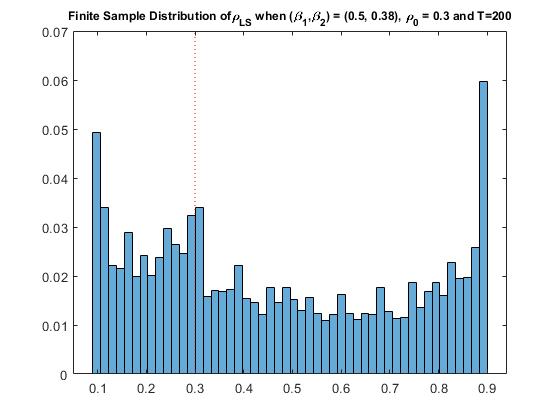}
\end{subfigure}
\begin{subfigure}{.48\textwidth}
  \centering
    \includegraphics[width=\linewidth]{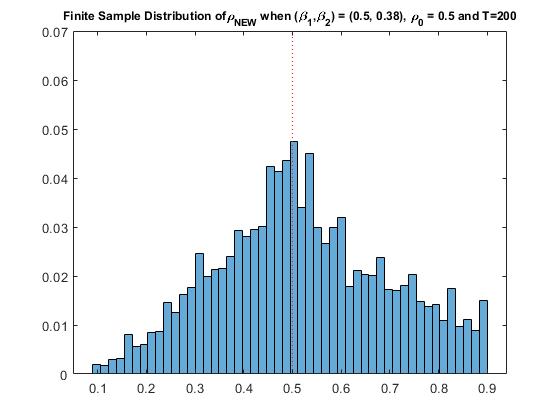}
\end{subfigure}
\begin{subfigure}{.48\textwidth}
  \centering
    \includegraphics[width=\linewidth]{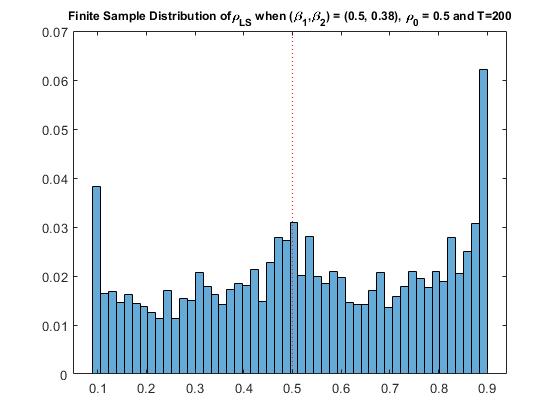}
\end{subfigure}
\begin{subfigure}{.48\textwidth}
  \centering
    \includegraphics[width=\linewidth]{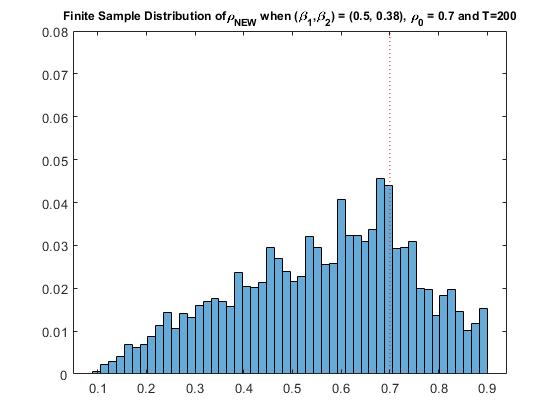}
\end{subfigure}
\begin{subfigure}{.48\textwidth}
  \centering
    \includegraphics[width=\linewidth]{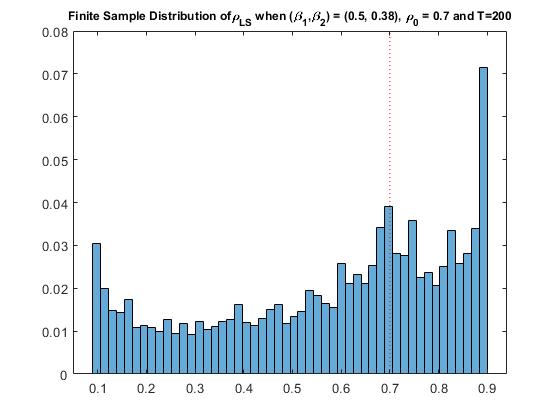}
\end{subfigure}
\caption{(Stationary to stationary) Finite sample distributions of the new estimator (left) and the LS estimator (right) when the lag coefficient pre- and post-break are $(\beta_1,\beta_2) = (0.5,0.38)$ at break points $\rho_0 = 0.3, 0.5$, and 0.7, respectively. }
\label{fig:ar_fini_stationary}
\end{figure}

\begin{figure}[H]
\centering
\begin{subfigure}{.48\textwidth}
  \centering
    \includegraphics[width=\linewidth]{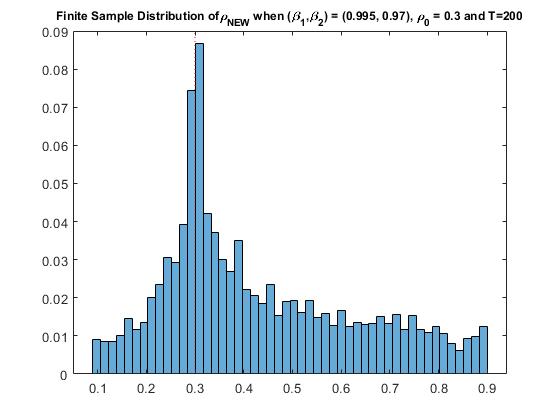}
\end{subfigure}
\begin{subfigure}{.48\textwidth}
  \centering
    \includegraphics[width=\linewidth]{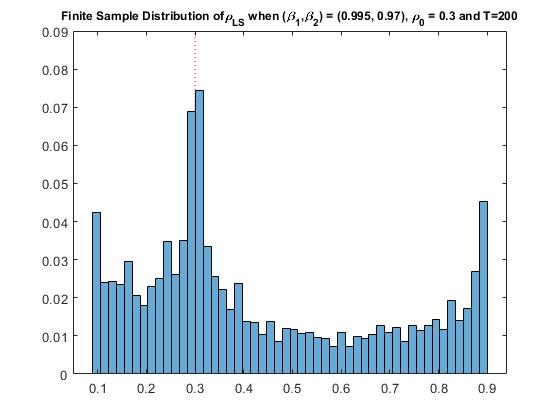}
\end{subfigure}
\begin{subfigure}{.48\textwidth}
  \centering
    \includegraphics[width=\linewidth]{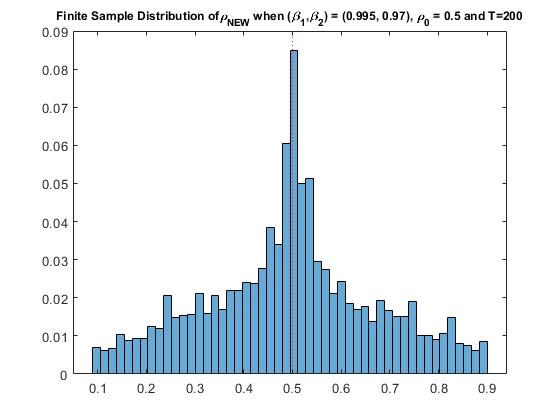}
\end{subfigure}
\begin{subfigure}{.48\textwidth}
  \centering
    \includegraphics[width=\linewidth]{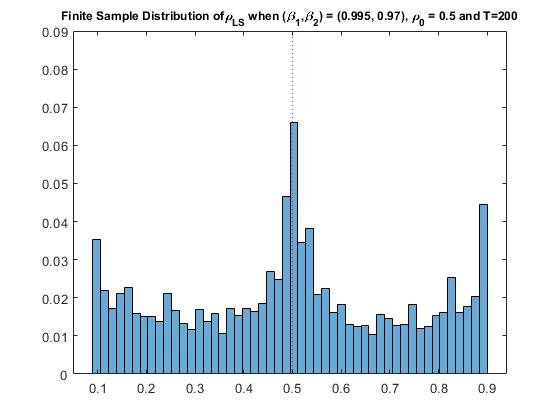}
\end{subfigure}
\begin{subfigure}{.48\textwidth}
  \centering
    \includegraphics[width=\linewidth]{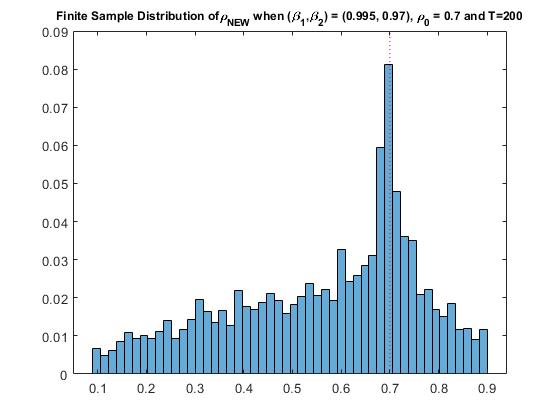}
\end{subfigure}
\begin{subfigure}{.48\textwidth}
  \centering
    \includegraphics[width=\linewidth]{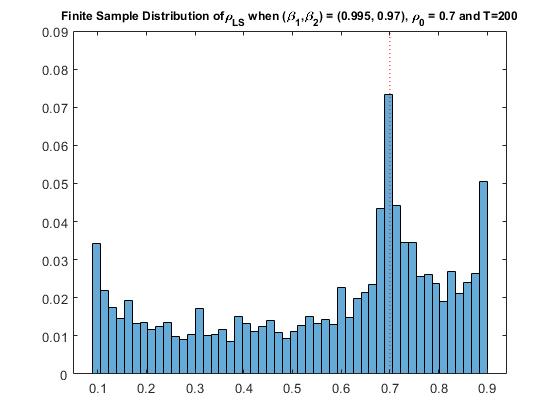}
\end{subfigure}
\caption{(Local-to-unity to stationary) Finite sample distributions of the new estimator (left) and the LS estimator (right) when the lag coefficient pre- and post-break are $(\beta_1,\beta_2) = (0.995,0.97)$ at break points $\rho_0 = 0.3, 0.5$, and 0.7, respectively. }
\label{fig:ar_fini_localtounit}
\end{figure}

\begin{table}[H]
\caption{Finite sample RMSE, bias, and the standard error of the new estimator and the LS estimator of the break point under the AR(1) model (\ref{eq:DGPinfill_ar}) with parameter values $(\beta_1,\beta_1,\rho_0)$ and $T = 200$. The number of replications is 5,000.}
\begin{center}  
\begin{tabular}{lllC{1.5cm}C{1.5cm}C{1.5cm}C{1.5cm}C{1.5cm}C{1.5cm}} 
\hline
\hline
& & & \multicolumn{2}{c}{RMSE} & \multicolumn{2}{c}{Bias} & \multicolumn{2}{c}{Standard error}\Tstrut\\
\cline{4-9}
$\beta_1$ & $\beta_2$ & $\rho_0$ & NEW & LS & NEW & LS & NEW & LS\Tstrut\\
\hline
\multirow{3}{*}{0.5} & \multirow{3}{*}{0.38} & 0.3 & 0.2627 & 0.3091 & 0.1821 & 0.1657 & 0.1893 & 0.2610\Tstrut\\
 &  & 0.5 & 0.1763 & 0.2452 & 0.0204 & 0.0282 & 0.1751 & 0.2436\\
 &  & 0.7 & 0.2285 & 0.2725 & -0.1379 & -0.1223 & 0.1822 & 0.2435\\
 \hline
\multirow{3}{*}{0.995} & \multirow{3}{*}{0.97} & 0.3 & 0.2369 & 0.2780 & 0.1319 & 0.1279 & 0.1967 & 0.2469\Tstrut\\
 &  & 0.5 & 0.1754 & 0.2328 & -0.0042 & -0.0047 & 0.1754 & 0.2327\\
 &  & 0.7 & 0.2375 & 0.2784 & -0.1358 & -0.1336 & 0.1948 & 0.2442\\
 \hline
\hline
\end{tabular}
\end{center}
\label{table:rmse_ar1_fini}
\end{table}

%%%%%%%%%%%%%%%%%%%%%%%%%%%%%%%%%%%%%%%%%%%%%%%%%%%%%%%%%%%%%%%%%%%
\section{Empirical Application}\label{sec:empirical}

In this section, I estimate the structural breaks in two empirical applications. I analyze the performance of the break point estimator by comparing it with the LS estimator and historical events documented in the literature. Furthermore, I show that the estimator is robust to trimming the sample period, whereas the LS estimator varies significantly depending on the trimmed sample. The first application is about the structural break in postwar U.S. real GDP growth rate. The second application is estimating the break date on the U.S. and UK stock returns using the return prediction model from \citet*{Paye2006}.

%%%%%%%%%%%%%%%%%%%%%%%%%%%%%%%%%%%%%%%%%%%%%%%
%%%%%%% U.S. real GDP growth rate; labour productivity
%%%%%%%%%%%%%%%%%%%%%%%%%%%%%%%%%%%%%%%%%%%%%%%
\subsection{U.S. real GDP growth rate}\label{subsec:empirical_gdpgrowth}

In the macroeconomics literature, shocks that affect the mean growth rate are often modeled as a one-time structural break because of their rare occurrence. However, existing estimation methods fail to capture the graphical evidence of postwar European and U.S. growth, slowing sometime in the 1970s. This is known as the ``productivity growth slowdown," which is widely hypothesized in macroeconomics literature. For instance, \citet*{Bai1998b} show that for the U.S., most test statistics reject the no-break hypothesis; however, the estimated confidence interval does not contain the slowdown in the 1970s. 

I estimate the structural break of an autoregressive model using the postwar quarterly U.S. real GDP growth rate. I use real GDP in chained dollars (base year 2012) data from the Bureau of Economic Analysis (BEA) website for the sample period 1947Q1-2018Q2, seasonally adjusted at annual rates. Annualized quarterly growth rates are calculated as 400 times the first differences of the natural logarithms of the levels data. I assume that log output has a stochastic trend with a drift and a finite-order representation. Following the \citet*{Eo2015} approach, I use \citeauthor{Kurozumi2011}'s \citeyearpar{Kurozumi2011} modified Bayesian information criterion (BIC) for lag selection to account for structural breaks. The highest lag order selected is 1 for output growth, given an upper bound of four lags and four breaks. The AR(1) model (\ref{eq:outputgrowth_ar1}) is estimated under three cases. The first case is a break only in the drift term ($\gamma \neq 0$, $\delta = 0$), the second case is a break only in the coefficient of lag, the ``propagation term" ($\gamma = 0$, $\delta \neq 0$), and finally, a break in constant and coefficient. 
\begin{equation}\label{eq:outputgrowth_ar1}
    \Delta y_{t} = \beta + \phi \Delta y_{t-1} + \bbone\{t > k_0 \} \left(\gamma + \delta \Delta y_{t-1}\right) + \varepsilon_{t}.
\end{equation}
I assume the error term $\{ \varepsilon_{t} \}$ is serially uncorrelated with mean zero disturbances. If a structural break occurs in constant and lag coefficients, the long-run growth rate of log output will change from $E[\Delta y_{t}] = \beta/(1-\phi)$ to $(\beta +\gamma)/(1-\phi-\delta)$ and the volatility of the growth rate will change from $Var[\Delta y_{t}] = \sigma^{2}/(1-\phi^{2})$ to $\sigma^{2}/(1-(\phi+\delta)^{2})$ at time $k_0$. 

Using the notations of model (\ref{eq:DGPmulti}), we have $x_t = (1,\, \Delta y_{t-1})'$, the dependent variable is $\Delta y_t$, and for each model $z_t = R'x_t$ is as follows:
\begin{itemize}
    \item M1: $R = (1,\, 0)'$, $z_t = 1$ 
    \item M2: $R = (0,\, 1)'$, $z_t = \Delta y_{t-1}$
    \item M3: $R = (1, \, 1)'$, $z_t = (1, \Delta y_{t-1})'$
\end{itemize}
The break point estimator $\hat{k}_{NEW}$ in (\ref{eq:newhatk_multi}) is obtained using the weight $\Omega_k = \omega_k^2 I_q$, where $\omega_k = (k/T(1-k/T))^{1/2}$ and $q = $ dim($z_t$)\footnote{For consistency of the break point estimator in a AR(1) model, we use $\Omega_k = \omega_k^2 I_q$, where $\omega_k$ is a function of $k/T$ only.}. I use the full sample 1947Q3-2018Q2 ($T=284$) to estimate the structural break date, then a shorter sub-sample to see if the break date estimates change. The optimization space of $k$ is the sample trimmed by fraction $\alpha = 0.1$ at both ends; the grid starts at 1954Q2 and ends at 2011Q1. The second and third columns of Table \ref{table:outputgrowth_ar1} show the break date estimates for the full sample. For M1 and M3, the two estimates are extremely different from each other, $\hat{k}_{NEW}$ is 1973Q1, whereas $\hat{k}_{LS}$ is 2000Q2. The break point estimates on the unit interval are approximately 0.36 and 0.75, respectively. Without any knowledge of historical events, one might think the finite sample properties of $\hat{\rho}_{LS}$ do not appear here, because it is not close to the boundaries of $0.1$ or $0.9$. 

However, the LS estimate switches to the boundary if we consider a sub-sample that is one decade shorter. Consider a sub-sample that ends at 2007Q1, with starting date 1947Q3, so that the search grid includes $\hat{k}_{LS}$ from all models. The LS estimate of M1 changes drastically to 1953Q1, which is the end of the search grid, $\hat{\rho}_{LS} = 0.1$. In contrast, my estimator under M1 provides the same break date estimate $\hat{k}_{NEW} = $1973Q1. For M3, both estimates change, so $\hat{k}_{NEW} = $1966Q1 and $\hat{k}_{LS} = $1958Q1. Compared to the full sample estimate, the change in the new estimate is 7 years, whereas it is over 40 years for the LS estimate. 

The break date estimate $\hat{k}_{NEW}=$1973Q1 under M1 corresponds to the productivity growth slowdown in the early 1970s. U.S. labor productivity experienced a slowdown in growth after the oil shock in 1973 (see \citet*{Perron1989} and \citet*{Hansen2001}). None of the models estimate a break date in the 1980s, which is known as ``the Great Moderation," referencing an empirical fact of a large reduction in the volatility of U.S. real GDP growth in 1984Q1, established by \citet*{Kim1999} and \citet*{McConnell2000}. I focus on events that affect the mean rather than the volatility of growth rate because the change in volatility is not a linear function of the change in the lag coefficient in model (\ref{eq:outputgrowth_ar1}).

\begin{table}[H]
\caption{Structural break date estimates of postwar U.S. real GDP growth rate in a AR(1) model. For each model, the first row is the break date estimate and the second row is the break point estimate (fraction in the corresponding sample).}
\begin{center}  
\begin{tabular}{L{1.3cm}cc|cc} 
\hline
\hline
& \multicolumn{2}{c|}{1947Q3-2018Q2} & \multicolumn{2}{c}{1947Q3-2007Q1}\Tstrut\Bstrut\\
 Model & NEW & LS & NEW & LS\Bstrut\\
\hline
 \multirow{2}{*}{M1}  & 1973Q1 & 2000Q2 & 1973Q1 & 1953Q1 \Tstrut\\
 & 0.36 & 0.75 & 0.43 & 0.10\\
 \hline
 \multirow{2}{*}{M2} & 1966Q1 & 1966Q1 & 1966Q1 & 1966Q1\Tstrut\\
  & 0.26 & 0.26 & 0.32 & 0.32\\
 \hline
 \multirow{2}{*}{M3} & 1973Q1 & 2000Q2 & 1966Q1 & 1958Q1 \Tstrut\\
 & 0.36 & 0.75 & 0.32 & 0.18\\ 
 \hline
\hline
\end{tabular}
\end{center}
\label{table:outputgrowth_ar1}
\end{table}

I check the sensitivity of the estimators to trimming the sample, by computing the two estimators for a total of 54 sub-samples, which end at different dates from 2005Q1 to 2018Q2 (the samples used to obtain estimates differ, not the fraction $\alpha$ of the optimization space within the sample). Because the sample is trimmed by one quarter each time, switching to a different estimate, which is farther, implies that the estimator is sensitive to trimming, rather than suggesting multiple breaks\footnote{It is likely that multiple structural breaks exist in the output growth rate because we are considering a sample that is over 70 years. Estimates that vary depending on the sub-sample could be evidence of more than one break in the sample period. The break date estimate $\hat{k}_{LS} =$ 2000Q2 can align with the tech bubble, also known as the dot-com crash in 2000. In relation to business cycles, $\hat{k}_{LS} = $1953Q1 and 1958Q1 are both recession in 1953, with the end of the Korean war. Under M2, both estimates from the full sample are 1966Q1, and the closest historical event that is likely to affect the output growth rate is the Vietnam war.}. Estimating multiple structural breaks using the weighting scheme is beyond the scope of this study. For LS estimation of multiple breaks see \citet*{Bai1998a} and \citet{Bai1998b}.

Table \ref{table:outputgrowth_ar1_subsample} provides the number of sub-samples that have the same break date estimates; the entries are fractions of the number of sub-samples out of 54 sub-samples. For M1 and M3, there are sub-samples in which the LS estimates are at the boundaries $\hat{\rho}_{LS} = 0.1$ or $0.9$. In contrast our break point estimates are either mid 1960s or early 1970s, which are in the fraction interval $\hat{\rho} \in [0.2, 0.5]$.

In short, estimating a structural break of postwar U.S. real GDP growth rate using our estimation method, provides evidence of a break occurring in 1973Q1, which corresponds to the productivity growth slowdown period. However, the LS estimates a break occurs in 2000 or 1953, depending on the time interval. Break date estimates are obtained for sub-samples with end dates 2005Q1 to 2018Q2 for both methods; the LS estimates vary considerably, with $\hat{\rho}_{LS}$ near 0.1 and 0.9 for almost 40\% of the sub-samples considered under M1. In contrast, my estimates are 1966Q1 or 1973Q1 for all sub-samples and models. This suggests that the difference in LS estimates, depending on the sample period, are due to their finite sample behavior (tri-modality) rather than evidence of multiple structural breaks.

\begin{table}[H]
\caption{Structural break date estimates of postwar U.S. real GDP growth rate in an AR(1) model using 54 sub-samples. The entries are the fractions of the number of sub-samples that have break date estimates corresponding to the first column. The second column is the interval of break point estimates that depend on sub-sample size. The start date is 1947Q3 and the end dates of sub-samples change across 2005Q1 to 2018Q2.}
\begin{center}  
\begin{tabular}{L{2cm}C{2.3cm}C{1cm}C{1cm}C{1cm}C{1cm}C{1cm}C{1cm}} 
\hline
\hline
  &  & \multicolumn{2}{c}{M1} & \multicolumn{2}{c}{M2} & \multicolumn{2}{c}{M3}\Tstrut\\
 \cline{3-8}
Break date & $\hat{\rho}$ interval & NEW & LS & NEW & LS & NEW & LS\Tstrut\\
\hline
1953Q1 & 0.10 & & 0.17 &  &  & & \Tstrut\\
1958Q1 & [0.17, 0.19] & &  &  &  & & 0.22\\
1966Q1 & [0.26, 0.33] & & 0.06 & 1 & 1 & 0.30 & 0.09\\
1973Q1 & [0.36, 0.45] & 1 & 0.05 &  &  & 0.70 & \\
2000Q2 & [0.74, 0.86] & & 0.52 &  &  & & 0.48\\
2006Q1 & [0.85, 0.90] & & 0.20 &  &  & & 0.21\Bstrut \\
 \hline
\hline
\end{tabular}
\end{center}
\label{table:outputgrowth_ar1_subsample}
\end{table}

\smallskip

%%%%%%%%%%%%%%%%%%%%%%%%%%%%%%%%%%%%%%%%%%%%%%%%%%%%%%%%%%%%%%%%%%%%%%%
\subsection{Stock return prediction models}\label{subsec:empirical_stock}

\citet*{Paye2006} studied the instability in models of ex-post predictable components in stock returns by examining structural breaks in the coefficients of state variables. The regression model (\ref{eq:stockreturnmodel}) is specified with four state variables as follows: the lagged dividend yield, short-term interest rate, term spread, and default premium. The model allows for all coefficients to change because no strong reason exists to believe that the coefficient on any of the regressors should be immune from shifts. The multivariate model with a one-time structural break at $k$ with $t=1,\ldots,T$ is
\begin{align}
\label{eq:stockreturnmodel}
    Ret_t & = \beta_{0}+\beta_{1}Div_{t-1} + \beta_{2}Tbill_{t-1} + \beta_{3}Spread_{t-1} + \beta_{4}Def_{t-1} \\
\nonumber    
     & \hspace{0.7cm} + \bbone\{t>k\}\left(\delta_0 + \delta_{1}Div_{t-1} + \delta_{2}Tbill_{t-1} + \delta_{3}Spread_{t-1} + \delta_{4}Def_{t-1}\right) + \varepsilon_t,
\end{align}
where $Ret_t$ represents the excess return for the international index in question during month $t$, $Div_{t-1}$ is the lagged dividend yield, $Tbill_{t-1}$ is the lagged local country short interest rate, $Spread_{t-1}$ is the lagged local country spread, and $Def_{t-1}$ is the lagged US default premium. From the notation of model (\ref{eq:DGPmulti}), $y_t = Ret_t$ and for the multivariate model, $x_t = z_t = (1,Div_{t-1},Tbill_{t-1},$ $Spread_{t-1},Def_{t-1})$. For the univariate model with dividend yield $x_t = z_t = (1,Div_{t-1})$, which is defined analogously for other univariate models. The weight matrix is $w_k = T^{-1}Z_k'MZ_k$, where $Z_k = (0,\ldots,0,z_{k+1},\ldots,z_{T})'$ and $M = I-X(X'X)^{-1}X'$. Following the approach of \citet*{Paye2006}, I examine univariate models to facilitate the interpretation of coefficients, in addition to the multivariate model (\ref{eq:stockreturnmodel}). 

I collected data from Global Financial Data and Federal Reserve Economic Data (FRED). The indices of the total return and dividend yield series are the S\&P 500 for the U.S. and the Financial Times Stock Exchange (FTSE) All-share for the UK. The dividend yield is expressed as an annual rate and constructed as the sum of dividends over the preceding 12 months, divided by the current price. For both countries, the three-month Treasury bill (T-bill) rate is used as a measure of the short-term interest rate and the 20-year government bond yield is the measure of the long-term interest rate. Excess returns are the total return stocks in the local currency less the total return on T-bills. The term spread is constructed as the difference between the long- and the short-term local country interest rate. The U.S. default premium is the differences in yields between Moody's Baa and Aaa rated bonds. The search grid is obtained by trimming each sample period by fraction $\alpha=0.15$ (which is equivalent to the trimming window of \citet*{Paye2006}). For the full sample, the search grid is 1960:2-1996:3, and for the sub-sample it is 1975:1-1998:10.

Under the univariate model, with the lagged dividend yield as a single forecasting regressor, the LS estimate of the break point for the S\&P 500 is close to the boundary of the search grid. \citet*{Paye2006} note that the NYSE or S\&P 500 indices have the same estimated break date when the trimming window is shortened, and thus, the discrepancy is not the sole explanation for the timing of the break. However, it is likely that estimates are near the boundaries because of the finite sample behavior of the LS estimator. I check whether the new estimator provides a different break date estimate under model (\ref{eq:stockreturnmodel}), using data similar to the first dataset from \citet*{Paye2006}, which is monthly data on the U.S. and the UK stock returns from 1952:7 to 2003:12. For comparison, I also estimate the break using a shorter period 1970:1-2003:12, which is equivalent to the sample period of their second dataset.

Table \ref{table:stockreturn_us} provides estimates of the two samples using the S\&P 500 index. One notable feature is that the LS estimates that a break occurred in December 1994, with break point $\hat{\rho}_{LS} = 0.83$, whereas my method estimates a break in the mid-1980s and $\hat{\rho}_{NEW} = 0.62$. Although the LS estimate is close to the boundaries of the grid, it gives the same break estimate in the sub-sample. This suggests that a break may have occurred multiple times. \citet*{Paye2006} use the \citet*{Bai1998a} method and find that two structural breaks occur in the return model (\ref{eq:stockreturnmodel}), using the S\&P 500, where each break occurs at 1987:7 and 1995:3. They note that the break in 1987 appears to be an isolated break, not appearing in other international markets. These two break date estimates are similar to estimates in Table \ref{table:stockreturn_us}, which assume a one-time structural break.

\begin{table}[H]
\caption{Structural break date estimates of the U.S. stock return (S\&P 500) prediction model for samples 1952:7-2003:12 and 1970:1-2003:12. For each model, the first row is the break date estimate and the second row is the break point estimate (fraction in the corresponding sample). }
\begin{center}  
\begin{tabular}{lcc|cc} 
\hline
\hline
 & \multicolumn{2}{c|}{1952:7-2003:12} & \multicolumn{2}{c}{1970:1-2003:12} \Tstrut\\
 Model & NEW & LS  & NEW & LS\\
\hline
\multirow{2}{*}{Multi.} & 1984:8 & 1994:12 & 1982:8 & 1994:12\Tstrut\\
 & 0.62 & 0.83 & 0.37 & 0.74\\
 \hline
 \multirow{2}{*}{Div. yield} & 1982:8 & 1995:1 & 1982:8 & 1996:9 \Tstrut\\
  & 0.59 & 0.83 & 0.37 & 0.79\\
\hline
 \multirow{2}{*}{T-bill} & 1974:10 & 1974:10 & 1982:8 & 1975:1 \Tstrut\\
  & 0.43 & 0.43 & 0.37 & 0.15\\
\hline 
 \multirow{2}{*}{Spread} & 1983:5 & 1976:2 & 1987:9 & 1976:2 \Tstrut\\
  & 0.60 & 0.46 & 0.52 & 0.18\\
 \hline
 \multirow{2}{*}{Def.prem.}  & 1968:12 & 1965:11 & 1982:8 & 1975:7 \Tstrut\\
  & 0.32 & 0.26 & 0.37 & 0.16\\
  \hline
\hline
\end{tabular}
\end{center}
\label{table:stockreturn_us}
\end{table}

An alternative explanation of the break in the early 1980s is that the estimation method captures a change in the individual state variable itself rather than the coefficient of the prediction model (\ref{eq:stockreturnmodel}), because the noisy nature of stock market returns makes it extremely difficult to detect a break. For instance, the estimate $\hat{\beta}_1$ could be capturing noise caused by the movement in $Div_{t-1}$ (see Figure \ref{fig:us_tbill_divyield} in Appendix C).

For UK stock returns, both methods obtain a break date estimate that is (or close to) 1975:1 under all models and sample periods. This is different from the result using the S\&P 500 index series, because the excess return for the FTSE All-share index increases nearly 10 standard deviations from 1975:1 to 1975:2 \footnote{For the sample period 1952:7-2003:12, the mean excess return of the FTSE All-share index is 0.5949 and the standard deviation is 5.4890. At $t = $1975:1, the excess return $Ret_t = 0.4556$ and at $t = $1975:2 we have $Ret_t = 53.2187$; therefore, the change is approximately 9.6 standard deviations.}. Hence, the change in excess returns is large enough for the LS to detect the break point appropriately. \citet*{Paye2006} relate the break in the mid-1970s to the large macroeconomic shocks reflecting large oil price increases; breaks in the underlying economic fundamentals process can explain breaks in financial return models. If this is the case, then the break magnitude is large enough for both methods to accurately estimate the break date 1975:1.

\begin{table}[H]
\caption{Structural break date estimates of the UK (FTSE) stock return prediction model for samples 1952:7-2003:12, and 1970:1-2003:12. For each model, the first row is the break date estimate and the second row is the break point estimate (fraction in the corresponding sample).}
\begin{center}  
\begin{tabular}{lcc|cc}
\hline
\hline
& \multicolumn{2}{c|}{1952:7-2003:12} & \multicolumn{2}{c}{1970:1-2003:12} \Tstrut\\
 Model & NEW & LS  & NEW & LS\\
\hline
 \multirow{2}{*}{Multi.} & 1975:1 & 1975:1 & 1975:1 & 1975:1\Tstrut\\
& 0.44 & 0.44 & 0.15 & 0.15\\
\hline
 \multirow{2}{*}{Div. yield} & 1975:1 & 1975:1 & 1975:1 & 1975:1\Tstrut\\
 & 0.44 & 0.44 &  0.15 & 0.15\\
\hline
 \multirow{2}{*}{T-bill} & 1975:1 & 1974:12 & 1975:1 & 1975:1 \Tstrut\\
 & 0.44 & 0.29 & 0.15 & 0.15\\
\hline
 \multirow{2}{*}{Spread} &  1975:1 & 1975:6 & 1975:1 & 1975:3 \Tstrut\\
& 0.44 & 0.45 & 0.15 & 0.15\\
 \hline
 \multirow{2}{*}{Def.prem.} & 1979:5 & 1975:6 & 1975:1 & 1975:3 \Tstrut\\
  & 0.52 & 0.45 & 0.52 & 0.15\\
 \hline
\hline
\end{tabular}
\end{center}
\label{table:stockreturn_uk}
\end{table}

%%%%%%%%%%%%%%%%%%%%%%%%%%%%%%%%%%%%%%%%%%%%%%%%%%%%%%%%%%%%%%%%%%%%%%%%%%%%%
\section{Conclusion}\label{sec:conclusion}

This study provides an estimation method of the structural break point in multivariate linear regression models, when a one-time break occurs in a subset of (or all) coefficients. In particular, this study focuses on break magnitudes that are empirically relevant. In practice, it is likely that the shift in parameters is small in a statistical sense. The LS estimation widely used in the literature fails to accurately estimate the break point under small break magnitudes, which motivates us to construct the estimation method in this study.

I construct a weight function on the sample period normalized to a unit interval, which imposes small weights on the LS objective for potential break points with large estimation uncertainty. The break point estimator is the argmax of the objective function that is equal to the LS objective multiplied by a weight function. The break point estimator is consistent under regularity conditions on a general weight function, with the same rate of convergence as the LS estimator from \citet*{Bai1997}. The limit distribution under a small break magnitude derives under an in-fill asymptotic framework, following the approach by \citeauthor{Jiang2017} (\citeyear{Jiang2017}, \citeyear{Jiang2018}). For a structural break in a stationary linear process with a small break magnitude (inside the local $T^{-1/2}$ neighborhood of zero), the asymptotic distribution of the new estimator explicitly depends on the weight function. The limit distribution is also derived for a break in a local-to-unit root process, assuming the break magnitude is $O(T^{-1})$. Monte Carlo simulation results show that for a small break, the break point estimator reduces the RMSE compared to the LS estimator for all parameter values considered. 

This study provides two empirical applications as follows: structural breaks on the U.S. real GDP growth and the U.S. and UK stock return prediction models. My break point estimator is robust to trimming of the sample, in contrast to the LS. In particular, my method estimates the break date 1973Q1 in U.S. real GDP growth rates, which LS estimation has failed to confirm. In macroeconomics literature, the ``productivity growth slowdown" in the early 1970s is a widely known empirical fact.

In short, this study provides an alternative estimation method that estimates the timing of a structural break in linear regression models under empirically relevant break magnitudes. My estimator shows a uni-modal finite sample distribution under statistically small break magnitudes. To my knowledge, this is the first study to widen the class of break point estimators by generalizing least-squares. I provide theoretical results of the consistency of the estimator and an asymptotic distribution that represents finite sample behavior. If the break magnitude is small, my estimator outperforms the LS estimation in terms of RMSE. Thus, under statistically small but empirically relevant breaks, the estimator described in this study provides reliable inferences of the change point in models. The estimation method can be generalized to estimate multiple structural breaks, which is a topic for future research.

%% END MATTER
% \printindex %% Uncomment to display the index
% \nocite{}  %% Put any references that you want to include in the bib 
%               but haven't cited in the braces.
%\printbibliography[ heading=bibintoc, title={References}] 
\bibliographystyle{itaxpf}
%\bibliographystyle{erae} 
%                              There are many others.
%\setlength{\bibleftmargin}{0.25in}  % indent each item
%\setlength{\bibindent}{-\bibleftmargin}  % unindent the first line
%\def\baselinestretch{1.0}  % force single spacing
%\setlength{\bibitemsep}{0.16in}  % add extra space between items
\bibliography{main}  %% This looks for the bibliography in template.bib 
%                          which should be formatted as a bibtex file.
%                          and needs to be separately compiled into a bbl file.
%\singlespace  %to force bibilography environment to use single spacing for each entry 
              %double spacing between entries remains

\begin{appendix}

\section*{Appendix A}

\setcounter{lemma}{0}
\renewcommand{\thelemma}{A.\arabic{lemma}}

\setcounter{equation}{0}
\renewcommand{\theequation}{A.\arabic{equation}}

%Moreover, I suggest a representative weight function that is motivated by the Fisher information under a Gaussian assumption. Under this particular weight function, the new break point estimator is asymptotically equivalent to the mode of the Bayesian posterior distribution\footnote{Also known as the maximum a-posteriori probability (MAP) estimator.} when the prior depends on the Fisher information. 

In a Bayesian perspective, the weight function can be interpreted as a prior belief on parameters $\delta$ and $\rho$. Suppose the prior distribution of the break magnitude $\delta$ conditional on $\rho$ is normal, with mean zero and variance $\nu^2$, $\delta \vert \rho \sim N(0, \nu^2)$. Denote the prior distribution of the break point as $f(\rho)$ and let $\by = (y_1, \ldots, y_T)$. We assume Gaussian disturbances, $\varepsilon_t \overset{i.i.d.}{\sim} N(0,\sigma^2)$. The joint distribution of $\delta$, $\rho$, and $\by$ is
\begin{equation*}
    f(\by,\delta,\rho) = (2\pi)^{-1}(\nu^2 I_T^{-1})^{-1/2}\exp\left[-\frac{1}{2}\left\{\frac{\delta^2}{\nu^2} + \frac{(\bar{y}_k^{*}-\bar{y}_k -\delta)^2}{I_T^{-1}}\right\}\right]f(\rho),
\end{equation*}
where $I_T = I_T(\delta\vert\rho)$ is the Fisher information of $\delta$ conditional on $\rho$ and $l_T(\delta \,\vert\,\rho)$ is the conditional log-likelihood function.
\begin{equation*}
        I_T := E\left[-\frac{\partial^2 l_T(\delta \,\vert\,\rho)}{\partial \delta\partial \delta'} \right] = \sigma^{-2}T\rho(1-\rho).
\end{equation*}
The posterior distribution of $\rho$ is proportional to $f(\by,\delta,\rho)$, integrated with respect to $\delta$.
\begin{equation*}
    f(\rho \vert \by) \propto f(\rho)(I_T^{-1} + \nu^2)^{-1/2}\exp\left[-\frac{(\bar{y}_k^{*}-\bar{y}_k)^2}{2(I_T^{-1}+\nu^2)}\right].
\end{equation*}
We assume $\rho$ is bounded away from $\{0,1\}$ so that $I_T^{-1} = o(1)$, and for simplicity, assume $\sigma^2 = 1$. Suppose $\nu = T^{-1} I_T = \rho(1-\rho)$ and $f(\rho) \propto \nu$. In the limit, the posterior distribution is proportional to 
\begin{equation}\label{eq:posterior_limit}
    f(\rho \vert \by) \propto \exp\left[-\frac{(\bar{y}_k^{*}-\bar{y}_k)^2}{2\rho^2 (1-\rho)^2}\right].
\end{equation}
Given the data $\by$, the argmax function of the monotone transformation of the likelihood (\ref{eq:posterior_limit}) is asymptotically equivalent to the argmax of $Q_T(k)^2 = \omega_k^2 V_T(k)^2$, with $\omega_k = (k/T(1-k/T))^{1/2}$. That is, if the weight function is proportional to the square root of the Fisher information of $\delta$ conditional on $\rho$, the mode of the Bayesian posterior distribution (known as the maximum a-posteriori probability (MAP) estimator) is asymptotically equivalent to the break point estimator in (\ref{Def:newestimator_breakmean}).

Note that the Fisher information is interpreted as a way to measure the amount of information about the unknown parameter $\delta$, given $\rho$. Given two different values $\rho_1 \neq \rho_2$, the inequality $I_T(\delta\vert\rho_1) > I_T(\delta\vert\rho_2)$ reflects the fact that observations carry more information on the break magnitude if a break occurs at $\rho_1$, compared to $\rho_2$. In other words, there is more information of a structural break occurring at $\rho_1$ than at $\rho_2$. If a break occurs with high probability, its magnitude is likely to be far from zero; if it is not likely, then $\delta$ is close to zero. If we scale the Fisher information by sample size ($I := T^{-1}I_T$), then the prior distribution $\delta \vert\rho \sim N(0, I)$ incorporates the amount of information of a break at some fixed $\rho$. A large $I$ implies the variance of the prior distribution is large, thus $\delta$ is more spread out from zero and has large magnitude with high probability. In the opposite case, if there is less information, $\delta$ is centered toward mean zero and the break magnitude is likely to be small. Similarly, a prior belief of $\rho$ can be expressed using the Fisher information ($f(\rho) \propto I$ is equivalent to a Beta distribution with shape parameters $(2,2)$).

%%%%%%%%%%%%%%%%%%%%%%%%%%%%%%%%%%%%%%%%%%%%%%%%%%%%%%%%%%%%%%%%%%%%%%%%%
\section*{Appendix B}

\noindent \textbf{Proof of Theorem \ref{Thm:rateofconvergence}} 
\begin{proof}
By rearranging terms we have 
\begin{equation}\label{eq:QinGandH}
    Q_{T}(k)^2-Q_T(k_0)^2 = -\vert k_0 -k\vert G_{T}(k) + H_{T}(k),
\end{equation}
where $G_T(k)$ and $H_{T}(k)$ are defined as follows:
\begin{align}
\nonumber
    G_{T}(k) & := \frac{1}{\vert k_0 - k \vert} \delta_T'\left[(Z_0'MZ_0)^{1/2}\Omega_{k_0}(Z_0'MZ_0)^{1/2} \right. \\
\label{eq:Gk}
    & \hspace{0.7cm} \left. - (Z_0'MZ_k)(Z_k'MZ_k)^{-1/2}\Omega_k(Z_k'MZ_k)^{-1/2}(Z_k'MZ_0) \right]\delta_T
\end{align}
\begin{align}
\nonumber
    H_{T}(k) & := \varepsilon'MZ_k(Z_k'MZ_k)^{-1/2}\Omega_k(Z_k'MZ_k)^{-1/2}Z_k'M\varepsilon\\
\nonumber
    & \hspace{0.7cm} - \varepsilon'MZ_0(Z_0'MZ_0)^{-1/2}\Omega_{k_0}(Z_0'MZ_0)^{-1/2}Z_0'M\varepsilon \\
\label{eq:Hk}
    & \hspace{0.7cm} + 2\delta_T'(Z_0'MZ_k)(Z_k'MZ_k)^{-1/2}\Omega_k (Z_k'MZ_k)^{-1/2}Z_k'M\varepsilon \\
\nonumber    
    & \hspace{0.7cm} - 2\delta_T'(Z_0'MZ_0)^{1/2}\Omega_{k_0} (Z_0'MZ_0)^{-1/2}Z_0'M\varepsilon.
\end{align}
Lemma \ref{Lem:HajekRenyi} is equivalent to \citeauthor{Bai1997}'s (\citeyear{Bai1997}) lemma A.3, which is the generalized H\'{a}jek-R\'{e}nyi inequality for martingale differences to mixingales. For the proof see \citet*{Bai1998a}. 

\smallskip
%%%%%%%%%% Lemma 1: Inverse of Gk function is bounded
\begin{lemma}\label{Lem:Gkinversebound}
Under Assumptions \ref{A:multimodel} and \ref{A:Omegak}, for every $\epsilon >0$, there exists $\lambda >0$ and  $C<\infty$ such that $\inf_{\vert k-k_0 \vert > C\norm{\delta_T}^{-2}} G_T(k) \geq \lambda \norm{\delta_T}^2$, with probability at least $1-\epsilon$.
\end{lemma}

\smallskip

\begin{lemma}\label{Lem:HajekRenyi}
Under Assumption \ref{A:multimodel}, there exist a $M < \infty$ such that for every $c>0$ and $m>0$,
\begin{equation*}
    P\left(\sup_{m \leq k \leq T}\frac{1}{k}\norm{ \sum_{t=1}^{k}z_{t}\varepsilon_{t}} > c \right) \leq \frac{M}{c^2m}.
\end{equation*}
\end{lemma}

\smallskip
%%%%%%%%%% Lemma: Consistency w/o rate of convergence
\begin{lemma}\label{Lem:consistent}
Under Assumptions \ref{A:multimodel} and \ref{A:Omegak}, suppose $\delta_T$ is fixed or shrinking toward zero such that Assumption \ref{A:consistency_delta} is satisfied. Then the break point estimator $\hat{\rho}$ in (\ref{eq:newhatk_multi}) is consistent.
That is, for every $\epsilon >0$ and $\eta > 0$, there exists $T_0 > 0$ such that when $T > T_0$,
\begin{equation*}
    P\left(\vert \hat{\rho}-\rho_0 \vert  > \eta \right) < \epsilon.
\end{equation*}
Moreover, $\vert \hat{\rho}-\rho_0 \vert = O_p\left(T^{-1/2}\norm{\delta_T}\sqrt{\ln T}\right)$. 
\end{lemma}

\bigskip
The rate of convergence of the break point estimator $\hat{\rho}$ in (\ref{eq:newhatk_multi}) can be improved from lemma \ref{Lem:consistent}. For a fixed $\epsilon >0$ and $\eta >0$, inequality (\ref{eq:rate_Teta}) holds for any true break point $\rho_0 \in [\alpha,1-\alpha]$, when $T$ is large. 
\begin{equation}\label{eq:rate_Teta}
    P\left(\underset{\vert k-k_0\vert > T\eta}{\sup} Q_T(k)^2 \geq Q_T(k_0)^2\right) < \epsilon.
\end{equation}
This is equivalent to lemma \ref{Lem:consistent} because given the estimator $\hat{k}$, $Q_{T}(\hat{k})^2 -Q_{T}(k_0)^2 \geq 0$ by definition. This implies that to prove the improved rate of convergence $O_p\left(T^{-1}\norm{\delta_T}^{-2}\right)$, it is sufficient to show that for all $\epsilon > 0$, there exists a finite $C > 0$ so that for all $T>T_\epsilon$,
\begin{equation*}
    P\left(\sup_{k \in K_{T}(C)}Q_{T}(k)^2 \geq Q_{T}(k_0)^2\right) < \epsilon.
\end{equation*}
Here, $K_{T}(C) = \left\{k: \vert k-k_0 \vert > C\norm{\delta_T}^{-2},  \vert k-k_0 \vert \leq T\eta\right\}$ for some small fraction $\eta$. From identity (\ref{eq:QinGandH}), $Q_{T}(k)^2\geq Q_{T}(k_0)^2$ is equivalent to $H_{T}(k)/\vert k-k_0 \vert \geq  G_{T}(k)$. From lemma \ref{Lem:Gkinversebound}, it is sufficient to prove that 
%Note that for a given $\epsilon > 0$, we can choose $C$ to be a function of $D_\epsilon$ such as $C = D_\epsilon^2$. This is because for every $\epsilon$, there exists a finite $D_\epsilon$ and $T_\epsilon>0$ so that (\ref{eq:rate_lnT}) is satisfied for all $T>T_\epsilon$. 

%Given $\epsilon > 0$ and a constant $C$, $K_T(C)$ must be a non-empty set. This implies that when $T$ is large, $K_T(C)$ is non-empty  because the upper bound of $D_\epsilon$ increases with $T$. Thus, for any given $\epsilon>0$ we can choose $C$ to be a large value that depends on $D_\epsilon$.

\begin{equation}\label{eq:Hkbounded}
    P\left(\sup_{k \in K_T(C)}\left\vert \frac{H_{T}(k)}{k_0-k}\right\vert > \lambda\norm{\delta_T}^2 \right) < \epsilon.
\end{equation}
We use the expression $Z_0 = Z_k -Z_{\Delta}\text{sgn}(k_0-k)$ to rewrite the third and fourth terms of $H_T(k)$ given in (\ref{eq:Hk}) as
\begin{align}
\nonumber
    2\delta_T'&\left[(Z_0'MZ_k)(Z_k'MZ_k)^{-1/2}\Omega_k (Z_k'MZ_k)^{-1/2}Z_k'M\varepsilon - (Z_0'MZ_0)^{1/2}\Omega_{k_0} (Z_0'MZ_0)^{-1/2}Z_0'M\varepsilon \right]\\
\nonumber    
    & = 2\delta_T'\left[(Z_k'MZ_k)^{1/2}\Omega_k(Z_k'MZ_k)^{-1/2}Z_k'M\varepsilon-(Z_0'MZ_0)^{1/2}\Omega_{k_0}(Z_0'MZ_0)^{-1/2}Z_k'M\varepsilon\right]\\
\label{eq:Hk_thirdfourth} 
    & \hspace{0.5cm} + 2\delta_T'(Z_0'MZ_0)^{1/2}\Omega_{k_0}(Z_0'MZ_0)^{-1/2}(Z_\Delta'M\varepsilon)\,  \text{sgn}(k_0-k) \\
\nonumber
    & \hspace{0.5cm}-2\delta_T'(Z_{\Delta}'MZ_k)(Z_k'MZ_k)^{-1/2}\Omega_k(Z_k'MZ_k)^{-1/2}(Z_k'M\varepsilon)\, \text{sgn}(k_0-k).  
\end{align}
Note that for nonsingular matrices $S$ and $A$ with bounded norms, $SAS^{-1} = A+o_p(1)$. Also, we have $(Z_k'MZ_k)^{-1}Z_k'M\varepsilon = O_{p}(T^{-1/2})$ and $(Z_0'MZ_0)^{-1}(Z_k'MZ_k) = O_p(1)$ uniformly on $K_T(C)$. We use this to find the order of the first line of the right-side in (\ref{eq:Hk_thirdfourth}). 
\begin{align*}
    & \norm{2\delta_T'\left\{(Z_k'MZ_k)^{1/2}\Omega_k(Z_k'MZ_k)^{-1/2}Z_k'M\varepsilon-(Z_0'MZ_0)^{1/2}\Omega_{k_0}(Z_0'MZ_0)^{-1/2}Z_k'M\varepsilon\right\}}\\
    & \hspace{0.5cm} \leq \norm{2\delta_T'\left\{(Z_k'MZ_k)^{1/2}\Omega_k(Z_k'MZ_k)^{1/2}-(Z_0'MZ_0)^{1/2}\Omega_{k_0}(Z_0'MZ_0)^{1/2}O_p(1)\right\}}\\
    & \hspace{2cm}\times \norm{(Z_k'MZ_k)^{-1}Z_k'M\varepsilon} + o_p(1)\\
    & \hspace{0.5cm} \leq \norm{2\delta_T}\norm{Z_k'MZ_k\Omega_k - Z_0'MZ_0\Omega_{k_0}}\, O_{p}(T^{-1/2}) + o_p(1)\\
    & \hspace{0.5cm} = \norm{2\delta_T}\norm{(Z_k'MZ_k-Z_0'MZ_0)\Omega_k - Z_0'MZ_0(\Omega_{k_0}-\Omega_k)}\, O_{p}(T^{-1/2}) + o_p(1).
\end{align*}
Then the second norm can be rearranged by subtracting and adding $Z_0'MZ_k$ to the term $(Z_k'MZ_k - Z_0'MZ_0)$ and Assumption \ref{A:Omegak}.
\begin{align}
\nonumber
    Z_k'MZ_k - & Z_0'MZ_0 \\
\nonumber    
    & = \begin{cases}
    R'[X_{\Delta}'X_{\Delta}(X'X)^{-1}(X'X-X_k'X_k)-X_0'X_0(X'X)^{-1}X_{\Delta}'X_{\Delta}]R &  \text{ if } k \leq k_0 \\
    R'[X_{\Delta}'X_{\Delta}(X'X)^{-1}X_k'X_k-(X'X-X_0'X_0)(X'X)^{-1}X_{\Delta}'X_{\Delta}]R & \text{ if } k > k_0 
    \end{cases}\\
\label{eq:ZkminusZ0}
    & = \vert k_0-k\vert O_p(1),\\
\nonumber    
    \Omega_{k_0} - \Omega_k & = \vert k_0-k\vert T^{-1} O_p(1)
\end{align}
The norm $\norm{(k_0-k)^{-1}X_{\Delta}'X_{\Delta}}$ is bounded by assumption, hence the first line of (\ref{eq:Hk_thirdfourth}) has order $\vert k_0-k\vert \norm{\delta_T}O_{p}(T^{-1/2})$. The second and third lines of (\ref{eq:Hk_thirdfourth}) are 
\begin{align*}
    2\delta_T'(Z_0'MZ_0)^{1/2}& \Omega_{k_0}(Z_0'MZ_0)^{-1/2}(Z_\Delta'M\varepsilon)\,\text{sgn}(k_0-k) \\
    & = 2\delta_T'\Omega_{k_0}(Z_\Delta'M\varepsilon)\, \text{sgn}(k_0-k)+o_p(1)\\
    & = 2\delta_T'\Omega_{k_0}(Z_\Delta'\varepsilon-Z_{\Delta}'X(X'X)^{-1}X'\varepsilon) \,\text{sgn}(k_0-k)+o_p(1)\\
    & = 2\delta_T'\Omega_{k_0}Z_\Delta'\varepsilon\,\text{sgn}(k_0-k)+\vert k_0-k\vert T^{-1/2}\norm{\delta_T}O_p(1) + o_p(1),\\
    -2\delta_T' (Z_{\Delta}'MZ_k)&(Z_k'MZ_k)^{-1/2}\Omega_k(Z_k'MZ_k)^{-1/2}(Z_k'M\varepsilon)\,\text{sgn}(k_0-k)\\
    & = -2\delta_T'(Z_{\Delta}'MZ_k)\Omega_k(Z_k'MZ_k)^{-1}(Z_k'M\varepsilon)\,\text{sgn}(k_0-k)+o_{p}(1)\\
    & = \vert k_0-k\vert T^{-1/2}\norm{\delta_T}O_{p}(1).
\end{align*}
The first and second terms of $H_{T}(k)$ in (\ref{eq:Hk}) are $O_{p}(1)$ uniformly in $K_T(C)$, under Assumptions \ref{A:multimodel} and \ref{A:Omegak}. Therefore, $H_{T}(k)$ divided by $\vert k_0-k \vert$ is
\begin{equation}\label{eq:Hkdivided}
    \frac{H_T(k)}{\vert k_0-k \vert} = 2\delta_T'\Omega_{k_0}\frac{1}{\vert k_0-k \vert}Z_{\Delta}'\varepsilon \,\text{sgn}(k_0-k) + T^{-1/2}\norm{\delta_T}O_{p}(1) + \frac{O_{p}(1)}{\vert k_0-k \vert}.
\end{equation}
Now we can prove (\ref{eq:Hkbounded}) using the expression (\ref{eq:Hkdivided}). Let $1/\norm{\Omega_{k_0}} = A$ where $A < \infty$ by Assumption \ref{A:Omegak}. Without loss of generality, consider the case $k < k_0$. The first term of (\ref{eq:Hkdivided}) is bounded by lemma \ref{Lem:HajekRenyi}.

\begin{align*}
    P& \left(\sup_{k \in K_T(C)}\norm{2\delta_T'\Omega_{k_0}\frac{1}{k_0-k}\sum_{t=k+1}^{k_0}z_{t}\varepsilon_t} > \frac{\lambda\norm{\delta_T}^2}{3} \right)\\
    & \hspace{1.5cm} \leq P\left(\sup_{k_0-k \geq C\norm{\delta_T}^{-2}}\norm{\frac{1}{k_0-k}\sum_{t=k+1}^{k_0}z_{t}\varepsilon_t} > \frac{\lambda A\norm{\delta_T}}{6} \right)\\
    & \hspace{1.5cm} \leq M\left(\frac{\lambda A \norm{\delta_T}}{6}\right)^{-2}\frac{1}{C\norm{\delta_T}^{-2}} \\
    & \hspace{1.5cm} = \frac{36M}{\lambda^2 A^2 C} < \frac{\epsilon}{3}
\end{align*}
The probability is negligible for large $T$ because we can choose a large $C$ value accordingly. For any $\epsilon>0$ and $\eta>0$, we proved that the probability in (\ref{eq:rate_Teta}) is negligible for large $T$. Thus we can choose $C$ such that $K_T(C)$ is non-empty and the inequality above is satisfied for all $\epsilon > 0$ and $T > T_\epsilon$. The second term of (\ref{eq:Hkdivided}) is bounded due to the assumption $(T^{1/2}\norm{\delta_T})^{-1} \rightarrow 0$.
\begin{equation*}
    P \left(T^{-1/2}\norm{\delta_T}O_{p}(1) > \frac{\lambda\norm{\delta_T}^2}{3}\right) = P\left(\frac{O_{p}(1)}{T^{1/2}\norm{\delta_T}} > \frac{\lambda}{3}\right) < \frac{\epsilon}{3}.
\end{equation*}
The third term of (\ref{eq:Hkdivided}) is bounded for $k_0-k \geq C\norm{\delta_T}^{-2}$, since $O_{p}(1)/\vert k_0-k \vert$ $\leq O_{p}(1)\norm{\delta_T}^2/C$,
\begin{equation*}
    P\left(\sup_{k_0-k \geq C\norm{\delta_T}^{-2}}\frac{O_{p}(1)}{\vert k_0-k\vert} > \frac{\lambda \norm{\delta_T}^2}{3}\right)\leq P\left(\frac{O_{p}(1)}{C} > \frac{\lambda}{3}\right) < \frac{\epsilon}{3}
\end{equation*}
where $O_p(1)/C$ is small for large $T$, by choosing a large constant $C$. Hence the bound (\ref{eq:Hkbounded}) holds and the rate of convergence of the break point estimator $\hat{\rho} = \hat{k}/T$ in Theorem \ref{Thm:rateofconvergence} is proved: $\vert \hat{\rho}-\rho_0 \vert = O_p(T^{-1}\norm{\delta_T}^2)$.
\end{proof}
 
\bigskip

%%%%%%%%%%%%%%%%%%%%%%%%%%%%%%%%%%%%%%%%%%%%%%%%%%%%%%%%%%%%%%%%%% 
\noindent \textbf{Proof of Lemma \ref{Lem:Gkinversebound}}
\begin{proof}
Assumption \ref{A:Omegak} implies that $A_T(k)$ from (\ref{eq:Ak}) is positive definite and thus $G_{T}(k) = \delta_T'A_T(k)\delta_T$ $\geq \lambda_T(k)\norm{\delta_T}^2$, where $\lambda_{T}(k)$ is the minimum eigenvalue of $A_T(k)$. It is sufficient to argue that $\lambda_T(k)$ is bounded away from zero with probability tending to 1 as $\vert k_0 -k \vert$ increases. The matrices $Z_0'MZ_0$ and $Z_0'MZ_k$ are rearranged as follows, similar to $Z_k'MZ_k$ in (\ref{eq:zkmzk}).
\begin{align}
\nonumber
    Z_0'MZ_0 & = R'(X_0'X_0)(X'X)^{-1}(X'X-X_0'X_0)R\\
\label{eq:z0mzk}    
    Z_0'MZ_k & = \begin{cases}
    R'(X_0'X_0)(X'X)^{-1}(X'X-X_k'X_k)R & \text{if } k \leq k_0\\
    R'(X'X-X_0'X_0)(X'X)^{-1}(X_k'X_k)R & \text{if } k > k_0
    \end{cases}
\end{align}
Without loss of generality, assume $k \leq k_0$. The second term of $\vert k_0 - k \vert A_{T}(k)$ from (\ref{eq:Ak}) is 
\begin{align}
\nonumber
    (Z_0'MZ_k) & (Z_k'MZ_k)^{-1/2}\Omega_k (Z_k'MZ_k)^{-1/2} (Z_k'MZ_0)\\
\label{eq:Ak_secondterm}
    & = \left[R'(X_0'X_0)(X'X)^{-1}(X'X-X_k'X_k)R \right] \left[(Z_k'MZ_k)^{1/2}\Omega_k^{-1}(Z_k'MZ_k)^{1/2}\right]^{-1}\\
\nonumber
    & \hspace{0.5cm} \times \left[R'(X'X-X_k'X_k)(X'X)^{-1}(X_0'X_0)R\right].
\end{align}
Define the following matrices.
\begin{align}
\nonumber
    F_k & := (X_k'X_k)^{-1}-(X'X)^{-1} =(X'X)^{-1}(X'X-X_k'X_k)(X_k'X_k)^{-1}\\
\label{eq:FkF0Bdefinition}    
    F_0 & := (X_0'X_0)^{-1}-(X'X)^{-1} =(X'X)^{-1}(X'X-X_0'X_0)(X_0'X_0)^{-1}\\
\nonumber \Omega_{X,k}& := \begin{bmatrix}
I_{(p-q)} & \bzero_{(p-q)\times q} \\
\bzero_{q\times(p-q)} &  \Omega_k
\end{bmatrix},\;\;\;    B := \Omega_{X,k}^{-1/2}F_k^{1/2}X_k'X_k.
\end{align}
Both $F_k$ and $F_0$ are positive definite matrices under Assumption \ref{A:multimodel}. Hence, each matrix can be decomposed into $F_k = \left(F_k^{1/2}\right)^{2}$ and $F_0 = \left(F_0^{1/2}\right)^{2}$ where $F_k^{1/2}$ and $F_0^{1/2}$ are nonsingular $(p\times p)$ matrices with $p =$ dim($x_t$). $\Omega_{X,k}$ is a $(p\times p)$ matrix where $I_{(p-q)}$ is an identity matrix with rank $(p-q)$ and zeros in non-diagonal block matrices such that $R'\Omega_{X,k}R = \Omega_k$. The projection matrix of $BR$ is $I_p - BR(R'B'BR)^{-1}R'B'$, which is positive semi-definite. If we multiply $R'(X_0'X_0)F_k^{1/2}\Omega_{X,k}^{1/2}$ to the left and its transpose to the right of the projection matrix, the following inequality is obtained:
\begin{align*}
    R' (X_0'X_0) & F_k^{1/2}\Omega_{X,k}F_k^{1/2}(X_0'X_0)R \\
    & \geq R'(X_0'X_0)F_k(X_k'X_k)R(R'B'BR)^{-1}R'(X_k'X_k)F_k(X_0'X_0)R \\
    & = R'(X_0'X_0)(X'X)^{-1}(X'X'-X_k'X_k)R(R'B'BR)^{-1}\\
    & \hspace{1cm} \times  R'(X'X-X_k'X_k)(X'X)^{-1}(X_0'X_0)R.
\end{align*}
From (\ref{eq:FkF0Bdefinition}), $R'B'BR = R'(X_k'X_k)F_k^{1/2}\Omega_{X,k}^{-1}F_k^{1/2}(X_k'X_k)R$ $= (Z_k'MZ_k)^{1/2}\Omega_k^{-1}(Z_k'MZ_k)^{1/2}$ and the right side of the inequality is equivalent to (\ref{eq:Ak_secondterm}). Therefore, it is sufficient to show that the right-side of inequality (\ref{eq:Ainequality}) is bounded away from zero for large $k_0 -k$,
\begin{equation}\label{eq:Ainequality}
    A_{T}(k) \geq \frac{1}{\vert k_0-k\vert} \left[(Z_0'MZ_0)^{1/2}\Omega_{k_0}(Z_0'MZ_0)^{1/2}- R'(X_0'X_0) F_k^{1/2}\Omega_{X,k}F_k^{1/2}(X_0'X_0)R \right].
\end{equation}
Also from (\ref{eq:FkF0Bdefinition}), $(Z_0'MZ_0)^{1/2}\Omega_{k_0}(Z_0'MZ_0)^{1/2} = R'(X_0'X_0)F_0^{1/2}\Omega_{X,k_0}F_0^{1/2}(X_0'X_0)R$. Thus,
\begin{align*}
    \vert k_0-k\vert A_{T}(k) & \geq R'(X_0'X_0)\left[F_0^{1/2}\Omega_{X,k_0}F_0^{1/2}-F_k^{1/2}\Omega_{X,k} F_k^{1/2} \right](X_0'X_0)R \\
    & = \vert k_0-k \vert R'\tilde{A}_{T}(k)R,
\end{align*}
where $\tilde{A}_{T}(k) := \frac{1}{\vert k_0-k \vert}(X_0'X_0) \left[F_0^{1/2}\Omega_{X,k_0}F_0^{1/2}-F_k^{1/2}\Omega_{X,k} F_k^{1/2} \right](X_0'X_0)$. Then we have
\begin{align*}
    \norm{\tilde{A}_{T}(k)^{-1}} & = \vert k_0 - k \vert \norm{(X_0'X_0)^{-1}\left[F_0^{1/2}\Omega_{X,k_0} F_0^{1/2}-F_k^{1/2}\Omega_{X,k} F_k^{1/2} \right]^{-1}(X_0'X_0)^{-1}}\\
    & \leq \frac{1}{\norm{X_0'X_0}}\frac{1}{\vert k_0-k \vert^{-1}\norm{F_0^{1/2}\Omega_{X,k_0}F_0^{1/2}-F_k^{1/2}\Omega_{X,k} F_k^{1/2}}}\frac{1}{\norm{X_0'X_0}}.
\end{align*}
Note that for a nonsingular, bounded $(p \times p)$ matrix $S$, the norm does not change by multiplying $S$ on the left and $S^{-1}$ on the right of a matrix: $\norm{\Omega_{X,k}} = \norm{S\Omega_{X,k}S^{-1}}$. By assumption, $\Omega_{X,k} \geq \lambda_{\min}(\Omega_{X,k}) >0$ where $\lambda_{\min}$ denotes the minimum eigenvalue. Therefore, $(F_0^{1/2})^{-1}\Omega_{X,k_0} F_0^{1/2} = \Omega_{X,k_0}+o_p(1) \geq \lambda_{\min}(\Omega_{X,k_0}) + o_p(1)$. By subtracting and adding the matrix $F_k(F_0^{1/2})^{-1}\Omega_{X,k_0} F_0^{1/2}$ to the denominator of the second term, the following inequality holds: 
\begin{align*}
    & \hspace{0.5cm} \norm{F_0^{1/2}\Omega_{X,k_0} F_0^{1/2}-F_k^{1/2}\Omega_{X,k} F_k^{1/2}} \\
    & = \norm{(F_0-F_k)(F_0^{1/2})^{-1}\Omega_{X,k_0}F_0^{1/2}- F_k\left\{(F_k^{1/2})^{-1}\Omega_{X,k} F_k^{1/2}-(F_0^{1/2})^{-1}\Omega_{X,k_0} F_0^{1/2}\right\}}\\
    & \geq \left\vert \norm{(F_0-F_k)(F_0^{1/2})^{-1}\Omega_{X,k_0} F_0^{1/2}} - \norm{F_k\left\{(F_k^{1/2})^{-1}\Omega_{X,k}  F_k^{1/2}-(F_0^{1/2})^{-1}\Omega_{X,k_0}  F_0^{1/2}\right\}} \right\vert \\
     & = \left\vert \norm{(F_0-F_k)\Omega_{X,k_0}} - \norm{F_k\left(\Omega_{X,k}-\Omega_{X,k_0}\right)} \right\vert + o_p(1), 
\end{align*}
where the inequality is from the inverse triangular inequality. Let $\tilde\lambda$ be the minimum value of $\lambda_{\min}(\Omega_{X,k})$ and $\lambda_{\min}(\Omega_{X,k_0})$. From Assumption \ref{A:Omegak}, we have $\vert k_0-k\vert^{-1}\norm{\Omega_k-\Omega_{k_0}} \leq b/T$. Hence,
\begin{align*}
    & \hspace{0.5cm}  \vert k_0-k\vert^{-1}\norm{F_0^{1/2} \Omega_{X,k_0}  F_0^{1/2}-F_k^{1/2}\Omega_{X,k} F_k^{1/2}} \\
    & \geq \left\vert (k_0-k)^{-1} \norm{(F_0-F_k)\tilde\lambda} - b/T\norm{F_k}  \right\vert + \vert k_0-k\vert^{-1}o_p(1)\\
    & \geq \tilde{\lambda}\norm{(k_0-k)^{-1}(F_0-F_k)} + o_p(1).
\end{align*}
Let $X_{\Delta} := \text{sgn}(k_0-k)(X_k -X_0)$. By rearranging terms similar to (\ref{eq:FkF0Bdefinition}), $F_0-F_k = (X_0'X_0)^{-1} (X_{\Delta}'X_{\Delta})(X_k'X_k)^{-1}$ so that 
\begin{align*}
    \norm{\tilde{A}_{T}(k)^{-1}} & \leq  \frac{1}{\norm{X_0'X_0}^2}\frac{1}{\tilde{\lambda}\vert k_0-k\vert^{-1}\norm{F_0-F_k}}\\
    & \leq  \frac{1}{\tilde{\lambda}\norm{X_0'X_0}^2\norm{(X_0'X_0)^{-1}(k_0-k)^{-1}X_{\Delta}'X_{\Delta}(X_k'X_k)^{-1}}}.
\end{align*}
From Assumptions \ref{A:multimodel} and \ref{A:Omegak}, the right-side of the inequality is bounded:\\ $\tilde\lambda \norm{T^{-1}X_0'X_0}^2\norm{T^2(X_0'X_0)^{-1}(X_k'X_k)^{-1}} < M$, for some $M < \infty$. In addition, the minimum eigenvalue of $(k_0-k)^{-1}(X_{\Delta}'X_{\Delta})$ is bounded away from zero with large probability so that
$1/\norm{(k_0-k)^{-1}X_{\Delta}'X_{\Delta}}$ is bounded with large probability for all large $k_0-k$. Thus $\norm{\tilde{A}_{T}(k)^{-1}}$ is bounded with large probability for all large $ k_0-k$. This implies that the minimum eigenvalue of $\tilde{A}_{T}(k)$ is bounded away from zero for all large $k_0-k$ and this is also true for $A_{T}(k) = R'\tilde{A}_{T}(k)R$ because $R$ has full column rank. 
\end{proof}

\bigskip
%%%%%%%%%%%%%%%%%%%%%%%%%%%%%%%%%%%%%%%%%%%%
\noindent For the proof of Lemma \ref{Lem:consistent} we use Proposition \ref{Prop:HajekRenyi_mds} and Lemma \ref{Lem:HajekRenyi_mix}.

\begin{proposition}\label{Prop:HajekRenyi_mds}
Let $\varepsilon_1,\varepsilon_2,\ldots,$ be a sequence of martingale differences with $E[\varepsilon_t^2]=\sigma^2$ and $\{c_k\}$ be a decreasing positive sequence of constants. The H\'{a}jek-R\'{e}nyi inequality takes the following form.
\begin{equation*}
    P\left(\max_{m\leq k \leq T}c_k \left\vert\sum_{t=1}^k \varepsilon_t \right\vert > \alpha \right) \leq \frac{\sigma^2}{\alpha^2}\left(m c_m^2 + \sum_{t=m+1}^{T}c_t^2 \right).
\end{equation*}
\end{proposition}

\smallskip

\noindent \citet*{Hajek1955} proved the inequality assuming i.i.d. random variables, and was later generalized to martingales by \citet*{Birnbaum1961}. We use the generalized H\'{a}jek-R\'{e}nyi for martingale difference sequences to prove Lemma \ref{Lem:HajekRenyi_mix}, where $\{\varepsilon_t,\mathcal{F}_t\}$ are mixingale sequences under Assumption \ref{A:multimodel}.

\bigskip
%%%%%%%%%%%%%%%%%%%%%%%%%%%%%%%%%%%%%%%%%%%%
\begin{lemma}\label{Lem:HajekRenyi_mix}
Under Assumption \ref{A:multimodel}, for every $\alpha >0$ and $m > 0$ there exists $C < \infty$ such that
\begin{equation*}
    P\left(\sup_{m\leq k\leq T} \frac{1}{\sqrt{k}}\left\vert \sum_{t=1}^k z_t \varepsilon_t \right\vert > \alpha \right) \leq \frac{C \ln T}{\alpha^2},
\end{equation*}
and thus, $\sup_k T^{-1/2}\norm{Z_k'\varepsilon} =  O_p\left(\sqrt{\ln T}\right)$.

\end{lemma}

\begin{proof}
Denote $\xi_t = z_t\varepsilon_t$ and proceed. Let $\{\xi_t, \mathcal{F}_t\}$ be $(q\times 1)$ $L^r$-mixingales, $r = 4+\gamma$ for some $\gamma > 0$ satisfying Assumption \ref{A:multimodel}$(vi)$. Define $\xi_{jt} := E[\xi_t\vert \mathcal{F}_{t-j}]-E[\xi_t\vert \mathcal{F}_{t-j-1}]$. Then $\xi_t = \sum_{j=-\infty}^{\infty} \xi_{jt}$, and hence $\sum_{t=1}^k \xi_{t} = \sum_{j=-\infty}^{\infty}\sum_{t=1}^k \xi_{jt}$. Denote $\norm{\cdot}_{s}$ for the $L^s$-norm. For each $T>0$,
\begin{equation}\label{eq:mixingale_pf1}
    P\left(\sup_{m\leq k \leq T}\frac{1}{\sqrt{k}}\norm{\sum_{t=1}^k \xi_t}_{2} > \alpha \right) \leq P\left(\sum_{j=-\infty}^{\infty}\sup_{m\leq k\leq T}\frac{1}{\sqrt{k}}\norm{\sum_{t=1}^k \xi_{jt}}_{2} > \alpha \right).
\end{equation}
For each $j$, $\{\xi_{jt},\mathcal{F}_{t-j}\}$ forms a sequence of martingale difference and the generalized H\'{a}jek-R\'{e}nyi inequality (Proposition \ref{Prop:HajekRenyi_mds}) holds for this sequence. Let $b_j > 0$ for all $j$ and $\sum_{j=-\infty}^{\infty} b_j = 1$. The right-side of (\ref{eq:mixingale_pf1}) is bounded by
\begin{align*}
    \sum_{j=-\infty}^{\infty} & P\left(\sup_{m\leq k\leq T}\frac{1}{\sqrt{k}}\norm{\sum_{t=1}^k \xi_{jt}}_{2} > b_j\alpha \right) \\
    & \leq \frac{1}{\alpha^2} \sum_{j=-\infty}^{\infty}\frac{1}{b_j^2}\left( m^{-1} \sum_{i=1}^{m}E\norm{\xi_{ji}}_{2}^2 + \sum_{i=m+1}^{T}i^{-1} E\norm{\xi_{ji}}_{2}^2\right).
\end{align*}
Note that $\norm{\xi_{jt}}_2 \leq \norm{\xi_{jt}}_r$ by Liapounov's inequality. By definition, for $j \geq 0$, we have $\norm{\xi_{jt}}_r \leq \norm{E[\xi_{t}\vert \mathcal{F}_{t-j}]}_{r} + \norm{E[\xi_{t}\vert \mathcal{F}_{t-j-1}]}_{r}$ and for $j < 0$, $\norm{\xi_{jt}}_{r} \leq \norm{\xi_{t}-E[\xi_t\vert \mathcal{F}_{t-j-1}]}_r + \norm{\xi_t - E[\xi_t\vert \mathcal{F}_{t-j}]}_{r}$. Hence, from the definition of a mixingale,  $E\norm{\xi_{jt}}_{r}^{2} \leq 4c_{t}^2\psi_{\vert j\vert}^2$ and with Assumption \ref{A:multimodel}$(vi)(c)$, this implies $E\norm{\xi_{ji}}_2^2 \leq 4c_i^2 \psi_{\vert j\vert}^2 \leq 4K^2 \psi_{\vert j\vert}^2$. Then the right-side of (\ref{eq:mixingale_pf1}) is bounded by
\begin{equation}\label{eq:lnTbound_mix}
    \frac{1}{\alpha^2} \sum_{j=-\infty}^{\infty} 4b_j^{-2}K^2\psi_{\vert j\vert}^2\left( 1 + \sum_{i=m+1}^{T}i^{-1}\right) \leq \frac{1}{\alpha^2} \sum_{j=-\infty}^{\infty} 4b_j^{-2}K^2\psi_{\vert j\vert}^2\left( 1 + \ln T \right).
\end{equation}
We can choose appropriate $\{b_j\}$ so that $\sum_{j} b_j^{-2}\psi_{\vert j\vert}^2$ are bounded. From lemma A.6 of \citet*{Bai1998a}, let $\nu_0 = 1$ and $\nu_j = j^{-1-\kappa}$ for $j\geq 1$, where $\kappa>0$ is given in Assumption \ref{A:multimodel}$(vi)(d)$. Let $b_j = \nu_j/\left(1+2\sum_{i=1}^{\infty}\nu_i \right)$ and $b_{-j} = b_j$ for all $j \geq 0$. Then $\sum_{j=-\infty}^{\infty} b_j = 1$. By Assumption \ref{A:multimodel}$(vi)(d)$, we have
\begin{equation*}
    \sum_{j=-\infty}^{\infty} b_{j}^{-2}\psi_{\vert j \vert}^2 = \left(\psi_0^2 + 2\sum_{j=1}^{\infty}j^{2+2\kappa}\psi_{j}^2\right)\left(1+2\sum_{j=1}^{\infty}j^{-2-2\kappa}\right) < \infty.
\end{equation*}
Hence, the right-side inequality of (\ref{eq:lnTbound_mix}) is bounded by $\frac{C \ln T}{\alpha^2}$ for some $C>0$ and we obtain the result of Lemma \ref{Lem:HajekRenyi_mix}.
\end{proof}

\bigskip
%%%%%%%%%%%%%%%%%%%%%%%%%%%%%%%%%%%%%%%%%%%%
\noindent\textbf{Proof of Lemma \ref{Lem:consistent}}
\begin{proof}
We use the expression (\ref{eq:QinGandH}); the estimator $\hat{k}$ must satisfy $Q_{T}(\hat{k})^2 \geq Q_{T}(k_0)^2$ which is equivalent to $H_{T}(\hat{k}) \geq \vert k_0-\hat{k}\vert G_{T}(\hat{k})$. Therefore we have
\begin{align}
\nonumber
    P\left(\vert \hat{\rho}-\rho_0 \vert > \eta \right) & = P\left(\vert \hat{k}-k_0 \vert > T\eta \right)\\
\nonumber
    & \leq P\left(\sup_{\vert k-k_0\vert > T\eta} \vert H_{T}(k) \vert \geq \inf_{\vert k-k_0\vert > T\eta}\vert k_0 -k\vert G_{T}(k)\right)\\
\nonumber    
    & \leq P\left(\sup_{\vert k-k_0\vert > T\eta} \vert H_{T}(k) \vert \geq T\eta \inf_{\vert k-k_0\vert > T\eta} G_{T}(k)\right)\\
\label{eq:Hk_inequality}
    & \leq P\left(\tilde{G}^{-1}\sup_{p\leq k \leq T-p} T^{-1}\vert H_T(k) \vert \geq \eta \right),
\end{align}
where $\tilde{G} := \inf_{\vert k- k_0\vert > T\eta}G_{T}(k)$, which is positive and bounded away from zero by Lemma \ref{Lem:Gkinversebound} and the restriction $p \leq k \leq T-p$ is imposed to guarantee existence of $H_T(k)$. Thus, the consistency follows by showing that $T^{-1}\sup_{p\leq k \leq T-p} \vert H_T(k) \vert = o_{p}(1)$, where 
\begin{align}
\nonumber
    T^{-1}\vert H_T(k) \vert & \leq \left\vert T^{-1} \varepsilon'MZ_k(Z_k'MZ_k)^{-1/2}\Omega_k(Z_k'MZ_k)^{-1/2}Z_k'M\varepsilon \right\vert \\
\nonumber
    & \hspace{0.7cm} + \left\vert T^{-1} \varepsilon'MZ_0(Z_0'MZ_0)^{-1/2}\Omega_{k_0}(Z_0'MZ_0)^{-1/2}Z_0'M\varepsilon \right\vert\\
\label{eq:HkdivideT_inequality}
    & \hspace{0.7cm} + 2\left\vert T^{-1}\delta_T'(Z_0'MZ_k)(Z_k'MZ_k)^{-1/2}\Omega_k (Z_k'MZ_k)^{-1/2}Z_k'M\varepsilon \right\vert \\
\nonumber
    & \hspace{0.7cm} + 2\left\vert T^{-1}\delta_T'(Z_0'MZ_0)^{1/2}\Omega_{k_0} (Z_0'MZ_0)^{-1/2}Z_0'M\varepsilon \right\vert.
\end{align}
Lemma \ref{Lem:HajekRenyi_mix} implies that $\sup_k \norm{T^{-1/2} Z_k'M\varepsilon} = O_p\left(\sqrt{\ln T}\right)$. We use (\ref{eq:LIL}) to verify uniform convergence for all $k$.
\begin{equation}\label{eq:LIL}
    \sup_{p \leq k \leq T-p} \norm{(Z_k'MZ_k)^{-1/2}Z_k'M\varepsilon} = O_{p}\left(\sqrt{\ln T}\right).
\end{equation}
We show that the third and fourth terms of (\ref{eq:HkdivideT_inequality}) are $O_p\left(T^{-1/2}\norm{\delta_T}\ln T\right)$ and $O_p\left(T^{-1/2}\norm{\delta_T}\right)$, respectively. Denote $D_T := T^{-1/2}(Z_0'MZ_k)(Z_k'MZ_k)^{-1/2}$. From (\ref{eq:zkmzk}) and (\ref{eq:z0mzk}), $Z_0'MZ_k \leq Z_k'MZ_k$ for all $k$, and thus
\begin{equation*}
    \sup_{p \leq k \leq T-p}D_T'D_T \leq \sup_{p \leq k \leq T-p}T^{-1}Z_k'MZ_k = O_{p}(1),
\end{equation*}
then the third term of (\ref{eq:HkdivideT_inequality}) is bounded by 
\begin{equation*}
     2\norm{ T^{-1/2}\delta_T'(T^{-1}Z_k'MZ_k)^{1/2}\Omega_{k}}\norm{ (Z_k'MZ_k)^{-1/2}Z_k'M\varepsilon} = \norm{\delta_T}O_p\left(T^{-1/2}\sqrt{\ln T}\right).
\end{equation*}
For the true break date $k_0$, $\left\vert(Z_0'MZ_0)^{-1/2}Z_0'M\varepsilon \right\vert = O_p(1)$ under our regularity conditions. Hence, the fourth term of (\ref{eq:HkdivideT_inequality}) has order
\begin{equation*}
   \left\vert T^{-1}\delta_T'(Z_0'MZ_0)^{1/2}\Omega_{k_0}(Z_0'MZ_0)^{-1/2}Z_0'M\varepsilon \right\vert = \norm{\delta_T} O_p\left(T^{-1/2}\right).
\end{equation*}
From (\ref{eq:LIL}) and boundedness of $\Omega_k$, we have $\sup_k \norm{\Omega_k^{1/2}(Z_k'MZ_k)^{-1/2}Z_k'M\varepsilon} = O_p\left(\sqrt{\ln T}\right)$. Then the first and second terms of (\ref{eq:HkdivideT_inequality}) are bounded as below, respectively. 
\begin{gather*}
    \sup_k T^{-1} \norm{\Omega_k^{1/2}(Z_k'MZ_k)^{-1/2}Z_k'M\varepsilon}^2  = O_p\left(T^{-1}\ln T \right) \\
    T^{-1}\norm{\Omega_{k_0}^{1/2}(Z_0'MZ_0)^{-1/2}Z_0'M\varepsilon}^2 = O_p(T^{-1}).
\end{gather*}
By combining all four terms, $T^{-1}\sup_{p \leq k \leq T-p}\vert H_T(k) \vert = o_p(1)$. Hence, probability (\ref{eq:Hk_inequality}) is negligible for large $T$. 
\end{proof}

\bigskip

%%%%%%%%%%%%%%%%%%%%%%%%%%%%%%%%%%%%%%%%%
\noindent\textbf{Proof of Corollary \ref{Cor:delta_asym}}
\begin{proof}
Let $\hat{Z}_0$ denote $Z_k$ when $k$ is replaced by $\hat{k}$. Then the LS estimator of $\hat{\delta}(\hat{\rho})$ is obtained by regressing $MY$ on $M\hat{Z}_0$. By multiplying $M$ to model (\ref{eq:DGPmulti}) can be rewritten as $MY = M\hat{Z}_0 \delta_T + M\varepsilon^{*}$, where $\varepsilon^{*} = \varepsilon + (Z_0-\hat{Z}_0)\delta_T$. Then,
\begin{align*}
    \sqrt{T}\left(\hat{\delta}(\hat{\rho})-\delta_T\right) & = (T^{-1}\hat{Z}_0'M\hat{Z}_0)^{-1}T^{-1/2}\hat{Z}_0'M\varepsilon^{*} \\
    & = (T^{-1}\hat{Z}_0'M\hat{Z}_0)^{-1}\left(T^{-1/2}\hat{Z}_0'M\varepsilon + T^{-1/2}\hat{Z}_0'M(Z_0-\hat{Z}_0)\delta_T\right).
\end{align*}
We show that the right side converges in probability to the same limit as when $\hat{Z}_0$ is replaced by $Z_0$. First, we show that $\plim T^{-1/2}\hat{Z}_0'M(Z_0-\hat{Z}_0)\delta_T =0$. Without loss of generality, consider $k > k_0$.
\begin{align*}
    \norm{T^{-1/2}\hat{Z}_0'M(Z_0-\hat{Z}_0)\delta_T} & \leq T^{-1/2} \norm{\hat{Z}_0'(Z_0-\hat{Z}_0)-\hat{Z}_0'X(X'X)^{-1}X'(Z_0-\hat{Z}_0)} \norm{\delta_T} \\
    & \leq \frac{1}{\sqrt{T}\norm{\delta_T}} \norm{\sum_{t=k_0+1}^{\hat{k}}z_tz_t'}\norm{\delta_T}^2\\
    & \hspace{0.7cm} + \norm{\left(\sum_{t=\hat{k}+1}^{T}z_t x_t'\right)\left(\sum_{t=1}^{T}x_t x_t'\right)^{-1}} \frac{1}{\sqrt{T}\norm{\delta_T}} \norm{\sum_{t=k_0+1}^{\hat{k}}x_tz_t'}\norm{\delta_T}^2 \\
    & = \frac{1}{\sqrt{T}\norm{\delta_T}} O_p(1) = o_p(1).
\end{align*}
The last equality is because the sum has  $\hat{k}-k_0 = O_p\left(\norm{\delta_T}^{-2}\right)$ terms, hence, $\norm{\sum_{t=\hat{k}+1}^{k_0}z_t z_t'}$ $\times\norm{\delta_T}^2= O_p(1)$. Also,
\begin{align*}
    T^{-1}\norm{\hat{Z}_0'M\hat{Z}_0-Z_0'MZ_0} & \leq T^{-1}\norm{\hat{Z}_0'\hat{Z}_0-Z_0'Z_0}+T^{-1}\norm{(\hat{Z}_0-Z_0)'X(X'X)^{-1}X'\hat{Z}_0} \\
    & \hspace{0.7cm} + T^{-1}\norm{Z_0'X(X'X)^{-1}X'(\hat{Z}_0-Z_0)} \\
    & \leq \frac{1}{T\norm{\delta_T}^2}\norm{\sum_{t=k_0+1}^{\hat{k}}z_t z_t'}\norm{\delta_T}^2 + \frac{2}{T\norm{\delta_T}^2}\norm{\sum_{t=k_0+1}^{\hat{k}}x_t z_t'}\norm{\delta_T}^2 O_p(1) \\
    & = \frac{1}{T\norm{\delta_T}^2} O_p(1) = o_p(1).
\end{align*}
Thus, $\sqrt{T}\left(\hat{\delta}(\hat{\rho})-\delta_T\right) = (T^{-1}Z_0'MZ_0)^{-1}T^{-1/2}Z_0'M\varepsilon+ o_p(1)$, and the normality follows from the central limit theorem.
\end{proof}

\bigskip

%%%%%%%%%%%%%%%%%%%%%%%%%%%%%%%%%%%%%%%%%%%%%%%%%%%%%%%%%%%%%%%%%
\noindent \textbf{Proof of Theorem \ref{Thm:infillmult_dfixed}}

\begin{proof}
When $\delta_h = d_0\sqrt{h}$, the break point estimator $\norm{\delta_h}^2(\hat{k}- k_0) = \norm{d_0}^2(\hat{\rho}-\rho_0) =O_p(1)$ has values in the interval $(-\rho_0\norm{d_0}^2, (1-\rho_0)\norm{d_0}^2)$. Therefore, we only need to examine the behavior of the objective function $Q_T(k)^2$ for those $k$ in the neighborhood of $k_0$ such that $k = \left[k_0 + s\norm{d_0\sqrt{h}}^{-2}\right]$, with 
$s \in (-\rho_0\norm{d_0}^2, (1-\rho_0)\norm{d_0}^2)$. Then for any fixed $s$, if $h \rightarrow 0$ then $T \rightarrow \infty$ with $k/T \rightarrow \rho = \rho_0+ u$, and $T-k \rightarrow \infty$ with $(T-k)/T \rightarrow 1-\rho = 1-\rho_0-u$, where $u = s\norm{d_0}^{-2} \in (-\rho_0, 1-\rho_0)$. From the objective function (\ref{eq:objfn_Tscale}), we have
\begin{equation*}
    (T^{-1}Z_k'MZ_k)^{1/2}\sqrt{T}\hat{\delta}_k  = (T^{-1}Z_k'MZ_k)^{-1/2}(T^{-1}Z_k'MZ_0)d_0 + (T^{-1}Z_k'MZ_k)^{-1/2}T^{-1/2}Z_k'M\varepsilon.
\end{equation*}
Consider each of the terms as $h \rightarrow 0$, which is equivalent to $T \rightarrow \infty$.
\begin{align*}
   T^{-1}Z_k'MZ_0 & = T^{-1}\sum_{t=\max\{k,k_0\}+1}^T z_tz_t' - \left(T^{-1}\sum_{t=k+1}^T R'x_t x_t'\right)(T^{-1}X'X)^{-1}\left(T^{-1}\sum_{t=k_0+1}^T x_t x_t'R\right) \\
    & \longrightarrow \left(1-\max\{\rho,\rho_0\}\right)\Sigma_z - (1-\rho)(1-\rho_0)R'\Sigma_{x}\Sigma_{x}^{-1}\Sigma_{x}R \\
    & = \left(\min\{\rho,\rho_0\}-\rho\cdot\rho_0\right)\Sigma_z \\
    T^{-1}Z_k'MZ_k & = T^{-1}\sum_{t=k+1}^T z_tz_t' - \left(T^{-1}\sum_{t=k+1}^T R'x_t x_t'\right)(T^{-1}X'X)^{-1}\left(T^{-1}\sum_{t=k+1}^T x_t x_t'R\right) \\
    & \longrightarrow \rho(1-\rho)\Sigma_z\\
   T^{-1/2}Z_k'M\varepsilon & = T^{-1/2}\sum_{t=k+1}^T z_t \varepsilon_t - (T^{-1}Z_k'X)(T^{-1}X'X)^{-1}\left(T^{-1/2}\sum_{t=1}^T x_t \varepsilon_t \right) \\
   & \Rightarrow B_1(\rho)-\rho B_1(1).
\end{align*}
By assumption, $\Omega_k \overset{p}{\rightarrow} \bar{\Omega}(\rho)$ as $k/T \rightarrow \rho$. This implies that for a fixed $d_0$, the objective function $Q_T(k)^2$ weakly converges as follows. For $\rho \leq \rho_0$,
\begin{multline*}
    Q_T(k)^2 \Rightarrow \frac{1}{\rho(1-\rho)}\left[ B_1(\rho)-\rho B_1(1)-\rho(1-\rho_0)\Sigma_z d_0 \right]'\\
    \times \Sigma_{z}^{-1/2}\bar{\Omega}(\rho)\Sigma_{z}^{-1/2}
    \left[B_1(\rho)-\rho B_1(1)- \rho(1-\rho_0)\Sigma_z d_0\right],
\end{multline*}
and for $\rho > \rho_0$,
\begin{multline*}
    Q_T(k)^2 \Rightarrow  \frac{1}{\rho(1-\rho)}\left[ B_1(\rho)- \rho B_1(1)-\rho_0(1-\rho)\Sigma_z d_0 \right]'\\
    \times \Sigma_{z}^{-1/2}\bar{\Omega}(\rho)\Sigma_{z}^{-1/2}\left[ B_1(\rho)- \rho B_1(1)-\rho_0(1-\rho)\Sigma_z d_0 \right].
\end{multline*}
By continuous mapping theorem, the the in-fill asymptotic distribution of $T\norm{\delta_h}^2\hat{\rho}$ is the argmax functional of the limit of $Q_T(k)^2$, stated in Theorem \ref{Thm:infillmult_dfixed}.
\end{proof}

\bigskip

%%%%%%%%%%%%%%%%%%%%%%%%%%%%%%%%%%%%%%%%%%%%%%%%%%%%%%%%%
\noindent \textbf{Proof of Theorem \ref{Thm:infillmult_dlarge}}

\begin{proof}
We omit the proof of consistency of the break point estimator $\hat{\rho} \rightarrow \rho_0$, because it follows the same procedure as the proof of Theorem \ref{Thm:rateofconvergence}. Given the rate of convergence $\hat{\rho}-\rho_0 = O_p\left(T^{-1}\lambda_h^{-2} \right)$, we only need to examine the behavior of $Q_T(k)^2-Q_T(k_0)^2$ for those $k$ in the neighborhood of $k_0$ such that $k \in K(C)$, where $K(C) = \{k: \vert k-k_0 \vert \leq C\lambda_h^{-2}\}$ for some $C >0$. 

\bigskip

\begin{lemma}\label{lem:Qdiff_limitdist}
Consider the model (\ref{eq:DGPinfill_multi}) and the weight matrix $\Omega_k$ that satisfies Assumption \ref{A:Omegak}. For the break magntiude $\delta_h = d_0 \lambda_h$ that satisfies Assumption \ref{A:infillmult_delta}$(ii)$, 
\begin{multline*}
    Q_T(k)^2 - Q_T(k_0)^2 = -\lambda_h^2 d_0'Z_{\Delta}'Z_{\Delta}\Omega_{k_0}d_0 -\lambda_h^2 d_0'(Z_0'MZ_0)(\Omega_{k_0}-\Omega_k)d_0 \\
    + 2\lambda_h d_0'\Omega_{k_0} Z_{\Delta}'\varepsilon \, \text{sgn}(k_0 -k) + o_p(1)
\end{multline*}
where $Z_{\Delta}:= \text{sgn}(k_0-k)(Z_k-Z_0)$ and $o_p(1)$ is uniform on $K(C)$.
\end{lemma}

\smallskip

%%%%%%%%%%%%%%% Back to proof of Limit distribution
\noindent Because $\delta_h = d_0 \lambda_h$, for any constant $C$ of $K(C)$, we consider the limiting process of $Q_T(k)^2-Q_T(k_0)^2$ for $k = \left[k_0+ \nu \lambda_h^{-2}\right]$ and $\nu \in [-C, C]$. Consider $\nu \leq 0$ (i.e., $\rho \leq \rho_0$). From lemma \ref{lem:Qdiff_limitdist},
\begin{align*}
    Q_T(k)^2-Q_T(k_0)^2 & = -d_0'\left( \lambda_h^{2} \sum_{t=k+1}^{k_0}z_t z_t'\right)\Omega_{k_0} d_0 - d_0'(T^{-1}Z_0'MZ_0)\lambda_h^2 T(\Omega_{k_0}-\Omega_k)d_0 \\
    & \hspace{1cm}+ 2 d_0'\Omega_{k_0}\left( \lambda_h\sum_{t=k+1}^{k_0}z_t \varepsilon_t\right)+ o_p(1).
\end{align*}
For $k_0-k = \left[-\nu \lambda_h^{-2}\right]$, the partial sum in the first term converges as follows:
\begin{equation*}
    \lambda_h^{2}\sum_{t=k+1}^{k_0}z_t z_t' \overset{p}{\longrightarrow} \vert \nu \vert \Sigma_{z}.
\end{equation*}
Consider the second term; from the proof of Theorem \ref{Thm:infillmult_dfixed}, $T^{-1}Z_0'MZ_0 \overset{p}{\rightarrow} \rho_0(1-\rho_0)\Sigma_z$ and 
\begin{equation*}
    \lambda_h^2 T (\Omega_{k_0}-\Omega_k) = \frac{-\nu (\Omega_{k_0}-\Omega_k)}{-\nu \lambda_h^{-2}T^{-1}} = \frac{-\nu (\Omega_{k_0}-\Omega_k)}{\rho_0-\rho},
\end{equation*}
where $\Omega_k =\bar{\Omega}(\rho) + o_p(1)$ for fixed $\rho$. From Assumption \ref{A:Omegak}, $\bar{\Omega}(\rho)$ is a differentiable function of $\rho$ element-wise and by the mean value theorem, $\bar{\Omega}(\rho_0)-\bar{\Omega}(\rho) = \left.\frac{\partial\bar{\Omega}(\rho)}{\partial \rho}\right\vert_{\tilde{\rho}}(\rho_0-\rho)$ for some $\tilde{\rho} \in (\rho,\rho_0)$. Thus, the equation converges to $-\nu \left.\frac{\partial\bar{\Omega}(\rho)}{\partial \rho}\right\vert_{\rho_0} \equiv -\nu \nabla\bar{\Omega}_0$, under the break magntiude \ref{A:infillmult_delta}$(ii)$.

The third term has partial sum of $z_t \varepsilon_t$, which weakly converges to a Brownian motion process $B_1(-\nu)$ on $[0,\infty)$ that has variance $\vert \nu \vert \Xi$.
\begin{equation*}
    \lambda_h \sum_{t=k+1}^{k_0}z_t \varepsilon_t \Rightarrow B_1(-\nu).
\end{equation*}
Therefore,
\begin{align*}
    Q_T\left(\left[k_0+\nu\lambda_h^{-2}\right]\right)^2 & -Q_T(k_0)^2 \\
    & \Rightarrow -\vert \nu\vert d_0'\Sigma_z \bar{\Omega}_0 d_0 +\nu\rho_0(1-\rho_0) d_0'\Sigma_z \nabla \bar{\Omega}_0 d_0 + 2 d_0'\bar{\Omega}_0 B_1(-\nu). 
\end{align*}
Let $W_1(\cdot)$ and $W_2(\cdot)$ be an Wiener processes that are independent of each other on $[0,\infty)$. Define $A_{\nu} := \bar{\Omega}_0 -\text{sgn}(\nu) \rho_0(1-\rho_0) \nabla \bar{\Omega}_0$.
\begin{equation*}
    \wtilde{G}(\nu) := \begin{cases}
    -\frac{\vert \nu \vert}{2}(d_0'\Sigma_z A_{\nu} d_0) + (d_0'\bar{\Omega}_0\Xi\bar{\Omega}_0d_0)^{1/2}W_1(-\nu) & \text{ if } \nu \leq 0\\
    -\frac{\vert \nu \vert}{2}(d_0'\Sigma_z A_{\nu} d_0) + (d_0'\bar{\Omega}_0\Xi\bar{\Omega}_0 d_0)^{1/2}W_2(\nu) & \text{ if } \nu > 0.
    \end{cases}
\end{equation*}
From the continuous mapping theorem, the in-fill asymptotic distribution of the break point estimator is $\lambda_h^2(\hat{k}-k_0) \Rightarrow \argmax_{\nu} \wtilde{G}(\nu)$. Denote $A_\nu$ in terms of $u$ so that $A_u = \bar{\Omega}_0 -\text{sgn}(u) \rho_0(1-\rho_0) \nabla \bar{\Omega}_0$. Let $\nu=c u$ where $c = (d_0'\bar{\Omega}_0\Xi\bar{\Omega}_0 d_0)/(d_0'\Sigma_z A_u d_0)^2$ and $u \in (-\infty,\infty)$. For $\nu\leq 0$,
\begin{align*}
    \argmax_{\nu \in (-\infty,0]} \wtilde{G}(\nu) & = \argmax_{c u \in (-\infty,0]} -\frac{\vert u \vert}{2}c(d_0'\Sigma_z A_u d_0) + c^{1/2}(d_0'\bar{\Omega}_0\Xi\bar{\Omega}_0 d_0)^{1/2}W_1(-u) \\
    & = \argmax_{cu \in (-\infty,0]} (d_0'\Sigma_z A_u d_0)^{-1} \left\{ W_1(-u)-\frac{\vert u \vert}{2} \right\} \\
    & = c \argmax_{u \in (-\infty,0]} (d_0'\Sigma_z A_u d_0)^{-1}\left\{ W_1(-u)-\frac{\vert u \vert}{2} \right\},
\end{align*}
where the second equality is from
\begin{equation*}
    c(d_0'\Sigma_z A_u d_0) = c^{1/2}(d_0'\bar{\Omega}_0\Xi\bar{\Omega}_0 d_0)^{1/2}= (d_0'\bar{\Omega}_0\Xi\bar{\Omega}_0 d_0)/(d_0'\Sigma_z A_u d_0),
\end{equation*}
and the numerator does not depend on $u$. For $\nu > 0$ we have $\argmax_{\nu \in (0,\infty)} \wtilde{G}(\nu) =  c\argmax_{u \in (0,\infty)} (d_0'\Sigma_z A_u d_0)^{-1}\{W_2(u)-\vert u\vert/2\}$. Thus,
\begin{equation*}
    \lambda_h^{2}(\hat{k}-k_0) \Rightarrow \frac{(d_0'\bar{\Omega}_0\Xi\bar{\Omega}_0 d_0)}{(d_0'\Sigma_z A_u d_0)^2} \argmax_{u \in (-\infty,\infty)} (d_0'\Sigma_z A_u d_0)^{-1}\left\{ W(u)- \frac{\vert u \vert}{2} \right\},
\end{equation*}
for $W(\cdot)$ defined in Theorem \ref{Thm:infillmult_dlarge}.
\end{proof}

\bigskip

%%%%%%%%%%%%%%%%%%%%%%%%%%%%%%%%%%%%%%%%%%
\noindent\textbf{Proof of Lemma \ref{lem:Qdiff_limitdist}}
\begin{proof}
Use equation (\ref{eq:QinGandH}), $Z_0 = Z_k - Z_{\Delta}\text{sgn}(k_0-k)$ and rearrange terms so that
\begin{align*}
    \vert k_0 & -k\vert G_T(k) \\
    & = \delta_h'(Z_0'MZ_0)^{1/2}\Omega_{k_0} (Z_0'MZ_0)^{1/2}\delta_h\\
    & \hspace{0.5cm} -\delta_h(Z_0'MZ_k)(Z_k'MZ_k)^{-1/2}\Omega_k (Z_k'MZ_k)^{-1/2}(Z_k'MZ_0)\delta_h\\
    & = \delta_h'\left[Z_{\Delta}'MZ_{\Delta}\Omega_{k_0} + (Z_k'MZ_k)(\Omega_{k_0}-\Omega_k)\right]\delta_h \\
    & \hspace{0.5cm} + \delta_h'\left[2\text{sgn}(k_0-k)(Z_{\Delta}'MZ_k)(Z_k'MZ_k)^{-1/2}\Omega_k(Z_k'MZ_k)^{1/2}  \right. \\
    & \hspace{1.5cm} \left. -2 \text{sgn}(k_0-k)Z_{\Delta}'MZ_k\Omega_{k_0}- (Z_{\Delta}'MZ_k)(Z_k'MZ_k)^{-1/2}\Omega_k(Z_k'MZ_k)^{-1/2}Z_k'MZ_{\Delta}\right] \delta_h\\
    & = \lambda_h^2 d_0'\left[Z_{\Delta}'MZ_{\Delta}\Omega_{k_0} + (Z_k'MZ_k)(\Omega_{k_0}-\Omega_k)\right]d_0 + o_p(1).
\end{align*}
The last equality is because $(Z_{\Delta}'MZ_k) = \norm{\delta_h}^{-2}O_p(1)$ uniformly on $K(C)$ and 
\begin{gather*}
    (Z_k'MZ_k)^{-1/2}\Omega_k(Z_k'MZ_k)^{1/2}-\Omega_{k_0} = \Omega_k-\Omega_{k_0} + o_p(1) \leq b\vert k_0-k\vert/T = o_p(1)\\
    (Z_{\Delta}'MZ_k)(Z_k'MZ_k)^{-1/2}\Omega_k(Z_k'MZ_k)^{-1/2}Z_k'MZ_{\Delta} = \norm{\delta_h}^{-4}T^{-1}O_p(1).
\end{gather*}
Similarly, from $Z_k = Z_0 - Z_{\Delta}\text{sgn}(k_0-k)$, we can show that $(Z_k'MZ_k)(\Omega_{k_0}-\Omega_k) = (Z_0'MZ_0)(\Omega_{k_0}-\Omega_k) + o_p(1)$ and
\begin{align*}
    \lambda_h^2d_0(Z_{\Delta}'MZ_{\Delta})\Omega_{k_0}d_0 & =   \lambda_h^2d_0(Z_{\Delta}'Z_{\Delta}-Z_{\Delta}'X(X'X)^{-1}X'Z_{\Delta})\Omega_{k_0}d_0 \\
    & = \lambda_h^2d_0Z_{\Delta}'Z_{\Delta}\Omega_{k_0}d_0 + o_p(1),
\end{align*}
because $Z_{\Delta}'X = \norm{\delta_h}^{-2}O_p(1)$ and $(X'X)^{-1} = O_p(T^{-1})$. Hence,
\begin{equation}\label{eq:G_limitdist}
    \vert k_0 -k\vert G_T(k) = \lambda_h^2 d_0'Z_{\Delta}'Z_{\Delta}\Omega_{k_0}d_0 + \lambda_h^2 d_0'(Z_0'MZ_0)(\Omega_{k_0}-\Omega_k)d_0 + o_p(1).
\end{equation}
Next, consider $H_T(k)$ in equation (\ref{eq:Hkdivided}).
\begin{equation*}
    H_T(k) = 2\lambda_h d_0'\Omega_{k_0} Z_{\Delta}'\varepsilon \, \text{sgn}(k_0-k) + T^{-1/2}\norm{\delta_h}\vert k_0-k\vert O_p(1) + O_p(1).
\end{equation*}
Because $\vert k_0-k\vert \leq C\norm{\delta_h}^{-2}$ on $K(C)$, the second term in the above equation is bounded by $C T^{-1/2}\norm{\delta_h}^{-1}O_p(1) = o_p(1)$. The last term $O_p(1)$ is $o_p(1)$ uniformly on $K(C)$, which can be verified by rearranging terms using $Z_0 = Z_k - Z_{\Delta} \text{sgn}(k_0 -k)$.
\begin{align*}
    \varepsilon'MZ_k &(Z_kMZ_k)^{-1/2}\Omega_k (Z_kMZ_k)^{-1/2}Z_k'M\varepsilon -     \varepsilon'MZ_0(Z_0MZ_0)^{-1/2}\Omega_{k_0} (Z_0MZ_0)^{-1/2}Z_0'M\varepsilon\\
    & =  \varepsilon'MZ_k\left[(Z_k'MZ_k)^{-1/2}\Omega_k (Z_k'MZ_k)^{-1/2} -(Z_0'MZ_0)^{-1/2}\Omega_{k_0} (Z_0'MZ_0)^{-1/2}\right]Z_k'M\varepsilon \\
     & \hspace{0.7cm} + \varepsilon'MZ_k (Z_0'MZ_0)^{-1/2}\Omega_{k_0} (Z_0'MZ_0)^{-1/2} Z_{\Delta}'M\varepsilon \, \text{sgn}(k_0-k)\\
    & \hspace{0.7cm} +\varepsilon'MZ_{0} (Z_0'MZ_0)^{-1/2}\Omega_{k_0} (Z_0'MZ_0)^{-1/2} Z_{\Delta}'M\varepsilon  \, \text{sgn}(k_0-k) \\
    & = O_p\left(T^{-1}\norm{\delta_h}^{-2}\right) + O_p\left(T^{-1/2}\norm{\delta_h}^{-1}\right) + o_p(1).
\end{align*}
The first line is $O_p\left(T^{-1}\norm{\delta_h}^{-2}\right)$ is uniformly on $K(C)$,
\begin{align*}
    \varepsilon'MZ_k & \left[(Z_k'MZ_k)^{-1}\Omega_k -(Z_0'MZ_0)^{-1}\Omega_{k_0} \right]Z_k'M\varepsilon +o_p(1)\\
    & = \varepsilon'MZ_k (Z_k'MZ_k)^{-1}(Z_0'MZ_0- Z_k'MZ_k)(Z_0'MZ_0)^{-1}\Omega_k Z_k'M\varepsilon \\
    & \hspace{1cm} + \varepsilon'MZ_k(Z_0'MZ_0)^{-1}(\Omega_k-\Omega_{k_0})Z_k'M\varepsilon + o_p(1) \\
    & = T^{-1/2}\norm{\delta_h}^{-2} O_p(T^{-1/2})+ o_p(1),
\end{align*}
from (\ref{eq:ZkminusZ0}) and $(Z_k'MZ_k)^{-1} Z_k'M\varepsilon = O_p(T^{-1/2})$ uniformly on $K(C)$. The second and third lines are $O_p\left(T^{-1/2}\norm{\delta_h}^{-1}\right)$ from  $Z_{\Delta}'M\varepsilon \, \text{sgn}(k_0-k) = \vert k_0-k\vert^{1/2}O_p(1)= O_p\left(\norm{\delta_h}^{-1}\right)$. Hence,
\begin{equation*}
    H_T(k) = 2\lambda_h d_0'\Omega_{k_0}Z_{\Delta}'\varepsilon \,\text{sgn}(k_0-k) + o_p(1).
\end{equation*}
Combine this with (\ref{eq:G_limitdist}), we obtain the expression in Lemma \ref{lem:Qdiff_limitdist}.
\end{proof}

\bigskip
%%%%%%%%%%%%%%%%%%%%%%%%%%%%%%%%%%%%%%%%%%%%
\noindent\textbf{Proof of Theorem \ref{Thm:infill_ar}}
\begin{proof}
Lemma B.1 from \citet{Jiang2017} is restated below, and used without proof. 
\begin{lemma}\label{Lem:jiang17_lemB1}
For the process $y_t$ defined in (\ref{eq:DGPinfill_ar}), the following equations hold when $T = 1/h \rightarrow \infty$ with a fixed $\rho_0 = k_0/T$, for any $\rho \in [0, 1]$,
\begin{enumerate}[(a)]
    \item $T^{-1}\sum_{t=1}^{[\rho T]} y_{t-1}\varepsilon_t \Rightarrow \sigma^2 \int_{0}^{\rho} \wtilde{J}_0(r) dB(r)$;
    \item $T^{-2}\sum_{t=1}^{[\rho T]} y_{t-1}^2 \Rightarrow \sigma^2 \int_{0}^{\rho} \left[\wtilde{J}_0(r)\right]^2 dr$;
    \item $\left[\wtilde{J}_0(\rho)\right]^2-\left[\wtilde{J}_0(0)\right]^2 = 2\int_0^{\rho}\wtilde{J}_0(r)dB(r)-2\int_0^{\rho} \left(\mu+\delta \bbone\{r > \rho_0\}\right)\left[\wtilde{J}_0(r)\right]^2dr + \rho$;
    \item $\left[\wtilde{J}_0(1)\right]^2-\left[\wtilde{J}_0(\rho)\right]^2 = 2\int_\rho^{1}\wtilde{J}_0(r)dB(r)-2\int_\rho^{1} \left(\mu+\delta \bbone\{r > \rho_0\}\right)\left[\wtilde{J}_0(r)\right]^2dr + (1-\rho)$;
\end{enumerate}
where $\wtilde{J}_0(r)$ for $r \in [0, 1]$ is a Gaussian process defined in (\ref{eq:Jtilde0}) and $B(\cdot)$ is a standard Brownian motion.
\end{lemma}

\smallskip

Define the $(T\times 2)$ matrix $Y(k) = [Y_1(k) \,\vdots\, Y_2(k)]$ with $Y_1(k) = (y_0,\ldots,y_{k-1},0\ldots,0)'$, $Y_2(k) = (0\ldots,0,y_{k},\ldots,y_{T-1})'$ and $Y = (y_1,\ldots,y_T)'$. Then the LS objective function can be expressed as $S(k)^2 = Y'MY$ where
\begin{equation*}
    M = I_T - Y_1(k)[Y_1(k)'Y_1(k)]^{-1}Y_1(k)' - Y_2(k)[Y_2(k)'Y_2(k)]^{-1}Y_2(k)',
\end{equation*}
where $I_T$ is a $(T\times T)$ identity matrix. The model (\ref{eq:DGPinfill_ar}) can be written as
\begin{equation*}
    y_t = \beta_1 y_{t-1}+ (\beta_2-\beta_1)\bbone\{t > k_0\}y_{t-1}+\varepsilon_t = \beta_1 y_{t-1} + \eta_t,
\end{equation*}
where $\eta_t := (\beta_2-\beta_1)\bbone\{t > k_0\}y_{t-1}+\varepsilon_t$. Let $Y_{-} = (y_0,\ldots,y_{T-1})'$ and $\eta = (\eta_1,\ldots,\eta_T)'$. Then we have $Y = Y_{-}\beta_1 + \eta$, and the LS objective function is
\begin{align*}
    S(k)^2 & = (Y_{-}\beta_1 + \eta)'M'M(Y_{-}\beta_1 + \eta)\\
    & = \eta'\eta - \eta' Y_1(k)[Y_1(k)'Y_1(k)]^{-1}Y_1(k)'\eta - \eta'Y_2(k)[Y_2(k)'Y_2(k)]^{-1}Y_2(k)'\eta
\end{align*}
because $M$ is a idempotent matrix and $MY_{-} = \bzero$. Note that
\begin{equation*}
    \eta'\eta = \sum_{t=1}^{k_0} \eta_t^2 + \sum_{t=k_0+1}^{T} \eta_t^2 = \sum_{t=1}^{k_0} \varepsilon_t^2 + \sum_{t=k_0+1}^{T} ((\beta_2-\beta_1)y_{t-1}+\varepsilon_t)^2
\end{equation*}
which holds regardless of the choice of $k$, and 
\begin{align*}
    \eta' Y_1(k)[Y_1(k)'Y_1(k)]^{-1}Y_1(k)'\eta & = \frac{\left(\sum_{t=1}^{k}y_{t-1}\eta_t\right)^2}{\sum_{t=1}^{k}y_{t-1}^2},\\
    \eta'Y_2(k)[Y_2(k)'Y_2(k)]^{-1}Y_2(k)'\eta & = \frac{\left(\sum_{t=k+1}^{T}y_{t-1}\eta_t\right)^2}{\sum_{t=k+1}^{T}y_{t-1}^2}.
\end{align*}
Therefore, the break point estimator is
\begin{gather}
\label{eq:objfn_infillar}
    \hat{\rho} = \argmax_{\rho \in (0,1)} \omega(\rho)^2\, \mathcal{V}(\rho), \\
\nonumber    
    \mathcal{V}(\rho) := \left[\frac{\left(\sum_{t=1}^{[\rho T]}y_{t-1}\eta_t\right)^2}{\sum_{t=1}^{[\rho T]}y_{t-1}^2} + \frac{\left(\sum_{t=[\rho T]+1}^{T}y_{t-1}\eta_t\right)^2}{\sum_{t=[\rho T]+1}^{T}y_{t-1}^2}\right].
\end{gather}
When $\rho \leq \rho_0$, the terms in the numerator and denominator of $\mathcal{V}(\rho)$ weakly converges as follows.
\begin{equation*}
    T^{-1}\sum_{t=1}^{[\rho T]} y_{t-1}\eta_t = T^{-1}\sum_{t=1}^{[\rho T]} y_{t-1}\varepsilon_t \Rightarrow \sigma^2 \int_{0}^{\rho} \wtilde{J}_0(r)dB(r).
\end{equation*}
From Lemma \ref{Lem:jiang17_lemB1}, 
\begin{align*}
    T^{-1}\sum_{t=[\rho T]+1}^{T} y_{t-1}\eta_t & = T^{-1}\left[\sum_{t=[\rho T]+1}^{[\rho_0 T]} y_{t-1}\eta_t +\sum_{t=[\rho_0 T]+1}^{T} y_{t-1}\eta_t \right] \\
    & = T^{-1}\sum_{t=[\rho T]+1}^{T} y_{t-1}\varepsilon_t + T(\beta_2-\beta_1)T^{-2}\sum_{t=[\rho_0 T]+1}^{T} y_{t-1}^2 \\
    & \Rightarrow \sigma^2 \int_{\rho}^{1}\wtilde{J}_0(r)dB(r) -\delta \sigma^2 \int_{\rho_0}^{1} \left[\wtilde{J}_0(r)\right]^2 dr,
\end{align*}
\begin{equation*}
    T^{-2}\sum_{t=1}^{[\rho T]} y_{t-1}^2 \Rightarrow \sigma^2 \int_0^\rho \left[\wtilde{J}_0(r)\right]^2 dr, \;\; \text{ and }\;\; T^{-2}\sum_{t=[\rho T]+1}^{T} y_{t-1}^2 \Rightarrow \sigma^2 \int_\rho^1 \left[\wtilde{J}_0(r)\right]^2 dr.
\end{equation*}
Then the LS objective function $\mathcal{V}(\rho)$ in (\ref{eq:objfn_infillar}) weakly converges to
\begin{equation*}
    \mathcal{V}(\rho) \Rightarrow \sigma^2 \left[\frac{\left(\int_{0}^{\rho}\wtilde{J}_0(r)dB(r)\right)^2}{\int_0^\rho \left[\wtilde{J}_0(r)\right]^2 dr} + \frac{\left(\int_{\rho}^{1}\wtilde{J}_0(r)dB(r)-\delta \int_{\rho_0}^{1} \left[\wtilde{J}_0(r)\right]^2 dr \right)^2}{\int_\rho^1 \left[\wtilde{J}_0(r)\right]^2 dr}\right].
\end{equation*}
Lemma \ref{Lem:jiang17_lemB1} (c) and (d) implies that each term is rearranged as follows.
\begin{align*}
    \frac{\left(\int_{0}^{\rho}\wtilde{J}_0(r)dB(r)\right)^2}{\int_0^\rho \left[\wtilde{J}_0(r)\right]^2 dr} & = \frac{\left(\left[\wtilde{J}_0(\rho)\right]^2-\left[\wtilde{J}_0(0)\right]^2-\rho+2\mu\int_0^\rho \left[\wtilde{J}_0(r)\right]^2 dr \right)^2}{4\int_0^\rho \left[\wtilde{J}_0(r)\right]^2 dr} \\
    & = \frac{\left(\left[\wtilde{J}_0(\rho)\right]^2-\left[\wtilde{J}_0(0)\right]^2-\rho \right)^2}{4\int_0^\rho \left[\wtilde{J}_0(r)\right]^2 dr} + \mu^2 \int_0^\rho \left[\wtilde{J}_0(r)\right]^2 dr \\
    & \hspace{1cm} + \mu \left(\left[\wtilde{J}_0(\rho)\right]^2-\left[\wtilde{J}_0(0)\right]^2-\rho \right)
\end{align*}
\begin{multline*}
    \frac{\left(\int_{\rho}^{1} \wtilde{J}_0(r)dB(r)-\delta \int_{\rho_0}^{1} \left[\wtilde{J}_0(r)\right]^2 dr \right)^2}{\int_\rho^1 \left[\wtilde{J}_0(r)\right]^2 dr} = \frac{\left(\left[\wtilde{J}_0(1)\right]^2-\left[\wtilde{J}_0(\rho)\right]^2-(1-\rho) \right)^2}{4\int_\rho^1 \left[\wtilde{J}_0(r)\right]^2 dr} \\
    + \mu^2 \int_\rho^1 \left[\wtilde{J}_0(r)\right]^2 dr+ \mu\left(\left[\wtilde{J}_0(1)\right]^2-\left[\wtilde{J}_0(\rho)\right]^2-(1-\rho) \right).
\end{multline*}
As a result, the objective function of the break point estimator in (\ref{eq:objfn_infillar}) weakly converges to
\begin{align*}
   \frac{\omega(\rho)^2\mathcal{V}(\rho)}{\sigma^2} & \Rightarrow \frac{\omega(\rho)^2\left(\left[\wtilde{J}_0(\rho)\right]^2-\left[\wtilde{J}_0(0)\right]^2-\rho \right)^2}{4\int_0^\rho \left[\wtilde{J}_0(r)\right]^2 dr} + \frac{\omega(\rho)^2\left(\left[\wtilde{J}_0(1)\right]^2-\left[\wtilde{J}_0(\rho)\right]^2-(1-\rho) \right)^2}{4\int_\rho^1 \left[\wtilde{J}_0(r)\right]^2 dr} \\
    & \hspace{0.7cm} + \mu^2 \omega(\rho)^2 \int_0^1 \left[\wtilde{J}_0(r)\right]^2 dr + \mu\, \omega(\rho)^2\left(\left[\wtilde{J}_0(1)\right]^2-\left[\wtilde{J}_0(0)\right]^2-1\right).
\end{align*}
Following the same procedure above, $\omega(\rho)^2\mathcal{V}(\rho)/\sigma^2$ has the same limit when $\rho > \rho_0$. Therefore, by deleting the terms which are independent of the choice of $\rho$, the in-fill asymptotic distribution of $\hat{\rho}$ in (\ref{eq:objfn_infillar}) is
\begin{align*}
    \hat{\rho} & = \argmax_{\rho \in (0,1)} \omega(\rho)^2 \mathcal{V}(\rho)\\
    & \Rightarrow \argmax_{\rho \in (0,1)} \frac{\omega(\rho)^2\left(\left[\wtilde{J}_0(\rho)\right]^2-\left[\wtilde{J}_0(0)\right]^2-\rho \right)^2}{\int_0^\rho \left[\wtilde{J}_0(r)\right]^2 dr} + \frac{\omega(\rho)^2\left(\left[\wtilde{J}_0(1)\right]^2-\left[\wtilde{J}_0(\rho)\right]^2-(1-\rho) \right)^2}{\int_\rho^1 \left[\wtilde{J}_0(r)\right]^2 dr},
\end{align*}
which is identical to the distribution in Theorem \ref{Thm:infill_ar}.
\end{proof}

%%%%%%%%%%%%%%%%%%%%%%%%%%%%%%%%%%%%%%%%%%%%%%%
\section*{Appendix C}

\setcounter{table}{0}
\renewcommand{\thetable}{A.\arabic{table}}

\setcounter{figure}{0}
\renewcommand{\thefigure}{A.\arabic{figure}}

\begin{table}[H]
\caption{In-fill asymptotic RMSE, bias, and the standard error of the new  estimator and the LS estimator of the break point under model (\ref{eq:mc_stationarydgp}) with parameter values $(\rho_0,d_0)$ and $T = 100$. The number of replications is 5,000.}
\begin{center}  
\begin{tabular}{lL{0.8cm}C{1.5cm}C{1.5cm}C{1.5cm}C{1.5cm}C{1.5cm}C{1.5cm}} 
\hline
\hline
& & \multicolumn{2}{c}{RMSE} & \multicolumn{2}{c}{Bias}& \multicolumn{2}{c}{Standard error}\Tstrut\\
\cline{3-8}
 $\rho_0$ & $d_0$ & NEW & LS & NEW & LS& NEW & LS\Tstrut\\
\hline
 \multirow{3}{*}{0.15} & 1 & 0.4112 & 0.4284 & 0.3270 & 0.3174 & 0.2493 & 0.2876\Tstrut\\
  & 2 & 0.3763 & 0.3898 & 0.2895 & 0.2729 & 0.2403 & 0.2784\\
 & 4 & 0.3159 & 0.3071 & 0.2285 & 0.1807 & 0.2181 & 0.2483\\
 \hline
 \multirow{3}{*}{0.30} & 1 & 0.3004 & 0.3225 & 0.1865 & 0.1741 & 0.2355 & 0.2715\Tstrut\\
  & 2 & 0.2595 & 0.2959 & 0.1527 & 0.1471 & 0.2097 & 0.2567\\
 & 4 & 0.1930 & 0.2277 & 0.1059 & 0.0890 & 0.1614 & 0.2096\\
 \hline
 \multirow{3}{*}{0.50} & 1 & 0.2226 & 0.2635 & 0.0086 & 0.0081 & 0.2224 & 0.2634\Tstrut\\
 & 2 & 0.1894 & 0.2425 & 0.0081 & 0.0140 & 0.1893 & 0.2421\\
 & 4 & 0.1322 & 0.1887 & 0.0094 & 0.0122 & 0.1319 & 0.1883\\
 \hline
 \multirow{3}{*}{0.70} & 1 & 0.2796 & 0.3116 & -0.1560 & -0.1533 & 0.2320 & 0.2713\Tstrut\\
 & 2 & 0.2437 & 0.2837 & -0.1282 & -0.1240 & 0.2072 & 0.2551\\
 & 4 & 0.1840 & 0.2188 & -0.0890 & -0.0676 & 0.1611 & 0.2081\\
 \hline
 \multirow{3}{*}{0.85} & 1 & 0.3986 & 0.4116 & -0.3114 & -0.2981 & 0.2489 & 0.2838\Tstrut\\
 & 2 & 0.3616 & 0.3767 & -0.2717 & -0.2518 & 0.2386 & 0.2802\\
 & 4 & 0.2982 & 0.3003 & -0.2083 & -0.1651 & 0.2134 & 0.2508\\ 
 \hline
\hline
\end{tabular}
\end{center}
\label{table:rmse_infill_stationary}
\end{table}

\begin{table}[H]
\caption{In-fill asymptotic RMSE, bias, and the standard error of the new estimator and the LS estimator of the break point under the AR(1) model (\ref{eq:DGPinfill_ar}) with parameter values $(\beta_1,\beta_1,\rho_0)$ and $T = 200$. The number of replications is 5,000.}
\begin{center}  
\begin{tabular}{lllC{1.5cm}C{1.5cm}C{1.5cm}C{1.5cm}C{1.5cm}C{1.5cm}} 
\hline
\hline
& & & \multicolumn{2}{c}{RMSE} & \multicolumn{2}{c}{Bias} & \multicolumn{2}{c}{Standard error}\Tstrut\\
\cline{4-9}
$\beta_1$ & $\beta_2$ & $\rho_0$ & NEW & LS & NEW & LS & NEW & LS\Tstrut\\
\hline
\multirow{3}{*}{0.5} & \multirow{3}{*}{0.38} & 0.3 & 0.1909 & 0.2046 & 0.1888 & 0.0145 & 0.0290 & 0.2041\Tstrut\\
 &  & 0.5 & 0.0270 & 0.2368 & -0.0053 & -0.0751 & 0.0264 & 0.2246\\
 &  & 0.7 & 0.1995 & 0.3164 & -0.1970 & -0.1878 & 0.0316 & 0.2546\\
 \hline
\multirow{3}{*}{0.995} & \multirow{3}{*}{0.97} & 0.3 & 0.1284 & 0.1737 & 0.0529 & 0.0392 & 0.1170 & 0.1692\Tstrut\\
 &  & 0.5 & 0.0934 & 0.1749 & -0.0145 & -0.0260 & 0.0923 & 0.1730\\
 &  & 0.7 & 0.1759 & 0.2484 & -0.0934 & -0.1128 & 0.1490 & 0.2213\\
 \hline
\hline
\end{tabular}
\end{center}
\label{table:rmse_ar1_infill}
\end{table}

\begin{figure}[H]
    \centering
    \begin{subfigure}{.48\textwidth}
  \centering
    \includegraphics[width=\linewidth]{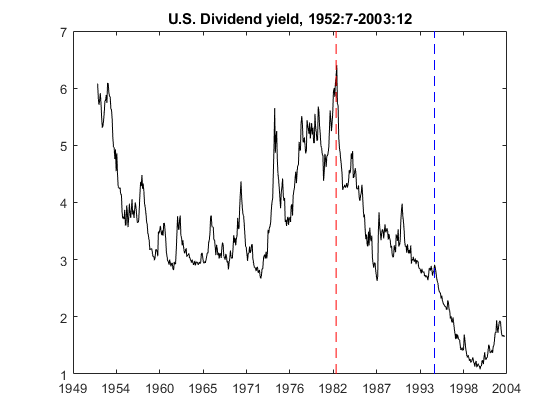}
\end{subfigure}
\begin{subfigure}{.48\textwidth}
  \centering
    \includegraphics[width=\linewidth]{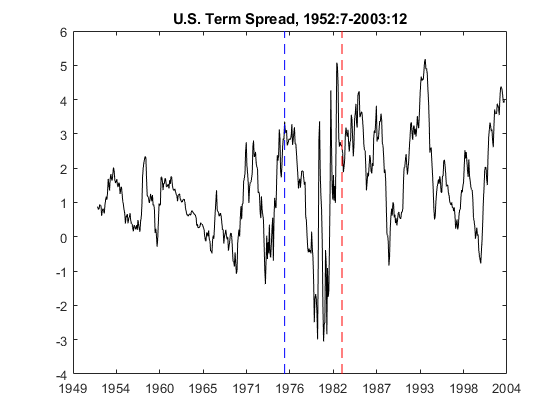}
\end{subfigure}
    \caption{U.S. dividend yield (left) and term spread (right), 1952:7-2003:12. Red and blue dotted lines are the new and LS break date estimates from the univariate model, respectively.}
\label{fig:us_tbill_divyield}
\end{figure}

\end{appendix}

\end{document}